\documentclass[letterpaper,12pt]{article}   
\usepackage{osajnl2} 

\usepackage{amsmath}
\usepackage{amsthm,amsbsy}

\usepackage{amssymb}
\usepackage{graphicx}
\usepackage{graphicx}
\usepackage{dcolumn}
\usepackage{bm}
\usepackage{amsfonts}
\usepackage{latexsym}
\usepackage{pdfsync}
    \usepackage{fullpage} 
\usepackage{colordvi}
\usepackage{color}
\usepackage{booktabs}
\usepackage{graphicx}
\usepackage{subfigure}
\usepackage{stmaryrd}
\usepackage{cite}
\input xy
\xyoption{all}

\newcommand{\commentout}[1]{}
\newcommand{\nwc}{\newcommand}
\nwc{\nn}{\nonumber}

\nwc{\nwt}{\newtheorem}
\nwt{cor}{Corollary}
\nwt{proposition}{Proposition}
\nwt{corollary}{Corollary}
\nwt{theorem}{Theorem}
\nwt{summary}{Summary}
\nwt{lemma}{Lemma}
\nwt{definition}{Defintion}
\nwt{remark}{Remark}

\nwc{\FF}{\mathcal{F}}
\nwc{\PP}{\mathcal{P}}
\nwc{\xx}{\mathbf{x}}
\nwc{\CC}{\mathbb{C}}
\nwc{\ZZ}{\mathbb{Z}}
\nwc{\RR}{\mathbb{R}}
\nwc{\bk}{\mathbf{k}}
\nwc{\bz}{\mathbf{z}}
\nwc{\bt}{\mathbf{t}}
\nwc{\bom}{\boldsymbol\omega}
\nwc{\bn}{\mathbf{n}}
\nwc{\bN}{\mathbf{N}}
\nwc{\PO}{\mathcal{P}_o}
\nwc{\PF}{\mathcal{P}_f}
\nwc{\QO}{\mathcal{Q}_o}
\nwc{\QF}{\mathcal{Q}_f}
\nwc{\PT}{\mathcal{T}}
\nwc{\real}{\text{re}}
\nwc{\imag}{\text{im}}
\nwc{\ep}{\epsilon}
\nwc{\vep}{\varepsilon}
\nwc{\tvep}{\tilde{\vep}}

\nwc{\mf}{\mathbf}
\nwc{\mb}{\mathbf}
\nwc{\ml}{\mathcal}
\nwc{\bj}{{\mb j}}
\nwc{\bA}{{\mb \Phi}}
\nwc{\IA}{\mathbb{A}} 
\nwc{\bi}{\mathbf i}
\nwc{\bo}{\mathbf o}
\nwc{\IS}{\mathbb{S}}
\nwc{\IC}{\mathbb{C}} 
\nwc{\ID}{\mathbb{D}} 
\nwc{\IM}{\mathbb{M}} 
\nwc{\IP}{\mathbb{P}} 
\nwc{\bI}{\mathbf{I}} 
\nwc{\IE}{\mathbb{E}} 
\nwc{\IF}{\mathbb{F}} 
\nwc{\IG}{\mathbb{G}} 
\nwc{\IN}{\mathbb{N}} 
\nwc{\IQ}{\mathbb{Q}} 
\nwc{\IR}{\mathbb{R}} 
\nwc{\IT}{\mathbb{T}} 
\nwc{\IZ}{\mathbb{Z}} 
\nwc{\IV}{\mathbb{V}}
\nwc{\IX}{\mathbb{X}}
\nwc{\IY}{\mathbb{Y}}

\nwc{\cE}{{\ml E}}
\nwc{\cP}{{\ml P}}
\nwc{\cQ}{{\ml Q}}
\nwc{\cL}{{\ml L}}
\nwc{\cX}{{\ml X}}
\nwc{\cW}{{\ml W}}
\nwc{\cZ}{{\ml Z}}
\nwc{\cR}{{\ml R}}
\nwc{\cV}{{\ml V}}
\nwc{\cT}{{\ml T}}
\nwc{\crV}{{\ml L}_{(\delta,\rho)}}
\nwc{\cC}{{\ml C}}
\nwc{\cA}{{\ml A}}
\nwc{\cK}{{\ml K}}
\nwc{\cB}{{\ml B}}
\nwc{\cD}{{\ml D}}
\nwc{\cF}{{\ml F}}
\nwc{\cS}{{\ml S}}
\nwc{\cM}{{\ml M}}
\nwc{\cG}{{\ml G}}
\nwc{\cH}{{\ml H}}

\nwc{\bT}{{\mb T}}
\nwc{\bM}{{\mb M}}
\nwc{\cbz}{\overline{\cB}_z}
\nwc{\supp}{{\hbox{\rm supp}}}
\nwc{\fR}{\mathfrak{R}}
\nwc{\bY}{\mathbf Y}

\nwc{\pft}{\cF^{-1}_2}
\nwc{\bU}{{\mb U}}
\nwc{\bPhi}{{\mb \Phi}}
\nwc{\bPsi}{{\mb \Psi}}
\nwc{\im}{{\rm i}}

\nwc{\bw}{{\mathbf w}}
\nwc{\mbm}{{\mathbf m}}
\nwc{\lbr}{\textlbrackdbl}
\nwc{\rbr}{\textrbrackdbl}
\nwc{\vzero}{{\mathbf 0}}
\nwc{\cN}{{\mathcal N}}
\nwc{\rbra}{\textrbrackdbl}
\nwc{\lbra}{\textlbrackdbl}
\nwc{\conv}{\hbox{conv}}
\nwc{\rank}{\hbox{rank}}

\nwc{\beq}{\begin{eqnarray}}
\nwc{\beqn}{\begin{eqnarray*}}
\nwc{\eeqn}{\end{eqnarray*}}
\nwc{\eeq}{\end{eqnarray}}

\begin{document}

\title{Phase Retrieval with Random Phase Illumination}


\author{Albert Fannjiang$^{1,*}$ and Wenjing Liao$^{1}$}
\address{$^1$ Department of Mathematics, University of California, Davis, CA 95616}
\address{$^*$Corresponding author: fannjiang@math.ucdavis.edu}

\begin{abstract}
This paper presents a detailed, numerical study on the performance of the standard phasing algorithms  with random phase illumination (RPI). Phasing with high resolution RPI and the oversampling ratio $\sigma=4$ determines
a unique phasing solution up to a global phase factor.
Under this condition, the standard phasing algorithms converge rapidly to the true solution without stagnation. Excellent
approximation is achieved after a small number of iterations,
  not just with high resolution but also low resolution RPI
  in the presence of  additive as well multiplicative noises.  It
  is shown that RPI with $\sigma=2$ is sufficient 
for phasing complex-valued images under a sector condition
and $\sigma=1$ for phasing nonnegative images. 
The Error Reduction
algorithm with RPI is proved to converge to  the true solution under  proper conditions. 

\end{abstract}

\ocis{100.5070, 170.1630, 340.7430.}

\maketitle 

\section{Introduction}
Fourier phase retrieval is the problem of reconstructing
an unknown image from  its Fourier magnitude data.
Phase retrieval is fundamental in many applications such as X-ray crystallography \cite{Millane}, astronomy \cite{FienupAstronomy}, coherent light microscopy \cite{MiaoNature}, quantum state tomography and remote sensing. 

 Due to the absence of the phase information, phase retrieval does not have a unique solution. Phase retrieval 
 literature has long settled with the notion
 of uniqueness modulo the trivial ambiguities of  spatial shift, conjugate inversion and
 global phase  \cite{Hayes} \cite{HayesReducible} and 
 focused on  circumventing the stagnation problem associated with the standard phasing algorithms. The numerical  stagnation problem 
 is often attributed to the nonconvex constraint 
 imposed by the Fourier magnitude data  \cite{Stagnation} \cite{FienupComplex}  \cite{Convex} \cite{Marchesini}. 

In this paper, we explore a phasing method based
on random phase modulator  which randomly
modifies the phases of the original image by
a  mask.
As proved  in \cite{UniqueRI} phasing with
 random (phase or amplitude)  illumination often
 leads to a unique solution up to a global phase factor
 (here dubbed {\em absolute uniqueness}).  In what follows we
 show  that phasing with random phase illumination (RPI) also 
leads to superior numerical performances, including rapid convergence, much reduced data and noise stability 
 of the standard algorithms.
 We show that under  proper conditions the Error-Reduction
 (ER) 
 algorithm with RPI converges to the true solution (Theorem \ref{FixedPoint1}).

Consider the discrete version of the phase retrieval problem: Let $\bn = (n_1,\ldots,n_d) \in \ZZ^d$ and $\bz = (z_1,\ldots,z_d) \in \CC^d$. Define the multi-index notation $\bz^\bn = z_1^{n_1}z_2^{n_2}\hdots z_d^{n_d}$. Let $\mathcal{C}(\cN )$ denote the set of finite complex-valued functions on $\ZZ^d$ vanishing outside
 \begin{equation*}
 \mathcal{N} = \{ \mathbf{0} \le \bn \le \bN\}, \quad \bN = (N_1,N_2,\ldots,N_d). 
 \end{equation*}
Here  $\mathbf{m} \le \bn$ if $m_j \le n_j, \forall j$. Denote $|\cN|  =\displaystyle \prod_{j=1}^d N_j$. 
 
The $z$-transform of a $d$ dimensional finite array   $f(\bn) \in \mathcal{C}(\cN )$ is given by
\begin{equation*}
F(\bz) =  \sum_{\bn} f(\bn) \bz^{-\bn}.
\end{equation*}
The Fourier transform can be obtained from the $z$-transform as
\begin{equation}
 F(e^{i 2\pi \bom}) = \sum_{\bn} f(\bn) e^{-2 \pi i \bom \cdot \bn}
\label{DFT}
\end{equation} 
for $\bom = (\omega_1,\omega_2 , \hdots , \omega_d),  0 \le \omega_j < 1$.

From the calculation 
 \beq
  |F(e^{i2\pi\bw})|^2&=& \sum_{\bn =-\bN}^{\bN}\sum_{\mbm+\bn\in \cN} f(\mbm+\bn)\overline{f(\mbm)}
   e^{-\im 2\pi \bn\cdot \bw}\nn
   \eeq
   we see that the Fourier magnitude measurement
   is equivalent to the standard discrete Fourier measurement of
   the correlation function 
          \beq
	 \label{aut}
	  \cC_f(\bn)=\sum_{\mbm\in \cN} f(\mbm+\bn)\overline{f(\mbm)}
	  \label{autof}
	  \eeq
if sampled at the lattice 
 \begin{equation}
\mathcal{L} = \Big\{\bom=(\omega_1,...,\omega_d)\ | \ \omega_j = 0,\frac{1}{2 N_j + 1},\frac{2}{2N_j + 1},...,\frac{2N_j}{2N_j + 1}\Big\}
\label{latice}
\end{equation}
which is $2^d$ times of the grid of the original image. 
The standard phasing problem  is to recover the array   $f(\bn)$ from its Fourier intensity measurement $Y(\bom)=|F(e^{i 2\pi \bom})|$ for $\bom \in \cL$ or smaller sampling sets. 

Clearly the correlation function $\cC_f$ and the Fourier magnitude data are invariant under   spatial translation 
$$f(\cdot) \rightarrow f(\cdot+\mathbf{t}) \text{ for some } \mathbf{t} \in \IZ^d,$$ 
conjugate inversion  
$$f(\cdot) \rightarrow \overline{f(\bN -\cdot)}$$ 
and constant global phase change
$$f(\cdot) \rightarrow e^{i\theta}f(\cdot).$$ 
These trivial associates all share the same global geometric information as the original object. 
 The classical results of uniqueness given in \cite{Hayes} \cite{HayesReducible} \cite{ComplexUniqueness} say that for almost
all objects in dimension two or higher the trivial associates
are the only ambiguities there are with phase retrieval. 

On the other hand, by dimension counting Miao {\em et al.}  \cite{MSC} have argued that overall $2$ times oversampling, independent of the dimension $d$,  uniquely determines
a unique  phasing solution up to spatial shift,
conjugate inversion and global phase factor. To measure
the degree of oversampling  we use  the oversampling ratio (OR)
\beqn
\sigma &=& \frac{\text{Fourier magnitude data number}}{\text{unknown-valued image pixel number}}
\eeqn
introduced in \cite{MSC}. As we demonstrate
below, Miao {\em et al.}'s conjecture can be
realized by using RPI, but
not uniform illumination.

As shown in \cite{UniqueRI} random illumination (RI) can
help remove the phasing ambiguities of spatial shift and
conjugate inversion. 
An illumination amounts to replacing the original image $f(\bn)$ by 
\begin{equation*}
g(\bn) = \lambda(\bn)f(\bn),
\end{equation*}
where $\lambda(\bn)$ is a known array representing the incident wave. In the case of  uniform illumination, $\lambda(\bn) = 1$. In the case of random phase illumination (RPI) \cite{RandomPhaseModulator}, 
\begin{equation}
\lambda(\bn) = e^{i\phi(\bn)}
\label{randomphase}
\end{equation}
where $\phi(\bn)$ are  random variables on $[0,2\pi]$, and in the case of random amplitude illumination \cite{RandomAmplitudeModulator,Candes1}, $\lambda(\bn)$ is an array of real random variables. RI can be facilitated by random phase/amplitude modulators or random masks.

The paper is organized as follows. We review the absolute uniqueness of phasing with RPI in Section \ref{sec2} and standard phasing algorithms in Section \ref{sec3} where
convergence of the Error Reduction (ER) iteration to the true solution is presented (Theorem \ref{FixedPoint1}). We present the numerical phasing results in Section \ref{sec4}. We conclude in Section \ref{sec5}.

For the rest of the paper  we use the following notation: the vector space $\mathcal{C}(\cN)$ is endowed  with the inner product $<f,g> = \sum_{\bn} \overline{f(\bn)}g(\bn)$. For a complex number $z$, $\Re(z)$ and $\Im(z)$ denote the real and imaginary part of $z$. $\measuredangle{z} \in [0,2\pi)$ denotes the phase (angle)  of $z$. When $z=0$, $\measuredangle{z}$ is taken to be $0$ unless specified otherwise. $[\alpha] =  \alpha(\text{mod}(2\pi))$. 

\commentout{Let $\Delta(\alpha,\beta)$ be the sector in $\CC$ containing all complex numbers whose phases are between $\alpha$ and $\beta$. The angular area of $\Delta(\alpha,\beta)$ is $\delta(\alpha,\beta)  = \left\{ \begin{array}{ll} \beta-\alpha &\text{ if } \alpha \le \beta \\
\beta-\alpha+2\pi  & \text{ else } \end{array}\right. .$ $\alpha \prec \theta \prec \beta$ denotes that $\theta$ lies between $\alpha \in [0,2\pi)$ and $\beta \in [0,2\pi)$ such that
$$\left\{ \begin{array}{ll} \alpha \le \theta \le \beta &\text{ if } \alpha \le \beta \\
\alpha \le \theta < 2\pi \text{ or } 0 \le \theta \le \beta & \text{ if } \alpha > \beta\end{array}\right. .$$}

\section{Uniqueness}
\label{sec2}
In the following we recall several uniqueness results from \cite{UniqueRI} relevant to phasing with RPI. 

First we define the {\em rank} of an array.
The support of the array consists of the set of
nonzero pixels.  The rank of the array 
 is the dimension of its support's convex hull in $\IR^d$.

\begin{theorem}Let $\lambda(\bn)$ be independent, continuous random variables on $\mathbb{S}^1$. 
Let $f(\bn)\in \mathcal{C}(\cN )$ be a real-valued array 
of rank $\geq$   2.  Then, with probability one, $f$ is determined absolutely uniquely up to $\pm$ sign by the Fourier magnitude measurement on 
$\cL$. 
\label{UniqueReal}
\end{theorem}
 
A more general, practical constraint is to restrict the image values within a certain sector of the complex plane. For instance, when the incident X-rays are low energy photons(soft X-rays), the electron density is complex. The real part represents the effective number of electrons that diffract the X-rays in phase and is usually positive but
becomes negative only when the energy of the incident X-rays is near an absorption edge. The imaginary part represents the absorption of the X-rays by the specimen and thus is always positive.

\begin{theorem}Let $\lambda(\bn)$ be independent, continuous random variables on $\mathbb{S}^1$. 
Let  $f$ be a complex-valued array of rank $\geq$ 2 such that $\measuredangle f(\bn)  \in [\alpha,\beta], \forall \bn$. Let $S$ denote the sparsity of the image and let $\llfloor S/2\rrfloor$ be the greatest integer less than or equal to $S/2$.

Suppose that the phases $\phi(\bn)$ of RPI are independent,  uniform random variables on $[0,2\pi]$. Then with probability no less than $1-|\cN|  (\beta-\alpha)^{\llfloor S/2\rrfloor}(2\pi)^{-\llfloor S/2\rrfloor}$, the object $f$ is uniquely determined, up to a global phase, by the Fourier magnitude measurement on $\mathcal{L}$. 
\commentout{
(ii) When $\beta - \alpha < \pi$, consider RAI with i.i.d. random variables $\lambda(\bn) \in \RR$ that are equally likely negative or positive, i.e. $\mathbb{P}\{ \lambda(\bn)>0\} = \mathbb{P} \{ \lambda(\bn) < 0\} = 1/2, \ \forall \bn$. Then with probability no less than $1-2^{-\llfloor(S-1)/2\rrfloor}|\cN| $, the image $f$ is uniquely determined, up to a global phase,  by the Fourier magnitude measurements on $\mathcal{L}$.
}

 The global phase is uniquely determined  if the angular sector $[\alpha,\beta]$ is tight
 in the sense that no proper subset of $[a,b]$ contains
 all the phases of the object.
\label{UniqueComplexPositive}
\end{theorem}

\commentout{
\begin{theorem} Suppose that the object support
has rank $\geq$   2. 
Suppose that $K$ pixels of the complex-valued object $f$ satisfy the magnitude constraint $0<a\leq |f(\bn)|\leq b$ and that $\{\lambda(\bn)\}$ are i.i.d. continuous r.v.s on real algebraic varieties in $\IC$ with $\IP\{|\lambda(\bn)/\lambda(\bn')|> b/a\,\,\hbox{or}\,\,  |\lambda(\bn)/\lambda(\bn')|<a/b\}=1-p>0$ for $\bn\neq\bn'$. Then the object $f$ is determined uniquely, up
to a global phase, by the Fourier magnitude data on $\cM$, with probability at least $1-|\cN| p^{[K/2]}$. 
\label{thm6}
\end{theorem}

\begin{theorem}\label{cor1} Suppose the
 object support has rank at least 2.
Suppose either 
 of the following cases holds:

 (i) The phases of the object $\{f(\bn)\}$ at  two points, where $f$ does not vanish,  belong to
 a known countable subset of $[0,2\pi]$ and the
 angles of $\{\lamb(\bn)\}$ are independent continuous r.v.s   on $[0,2\pi]$.

 (ii) The amplitudes of the object $\{f(\bn)\}$ at two points, where $f$ does not vanish, 
 belong to a known measure zero  subset of $\IR$  and the magnitudes of  $\{\lamb(\bn)\}$ are independent continuous r.v.s  on $(0,\infty)$. 
 
Then $f$ is determined uniquely,   up to a global phase,   by the
  Fourier magnitude measurement  on the lattice $\cL$
     with probability one.
 \end{theorem}
}

For general complex-valued images without any constraint, we use two independent RPIs to collect data.

\begin{theorem}
 Let $\lambda_1(\bn)$ and $\lambda_2(\bn)$ be two independent arrays of continuous random variables on $\mathbb{S}^1$. Let $f(\bn) \in \mathcal{C}(\cN )$ be any complex-valued array of rank $\geq 2$. Then almost surely $f(\bn)$ is uniquely determined, up to a constant phase factor, by the Fourier magnitude measurement on $\mathcal{L}$ with two illuminations $\lambda_1$ and $\lambda_2$. If the second illumination $\lambda_2(\bn)$ is deterministic while $\lambda_1(\bn)$ is random as above, then the same conclusion holds.
\label{UniqueComplex}
\end{theorem}

\section{Phasing Algorithms}
\label{sec3}
 To  find the true object satisfying both the object-domain constraint,  which is usually convex, and the frequency-domain constraint, 
 which is non-convex,  most phasing algorithms are  based on the idea of alternating  projections 
 from the convexity literature \cite{Convex}.
 
\subsection{Projections}

\begin{definition} Let $\mathcal{D}$ be a subset of $\cC(\cN)$, the orthogonal projection of $f \in \cC(\cN)$ on $\mathcal{D}$ is $ \displaystyle \text{argmin}_{g\in \mathcal{D}} \|g-f\|$
\end{definition}

If the minimizer is not unique, one of them is arbitrarily selected. When $\mathcal{D}$ is a closed convex subset of $\mathcal{C}(\cN)$, the minimizer is unique.
\begin{proposition}
Let $\mathcal{D}$ denote any closed convex subset of $\cC(\cN)$ and let $f$ be any element in $\cC(\cN)$. Then there exits a unique $h \in \mathcal{D}$ such that 
$$\text{inf}_{g\in \mathcal{D}} \|g-f\| = \|h-f\|.$$
\label{prop1}
\end{proposition}

Let $\Gamma$ be the set of functions satisfying the object-domain constraint, such as a known support  or positivity, and $\Omega$ be the set of functions satisfying the frequency-domain constraint imposed by
the known Fourier magnitude data.  A solution of phase retrieval is a function belonging to $\Gamma \cap \Omega$. Let $\PO$ and $\PF$ be the orthogonal projection on $\Gamma$ and $\Omega$ respectively. 

Let $\Lambda$ be the diagonal matrix with diagonal elements $\lambda(\bn)$, and set
$g = \Lambda f$. Let $\Phi$ be the discrete Fourier transform and set  
$
Y = |\Phi f| $. 

Given the Fourier intensity data  $Y$, we define the intensity fitting operator $\mathcal{T}$ as 
\begin{equation}
G'(\bom) = \PT\{G\}(\bom) = \left \{ \begin{array}{ll} Y(\bom)e^{i \measuredangle{G(\bom)}} & \text{ if } |G(\bom)| > 0 \\
Y(\bom) & \text{ if } |G(\bom)| = 0
\end{array}\right. .
\label{IntensityFitting}
\end{equation}
When $G(\bom) = 0$, $\measuredangle{G(\bom)}$ is not uniquely defined and  $\measuredangle{G(\bom)}$ is set $0$ in \eqref{IntensityFitting}.
In this case,
\begin{equation*}
\PF = \Lambda^{-1} \Phi^{-1}\PT \Phi \Lambda.
\end{equation*}
Indeed   $\measuredangle{G(\bom)}$ can be arbitrarily chosen at the zero set of  $G$, and we define
 \beq
 \mathcal{P}_f^{\theta} =\Lambda^{-1}\Phi^{-1} \PT^{\theta}\Phi\Lambda\label{5}
 \eeq
where  
\beq
\PT^{\theta}\{G\}(\bom) = \left \{ \begin{array}{ll} Y(\bom)e^{i \measuredangle{G(\bom)}} & \text{ if } |G(\bom)| > 0 \\
Y(\bom)e^{i \theta(\bom)} & \text{ if } |G(\bom)| = 0
\end{array}\right.\label{6}
\eeq

The object domain projection $\PO$ can take
a varied form depending on the problem. 

\begin{itemize}
\item When $\Gamma$ is the set of images with a given phase $\alpha$,
$$\PO \{h\}(\bn) = \mathcal{P}_\alpha\{h(\bn)\}= \max\{\Im(h(\bn)) \sin{\alpha} + \Re(h(\bn)) \cos{\alpha},0\} e^{i\alpha}.$$ 

\item When $\Gamma$ is the set of images with phases in $[\alpha,\beta]$ for $0 \le \alpha < \beta \le 2\pi$,
\begin{itemize}
\item if $\beta -\alpha \le \pi$, $\PO\{h\}(\bn) = \left \{ \begin{array}{ll} h(\bn) & \text{ if } \alpha \prec \measuredangle{h(\bn)} \prec \beta \\
\mathcal{P}_\beta \{h(\bn)\} & \text{ if } \beta \prec \measuredangle{h(\bn)} \prec [\beta+\pi/2] \\\Re
\big (\mathcal{P}_\alpha \{h(\bn)\}\big )& \text{ if } [\alpha - \pi/2] \prec \measuredangle{h(\bn)} \prec \beta\\ 
0 & \text{ else}
\end{array} \right. ,$

\item if $\beta - \alpha > \pi$, $\PO\{h\}(\bn) =  \left \{ \begin{array}{ll} 
h(\bn) & \text{ if } \alpha \prec \measuredangle{h(\bn)} \prec \beta \\
\mathcal{P}_\beta \{h(\bn)\} & \text{ if } \beta \prec \measuredangle{h(\bn)} \prec [(\alpha+\beta)/2+\pi] \\
\mathcal{P}_\alpha \{h(\bn)\} & \text{ if } [(\alpha+\beta)/2+\pi] \prec \measuredangle{h(\bn)} \prec \alpha
\end{array} \right. ,$
\end{itemize}
where $a \prec \theta \prec b$ means  $\theta$ is between $a$ and $b$ such that
$$\left\{ \begin{array}{ll} a \le \theta \le b &\text{ if } a \le b \\
a \le \theta < 2\pi \,\,\text{ or } \,\, 0 \le \theta \le b & \text{ if } a > b \end{array}\right. .$$

\item When $\Gamma$ is the set of real valued images,
$$\PO\{h\}(\bn) = \Re(h(\bn)).$$

\item When $\Gamma$ is the set of nonnegative real-valued images,
$$\PO\{h\}(\bn) = \max\{\Re(h(\bn)),0\}.$$

\item When $\Gamma$ is the set of complex valued images with nonnegative real and imaginary parts,
$$\Re(\PO\{h\}(\bn)) = \max(\Re(h(\bn)),0)$$
$$\Im(\PO\{h\}(\bn)) = \max(\Im(h(\bn)),0).$$

\item When $\Gamma$ is the set of images with support $S$,
$$\PO\{h\}(\bn) = \left \{ \begin{array}{ll} h(\bn) & \text{ if } \bn \in S \\ 0 & \text{ else}  \end{array}\right. .$$

\end{itemize}

Two error metrics $\vep_o$ and $\vep_f$ defined by
$$\vep_o(h) = \|\PO\{h\}-h\|,$$
$$\vep_f(h) = \|\PF\{h\}-h\|$$
play an important role of our studies. 
\commentout{
If $h$ satisfies the object domain constraint, i.e. $h \in \Gamma$,   $\vep_f$ is the effective error measuring how far $h$ is away from $\Omega$. On the contrary, if $h$ satisfies the Fourier intensity measurement, i.e. $h \in \Omega$, such as the HIO output, $\vep_o$ is the effective error measuring how far $h$ is from $\Gamma$.
}
\commentout{
The normalized errors are
\begin{equation}
\tvep_o(h) = ||\PO\{h\}-h||/||\PO\{h\}||,
\label{tvepo}
\end{equation}
\begin{equation}
\tvep_f(h) = ||\PF\{h\}-h||/||\PF\{h\}||.
\label{tvepf}
\end{equation}
}
When $\Phi\Lambda$ is unitary, as in the case of
RPI, 
$$\vep_f(h) = \|\PF\{h\}-h\| = \|\PT\Phi\Lambda h - \Phi \Lambda h\| = \|\ Y - |\Phi \Lambda h| \ \|.$$

\subsection{Oversampling}
The oversampling method has proven to be an effective, flexible  way of implementing various phasing algorithms  by converting
Fourier magnitude data more finely sampled than demanded by the original image grid into  zero padding which
then acts like  a support constraint of the original image
 \cite{Hayes, Miao2000, MiaoKS, MiaoNonperiodic}.   
 In this set-up, the oversampling ratio is given by
 \beqn
\sigma
&=& \frac{\text{image pixel number + zero-padding pixel number}}{\text{image pixel number}}.
\eeqn

\subsection{Error reduction (ER)}
ER algorithm \cite{Fienup}  is based on the Gerchberg-Saxton algorithm \cite{GS} and  is the most basic  phasing algorithm. ER is the plain version of  the alternated  projection method:
\begin{equation}
f_{k+1} = \PO \mathcal{P}_f f_k
\label{er}
\end{equation}
which can be conveniently represented by the
following diagram
\begin{figure}
\begin{center}
{
$$\xymatrix{
f_{k+1} /f_k\ar[rr]^{\Lambda} && g_k \ar[rr]^{\Phi} && G_k \ar[d]^{\PT} \\
f'_k \ar[u]^{\PO} && g'_k\ar[ll]^{\Lambda^{-1}} && G'_k \ar[ll]^{\Phi^{-1}}
}$$}
\caption{}
\label{fig1}
\end{center}
\end{figure}

ER enjoys the error-decreasing property
following the same argument in \cite{Fienup}.
\begin{proposition}
Let $\Gamma$ be a closed convex subset of $\cC(\cN)$. Let $\Phi$ and $\Lambda$ be unitary matrices. Then the array $\{f_k\}$ produced in \eqref{er} satisfies
\begin{equation}
\vep_f(f_{k+1}) \le \vep_f(f_k).
\end{equation}
The equality holds if and only if $f_{k+1} = f_k$.
\label{lemma1}
\end{proposition}
\begin{proof}
\beqn
\vep_f(f_k)  & = &\|f_k - f'_k\|\\
& \ge& \|f_{k+1} - f'_k\| \nonumber \\
                      & = & \|G_{k+1} - G'_{k}\| \\
                      & \ge&  \|G_{k+1} - G'_{k+1}\| \nonumber \\
                      & =&  \|f_{k+1} - f'_{k+1}\| \\
                      & =& \vep_f(f_{k+1}). \nonumber 
\eeqn
The equality holds only if $\|f_k - f'_k\| = \|f_{k+1} - f'_k\| $, where $f_{k+1} = \PO\{f'_k\}$. Since $\Gamma$ is a closed convex subset, $f_{k+1} = f_k$ according to Proposition \ref{prop1}.
\end{proof}


\begin{remark}
 Proposition \ref{lemma1} holds  for the $f_{k+1} = \PO\PF^{\theta} f_{k}$ with arbitrary $\theta(\bom)$. 
\end{remark}

Proposition \ref{lemma1} shows that the error $\vep_f(f_k)$ decreases strictly until it reaches a fixed point of $\PO\PF$, implying  that the ER iteration converges to a fixed point.
\begin{proposition}
Let $f_{k+1} = \PO\PF f_k.$ Let $\Gamma$ be a closed convex subset of $\cC(\cN)$ and $\Phi$,$\Lambda$ be unitary matrices. Then every convergent subsequence of $\{f_k\}$ converges to some $h$ such that 
\begin{enumerate}
\item if $\Phi \Lambda h(\bom) \neq 0 , \forall \bom\in \cL$, $h$ is a fixed point of $\PO\PF$.
\item if $\Phi \Lambda h(\bom) = 0 $ for some $\bom\in \cL$, $h$ is a fixed point of $\PO\PF^{\theta}$ for some $\theta$.
\end{enumerate}
\label{accumulation}
\end{proposition}
The proof of Proposition \ref{accumulation} is given in
the Appendix. 
The question is, Is a fixed point of ER necessarily a
phasing solution? With the uniform illumination, however, this is generally
not true \cite{GP}. When a fixed point fails to be
a phasing solution, it is called a trap and
can plague the reconstruction procedure (cf. Figure \ref{Cameraman}(a), \ref{Phantom}(a) and \ref{Phantom}(c)).  

Below, we answer this question in
the affirmative
under certain assumptions  for the case of RPI. 
The difficulty is
ER may converge to a  fixed point of $\PO\PF^{\theta}$ which fails to satisfy the Fourier magnitude data. In other words, the limiting point $h$ may not be a fixed point of $\PF^{\theta}$.

In the following main theoretical result of the paper, 
we prove that if $\PF^{\theta} h$ satisfies the zero-padding condition, then it must be the phasing solution.

\commentout{
If $\Phi\Lambda h$ vanishes somewhere, the orthogonal projection onto $\Omega$ is multivalued due to the ambiguity of $\measuredangle{0}$. $\PF$ is only one of the orthogonal projections. Theoretically, some convergent subsequence of $\{f_k\}$ might converge to a fixed point of $\PO\PF^{\theta}$ for some $\theta \neq \mathbf{0}$. Numerically, $\|f_{k+1} - f_k\| \rightarrow 0$, which implies that $f_k$ gets closer and closer to a fixed point of $\PO\PF$.}

 \begin{theorem}

Let $f(\bn)\in \mathcal{C}(\cN )$ be an array with $f(\mathbf{0}) \neq 0$ and of rank $\ge 2$. Let $\lambda(\bn)$ be $i.i.d.$ continuous random variables on $\mathbb{S}^1$. Let the Fourier magnitude be sampled on $\cL$.  Let
 $h$ be  a fixed point of $\PO\PF^{\theta}$ {such that $\PF^{\theta}h$ satisfies the zero-padding condition}.
 
  \begin{description}
\item[\rm (a)] If $f$ is real-valued, $h = \pm f$ with probability one,
\item[\rm (b)] If  $f$ satisfies the sector condition of Theorem \ref{UniqueComplexPositive}, then $h = e^{i \nu} f$, for some $\nu$,
and satisfies the same sector constraint with probability at least  $1- |\mathcal{N}| (\beta-\alpha)^{\llfloor S/2 \rrfloor} (2\pi)^{-\llfloor S/2\rrfloor}$.
\end{description}
\label{FixedPoint1}
\end{theorem}

\subsection{HIO}
The hybrid input-output (HIO) algorithm is a 
widely used, better-performing phasing method 
than ER's \cite{Fienup}. HIO differs from ER in how to update the image in the object domain
in order to avoid the trapping and stagnation.

Below we present a modified version of  Fienup's HIO 
which performs better than the original version.
We refer to Figure \ref{fig1} for  the notation.  
In HIO, the last step $\cP_o$ of ER iteration is replaced
by the following.

\begin{itemize}
\item 
When $\Gamma$ is the set of real-valued images,
\beq
\label{8}
\Re(f_{k+1}(\bn)) &=& \Re(f'_k(\bn))\\
\Im(f_{k+1}(\bn)) &= &\Im(f_{k}(\bn)) - \beta  \cdot \Im(f'_{k}(\bn)),
\eeq
If, in addition,  the nonnegativity constraint is assumed,
then 
\beq
\Re(f_{k+1}(\bn)) &=& \left\{ \begin{array}{ll}
\Re(f'_k(\bn)) & \text{if } \Re(f'_k(\bn)) \ge 0 \\
\Re(f_k(\bn)) - \beta \cdot \Re(f'_k(\bn))   & \text{if } \Re(f'_k(\bn)) < 0\end{array}, \right.\\
\Im(f_{k+1}(\bn)) &= &\Im(f_{k}(\bn)) - \beta  \cdot \Im(f'_{k}(\bn)).
\eeq

\item When $\Gamma$ is the set of complex-valued images with nonnegative real and imaginary parts,
\beq
\Re(f_{k+1}(\bn)) = \left\{ \begin{array}{ll}
\Re(f'_k(\bn)) & \text{if } \Re(f'_k(\bn)) \ge 0 \\
\Re(f_k(\bn)) - \beta \cdot \Re(f'_k(\bn))   & \text{if } \Re(f'_k(\bn)) < 0\end{array}. \right.\\
\Im(f_{k+1}(\bn)) = \left\{ \begin{array}{ll}
\Im(f'_k(\bn)) & \text{if } \Im(f'_k(\bn)) \ge 0 \\
\Im(f_k(\bn)) - \beta \cdot \Im(f'_k(\bn))   & \text{if } \Im(f'_k(\bn)) < 0\end{array}. \right.\label{13}
\eeq

\end{itemize}

\subsection{Algorithms with two illuminations}
Let $\lambda_1(\bn)$ and $\lambda_2(\bn)$ be two arrays representing two illuminating fields. Two sets of Fourier magnitude data
$Y_1 = | \Phi \Lambda_1 f|$
and
$Y_1 = | \Phi \Lambda_2 f|$ are collected, each
with an OR $\sigma$. 
Let $\mathcal{T}_1$ and $\mathcal{T}_2$ be  the intensity fitting operators corresponding to $Y_1$ and $Y_2$, respectively,   as in  \eqref{IntensityFitting}.  Thus the projections onto the set of images satisfying the Fourier magnitude data $Y_1$ and $Y_2$ are, respectively, 
$$\mathcal{P}_1= \Lambda_1^{-1} \Phi^{-1} {\mathcal{T}}_1 \Phi \Lambda_1$$
and 
$$\mathcal{P}_2 = \Lambda_2^{-1} \Phi^{-1} {\mathcal{T}}_2 \Phi \Lambda_2.$$

The corresponding ER algorithm with two sets of Fourier magnitude data $Y_1$ and $Y_2$ is given by 
\beq
\label{14}
f_{k+1} = \PO \mathcal{P}_2  \mathcal{P}_1 f_k.
\eeq
The corresponding HIO is obtained by replacing $\PO$ in (\ref{14}) by 
(\ref{8})-(\ref{13}).

\section{Numerical Simulations}
\label{sec4}
\begin{figure}[h]
\begin{center}
         \subfigure[]{
         \includegraphics[width = 6cm]{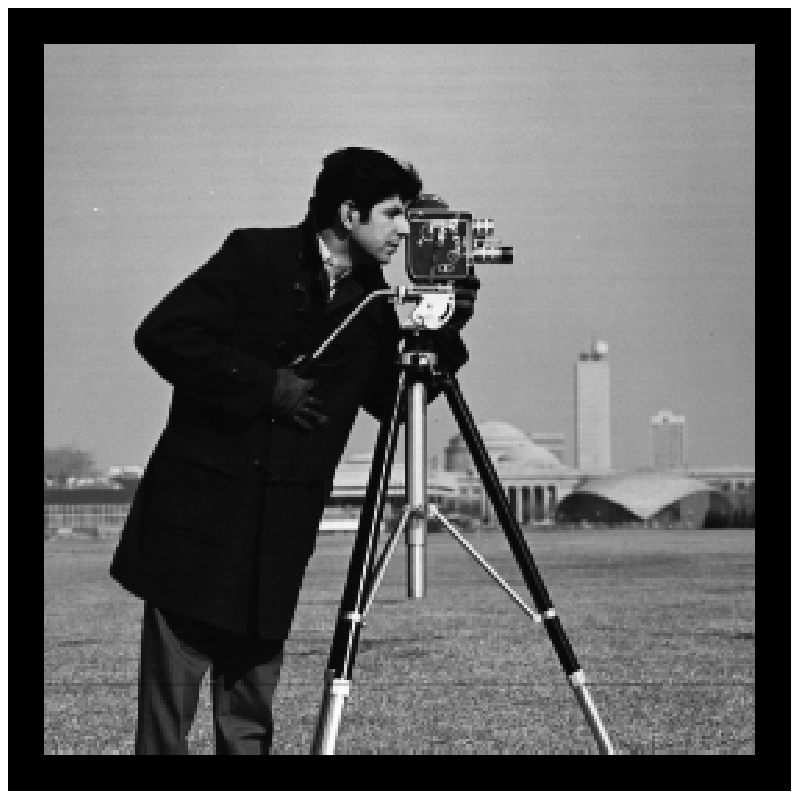}}
                 \subfigure[]{
         \includegraphics[width = 6cm]{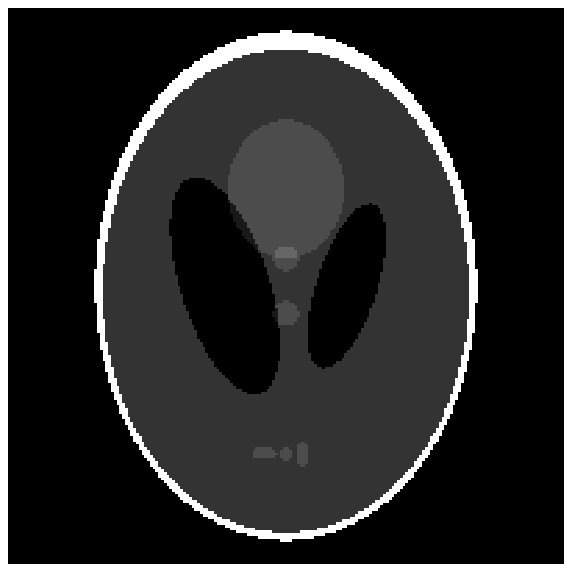}}
 \caption{Test images of loose support: (a) $269\times 269$ Cameraman (b) $200\times 200$ Phantom where
 the dark borders represent loose support.} 
 \end{center}
 \label{fig:true}
\end{figure}
In this section, we perform numerical phasing  from the Fourier intensity measurement with UI or RPI.

Our test images are the $256\times 256$ Cameraman
and  the $138 \times 184$ Phantom. 
We surround both images by dark (i.e. zero-valued) border  to create  images of loose support. Images
of loose support are typically more challenging to
reconstruct. 
For Cameraman the border is $13$ pixel wide in each
dimension and the resulting image has $269\times 269$ pixels in total. For  Phantom  the dark margin is such that the  resulting image has  $200 \times 200$ pixels.

For the oversampling ratio $\sigma$, we zero pad  the images to generate a $269\sqrt{\sigma} \times 269\sqrt{\sigma}$  Cameraman and  $200\sqrt{\sigma} \times 200\sqrt{\sigma}$ Phantom. We synthesize the Fourier magnitude data by  applying the FFT to the array.  

\commentout{ Our main point includes: once absolute uniqueness is enforced by the use of random phase illumination, both ER and HIO converge to the true solution quickly, even with the lack of convexity; Not only high resolution but also low resolution random phase illumination render a good recovery and strong stability to various type of noise, as we demonstrate that the reconstruction error increases almost linearly with respect to the noise level; With high resolution random phase illumination, real-valued nonnegative images are perfectly recovered when $\sigma = 1$, which means the Fourier intensity data are not oversampled and no support constraint is required. Phasing without support constraint is of great interest in physical sciences \cite{MabCS, Mab3D}. }

\subsection{Error, Residual and Noise}
Let $\hat{f}$ be the recovered image. The relative error is defined as  $$e(\hat{f}) = 
\left\{ 
\begin{array}{ll}
{\|f-\hat{f}\|}/{\|f\|} & \text{ if absolute uniqueness holds} \\
{\displaystyle \min_{\nu\in [0,2\pi)}\|f-e^{i\nu}\hat{f}\|}/{\|f\|} & \text{ if uniqueness holds only up to a global phase}
\end{array}, \right.
$$ 
and the relative residual is defined as
$$r(\hat{f}) = \frac{\| \ Y - |\Phi \Lambda \PO\{\hat{f}\}| \ \| }{\| Y\|}$$
where $\PO$ is introduced 
if $\hat f$ may not strictly satisfy the object domain constraint as in the case of HIO.

\commentout{
The normalized errors defined in \eqref{tvepo} and \eqref{tvepf} measure how far the reconstruction is from the object domain constraint and the Fourier intensity measurement respectively.}

We consider three types of noise: Gaussian, Poisson and
illumination noise, the last of which is defined as follows. 
Suppose the illumination field is noisy $\tilde{\lambda}(\bn) = \exp(i \tilde{\phi}(\bn))$ with $\tilde{\phi}(\bn) = \phi(\bn) + t(\delta,\bn)$ where  $t(\delta,\bn)$ are  independent, uniform random variables  in $[-\pi\delta / 100, \pi \delta /100], \delta>0$.

We also test phasing with low resolution illumination which does not consist independently distributed  pixel values 
 but independently distributed  blocks
of deterministic (indeed, uniform) values. In our experiments,  illumination of independent $40 \times 40$ blocks works well for real-valued nonnegative images
and,  for complex images, illumination of  independent $4 \times 4$ blocks works well.

\subsection{Convergence Test}

\begin{figure}[h]
  \centering
  \subfigure[]{
    \includegraphics[width = 3cm]{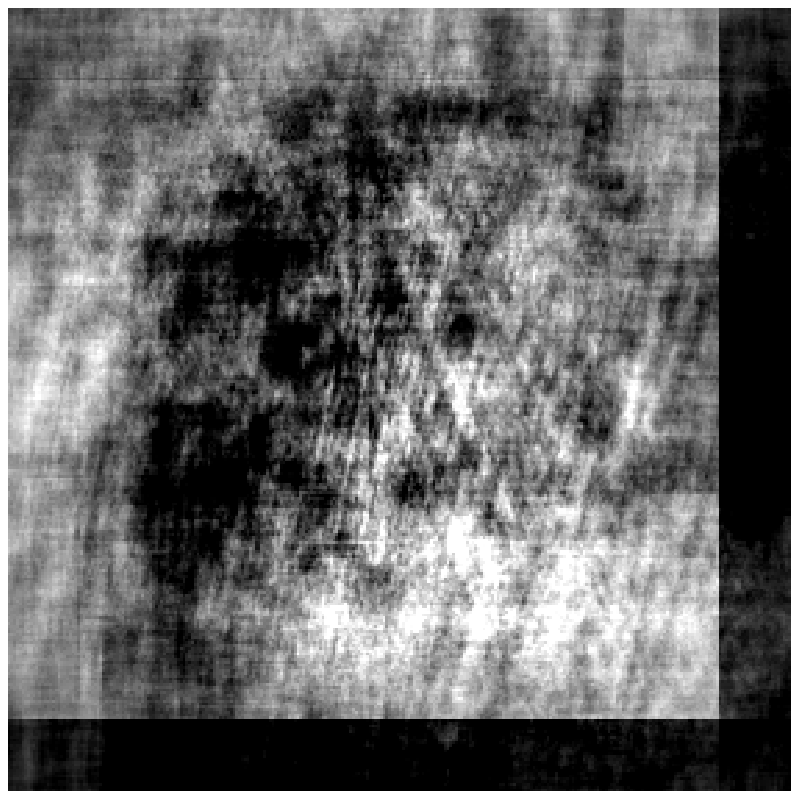}}
          \subfigure[]{
         \includegraphics[width = 4cm, height = 3cm]{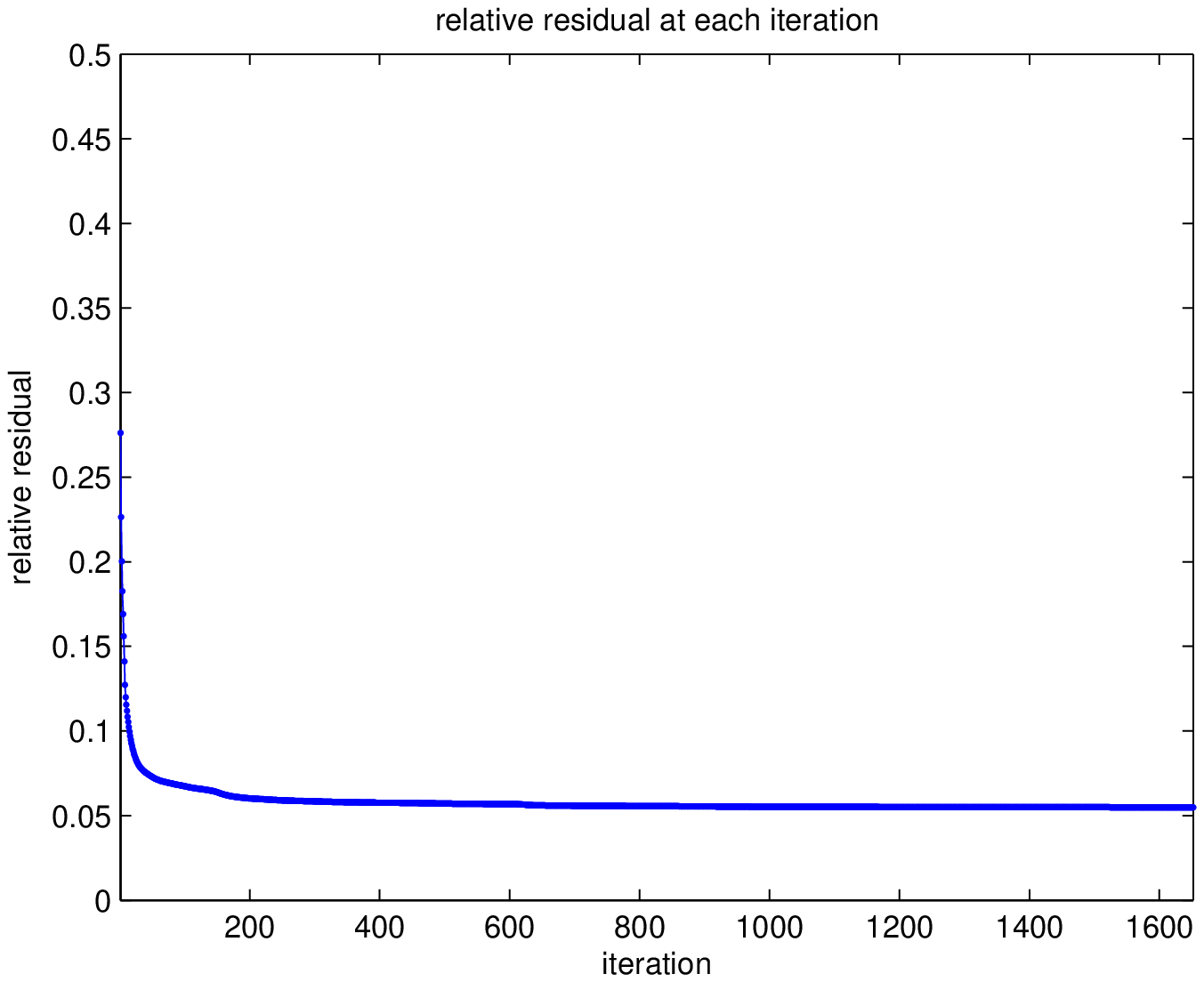}   \label{fig2b}}
        \subfigure[]{
    \includegraphics[width = 3cm]{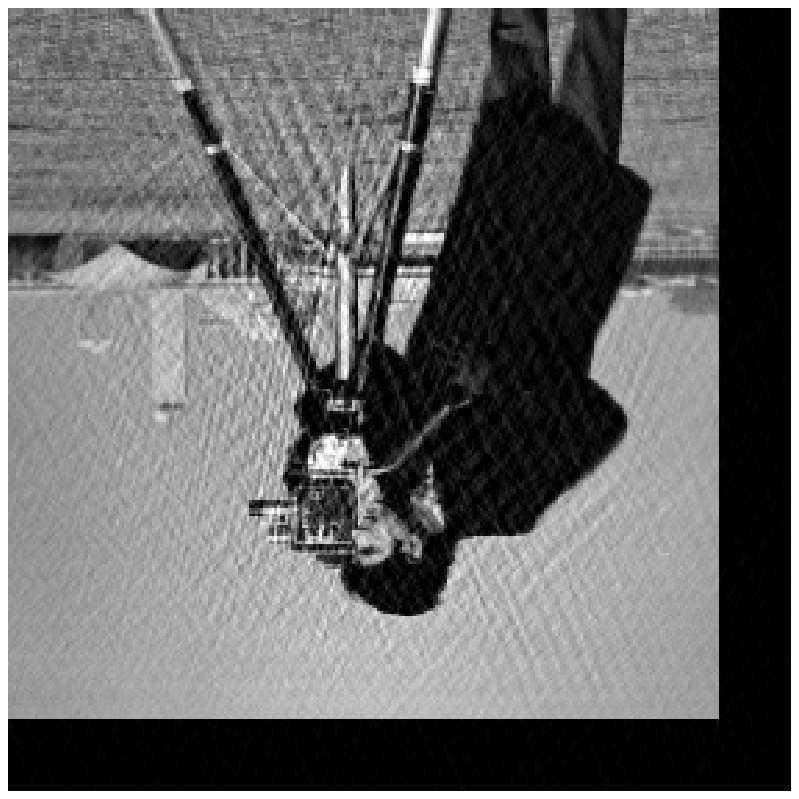}}
          \subfigure[]{
         \includegraphics[width = 4cm, height = 3cm]{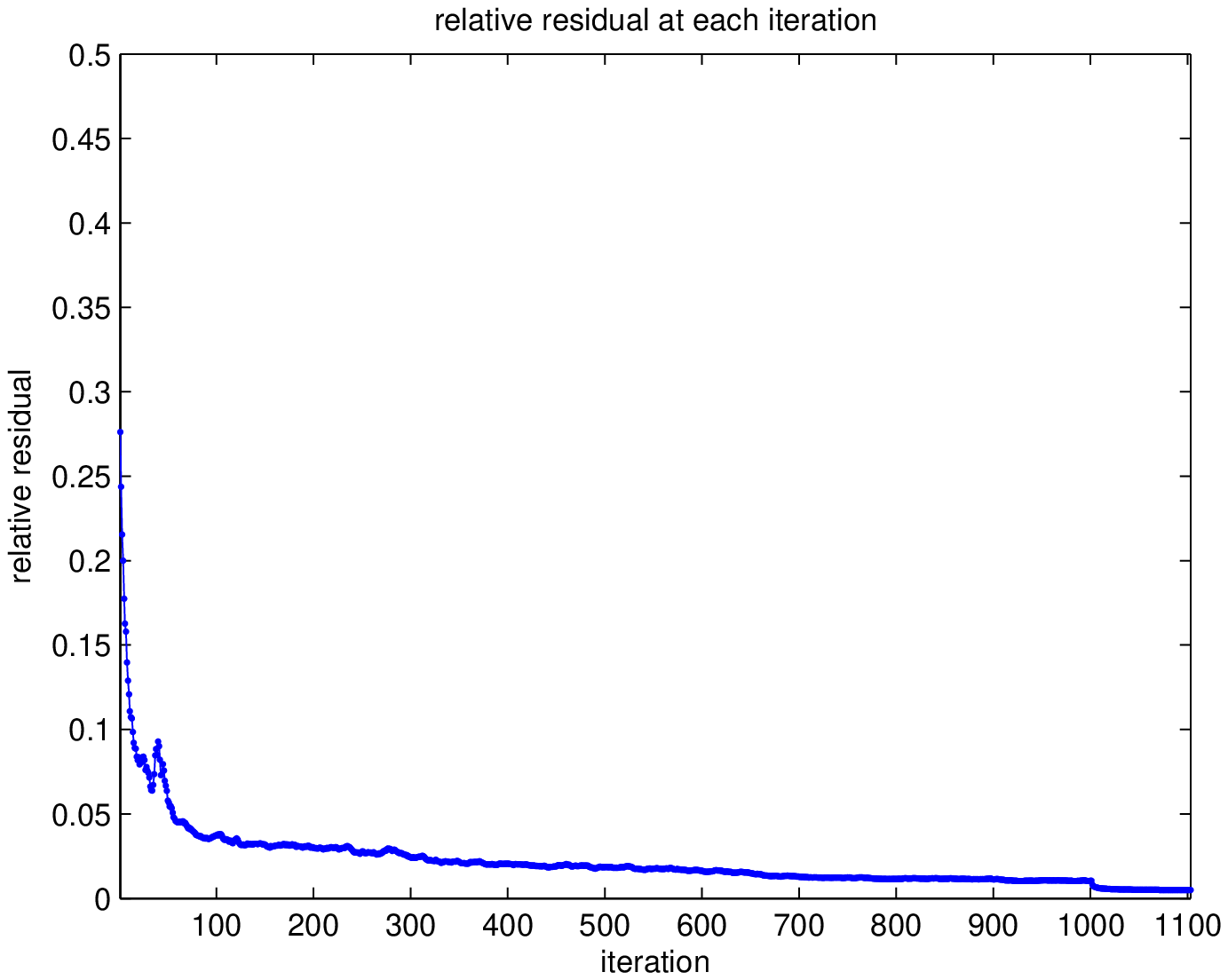}}
                 \subfigure[]{
    \includegraphics[width = 3cm]{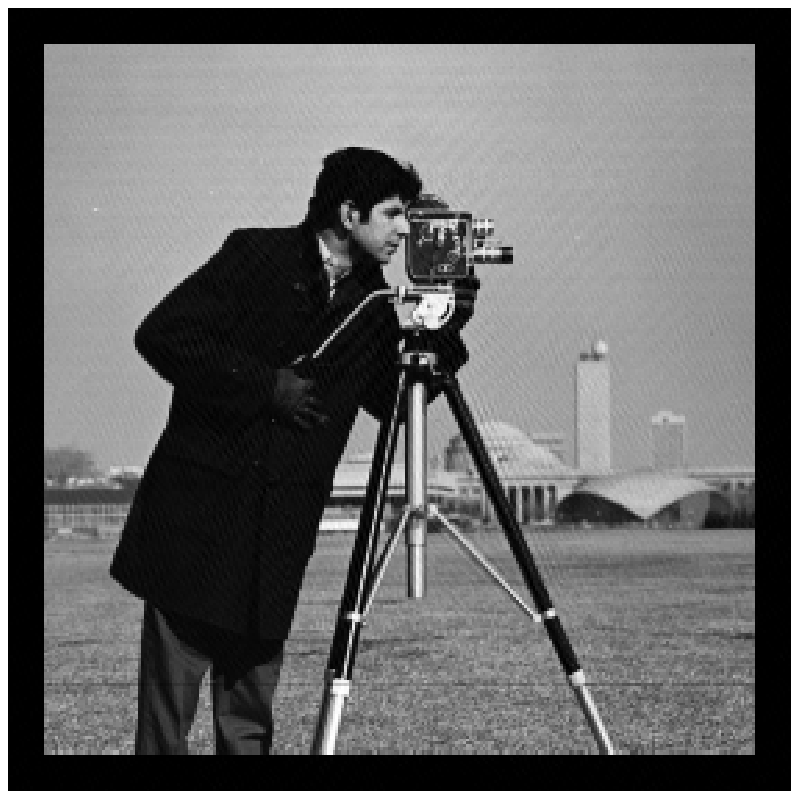}}
          \subfigure[]{
         \includegraphics[width = 4cm, height = 3cm]{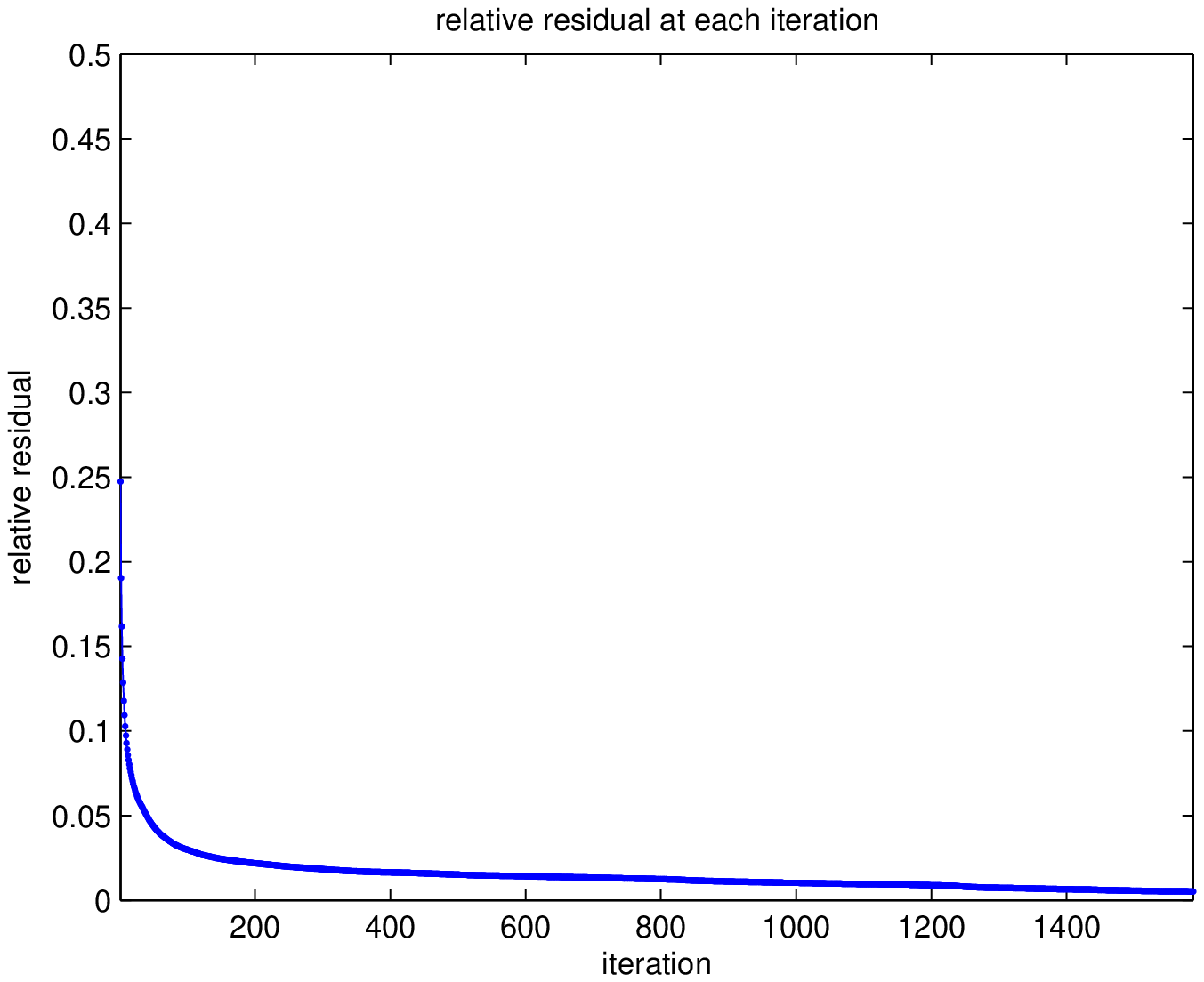}}
          \subfigure[]{
    \includegraphics[width = 3cm]{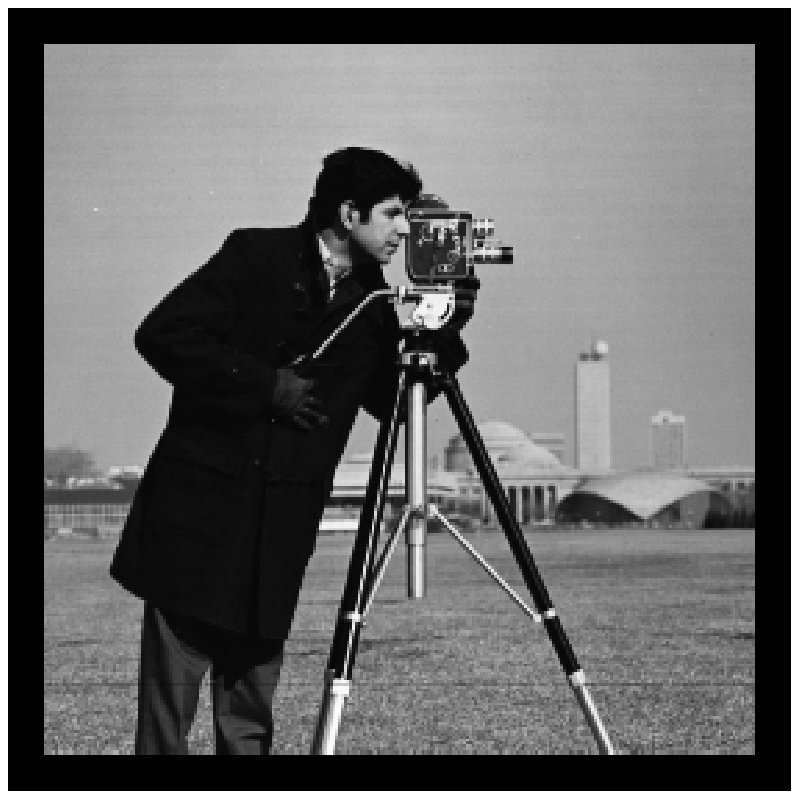}}
          \subfigure[]{
         \includegraphics[width = 4cm, height = 3cm]{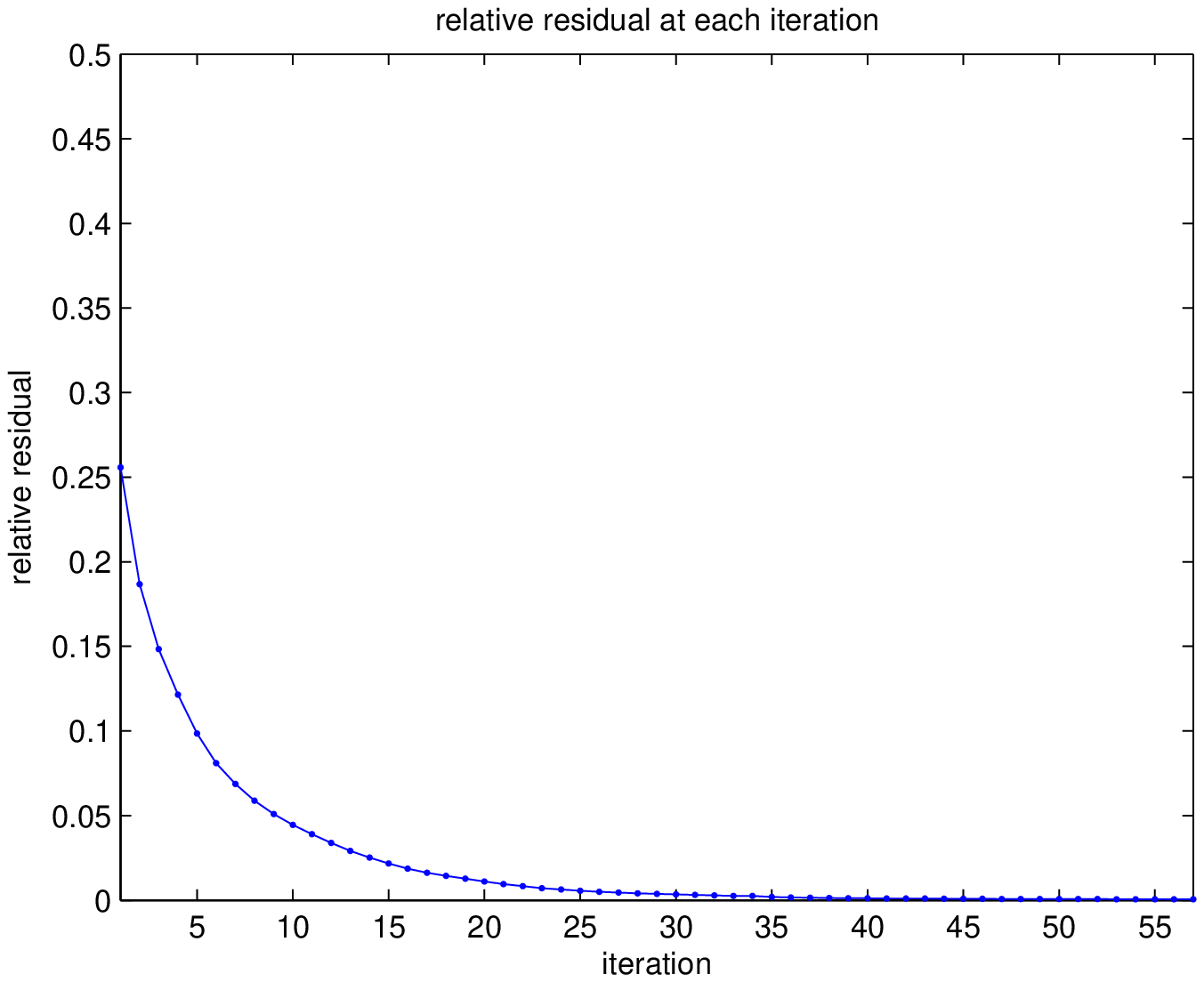}}
     \subfigure[]{
    \includegraphics[width = 3cm]{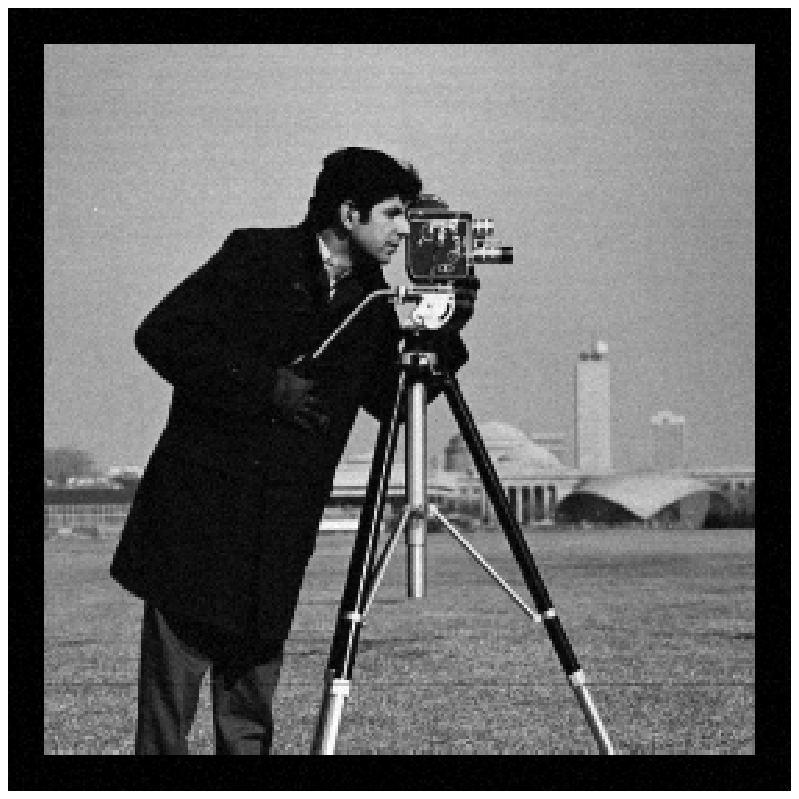}
 }
          \subfigure[]{
         \includegraphics[width = 4cm, height = 3cm]{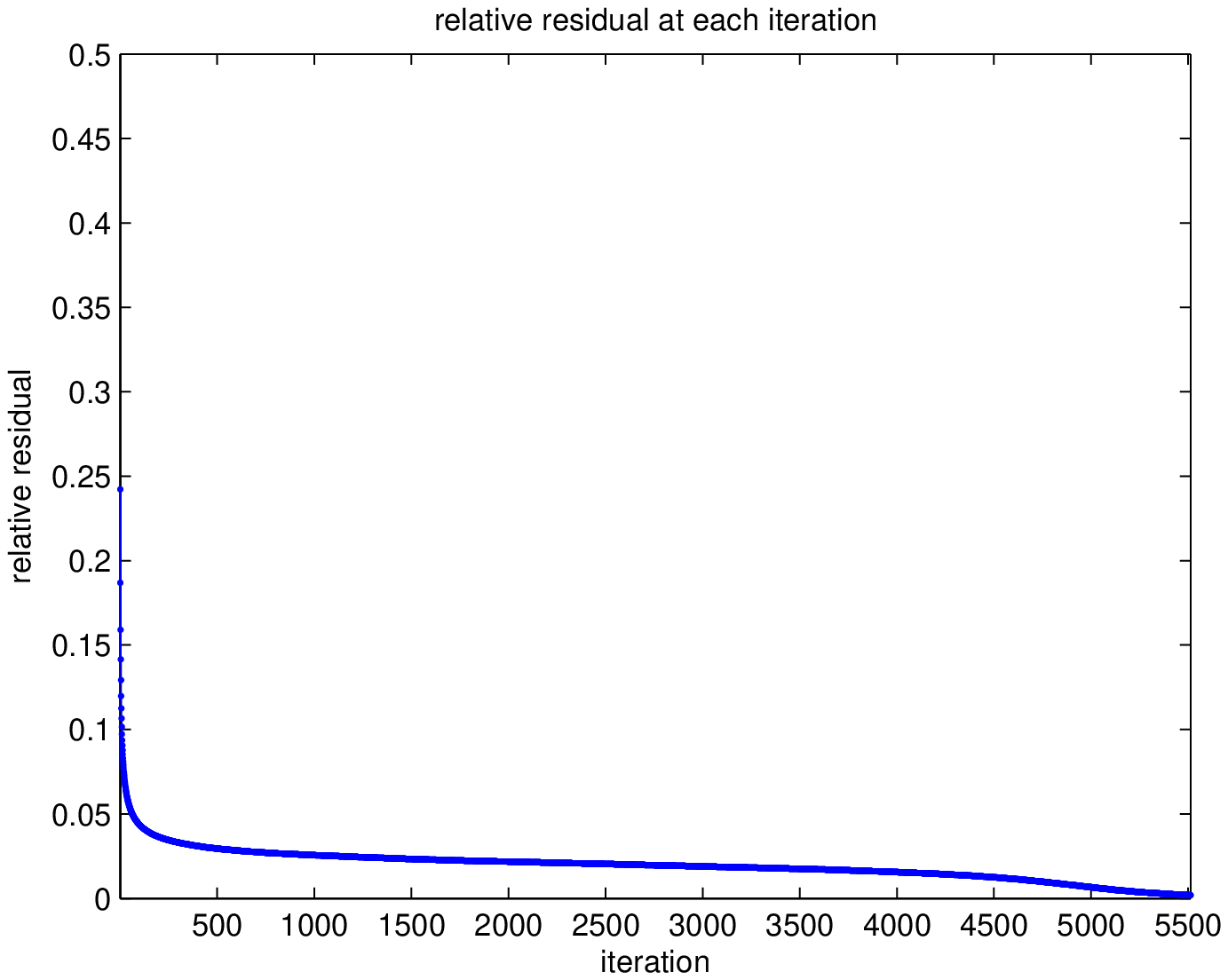}}
          \subfigure[]{
    \includegraphics[width = 3cm]{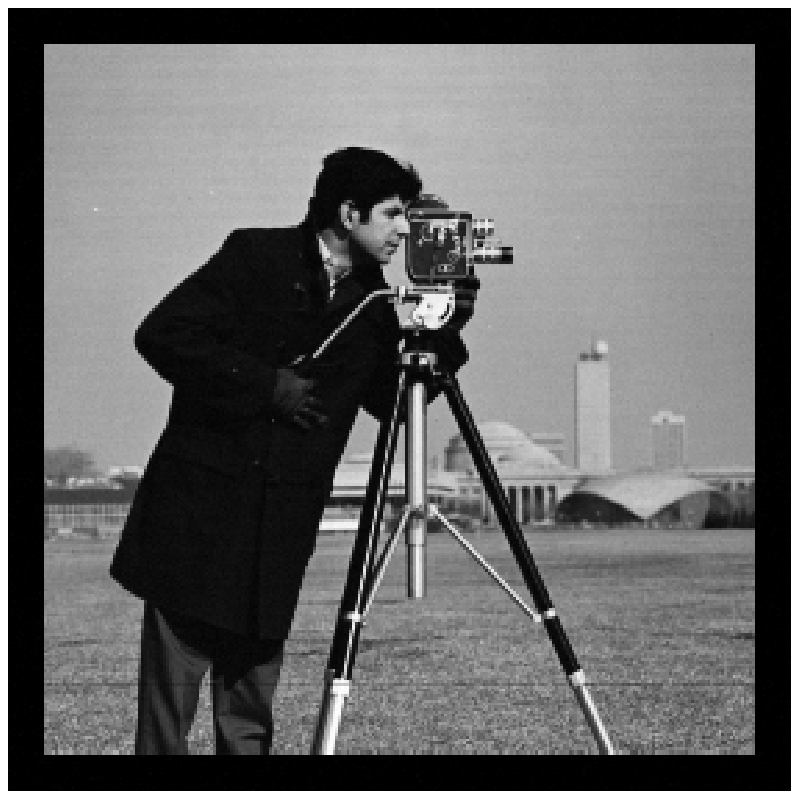}}
          \subfigure[]{
         \includegraphics[width = 4cm, height = 3cm]{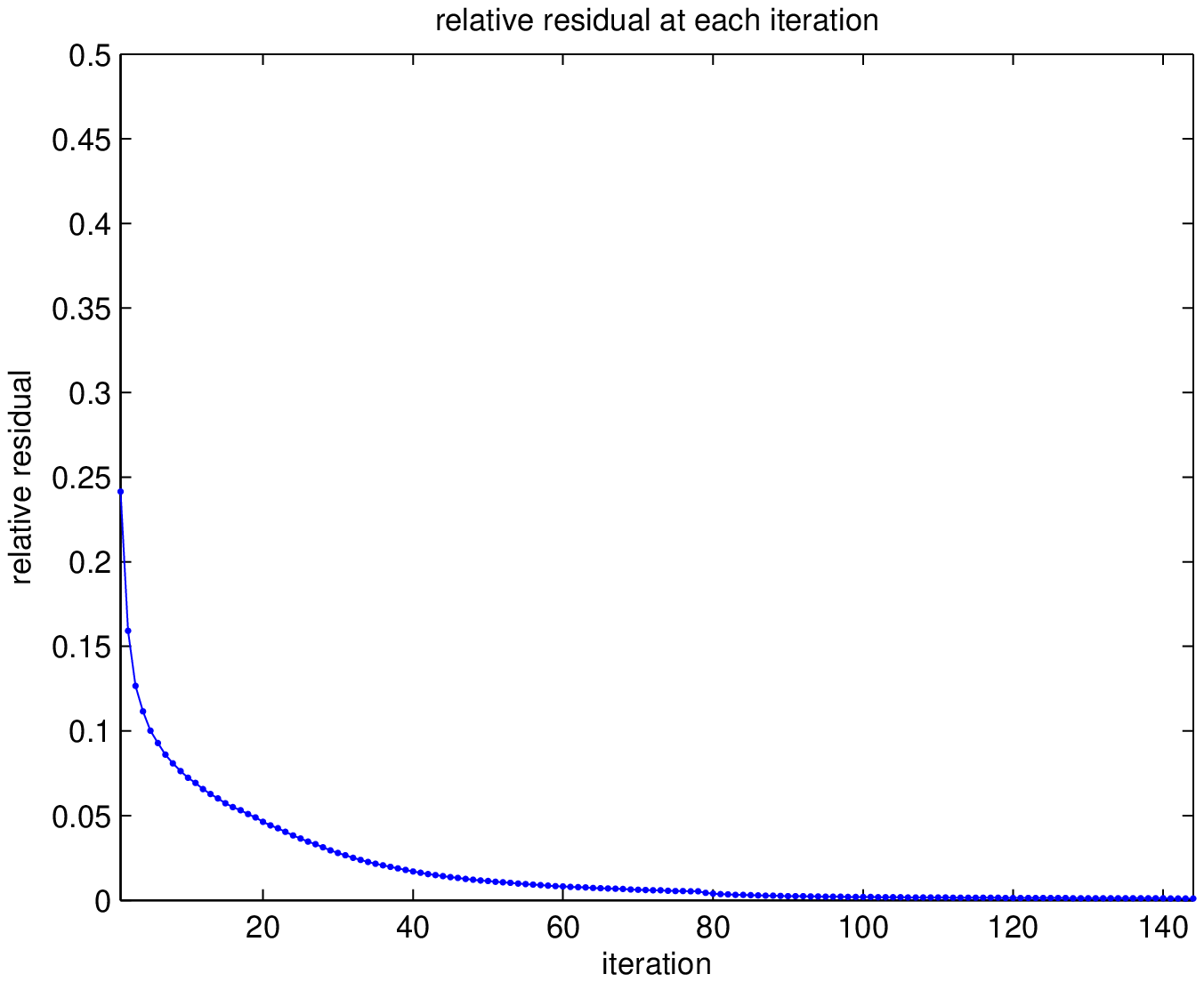}}
          \caption{(a) Recovery by $1651$ ER iterations with  UI and  $\sigma = 4$. (b) $r(f_k)$ versus k with $r(\hat{f}) \approx 5.49\%$. 
          (c) Recovery by $1000$ HIO + $103$ ER with UI and $\sigma = 4$. (d) $r(f_k)$ versus k with $r(\hat{f}) \approx 0.49\%$. 
          (e) Recovery by $1587$ ER steps with one low resolution RPI with $\sigma = 2$.  (f) $r(f_k)$ versus k with $r(\hat{f}) \approx 0.52\%$ 
          and $e(\hat{f}) \approx 2.51\%$.  
             (g) Recovery by $33$ HIO + $24$ ER steps with low resolution RPI  with $\sigma = 2$. (h)  $r(f_k)$ versus k with $r(\hat{f}) \approx 0.05\%$
           and   $e(\hat{f}) \approx 0.32\%$.  
             (i) Recovery by $5512$ ER steps with high resolution RPI with $\sigma = 1$.  (j) $r(f_k)$ versus k with  $r(\hat{f}) \approx 0.19\%$
             and $e(\hat{f}) \approx 3.27\%$. 
             (k) Recovery by $77$ HIO + $67$ ER steps with high resolution RPI with $\sigma = 1$. (l) $r(f_k)$ versus k with $r(\hat{f}) \approx 0.10\%$
             and $e(\hat{f}) \approx 1.39\%$.   
          }
        \label{Cameraman} 
\end{figure}

\begin{figure}[h]
  \centering
  \subfigure[]{
    \includegraphics[width = 3cm]{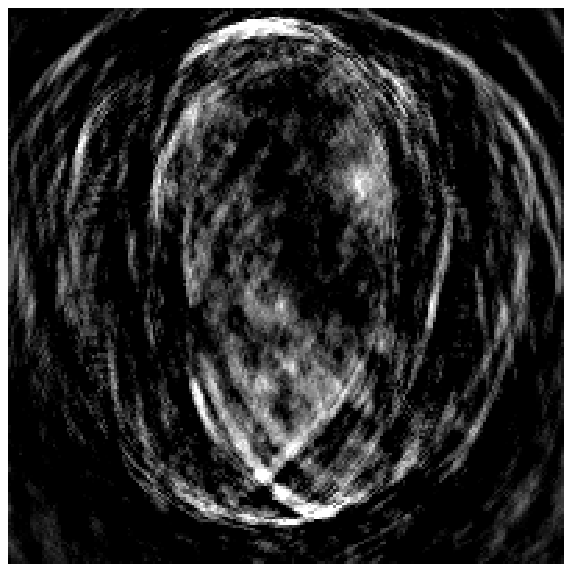}}
          \subfigure[]{
         \includegraphics[width = 4cm, height = 3cm]{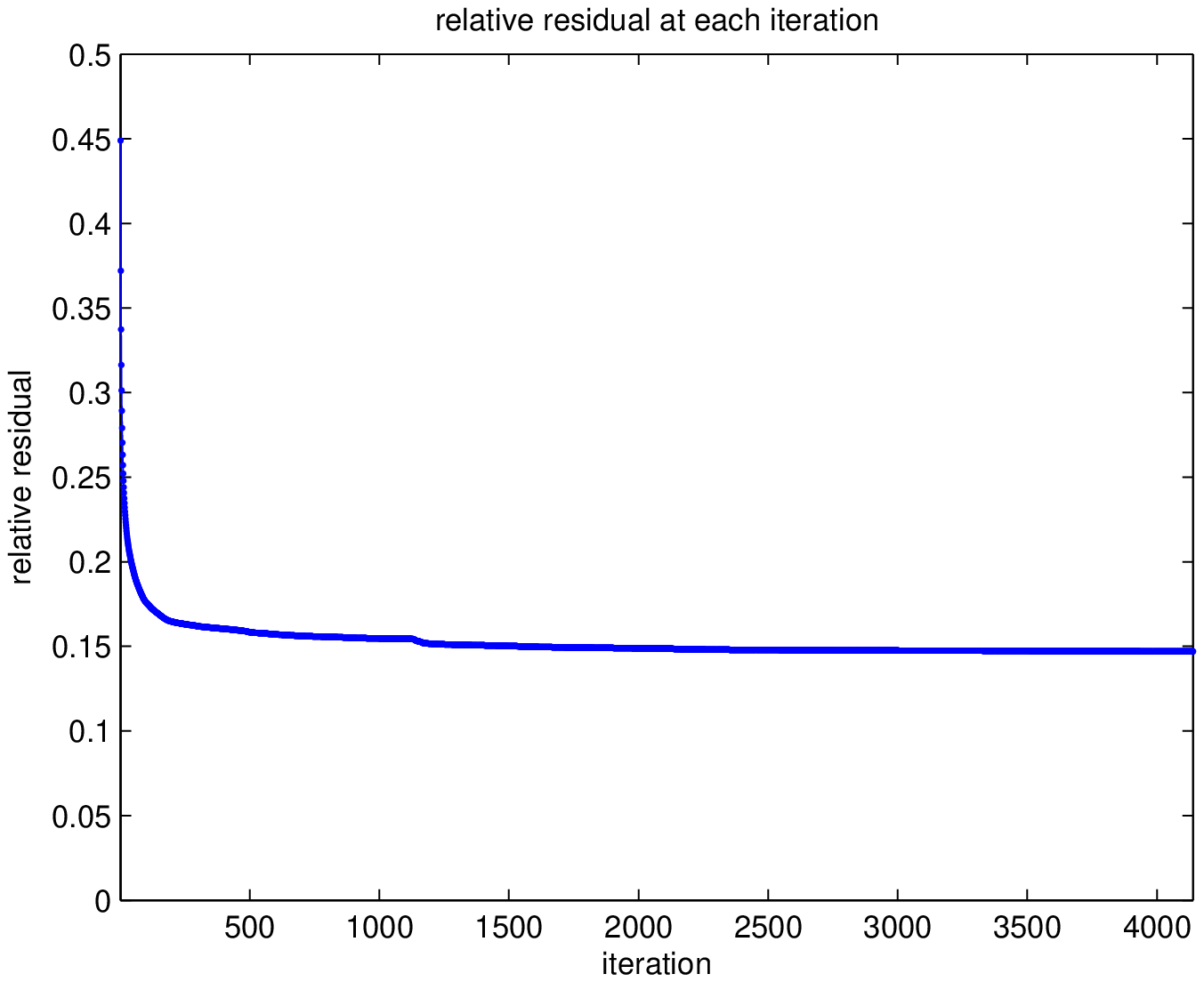}}
        \subfigure[]{
    \includegraphics[width = 3cm]{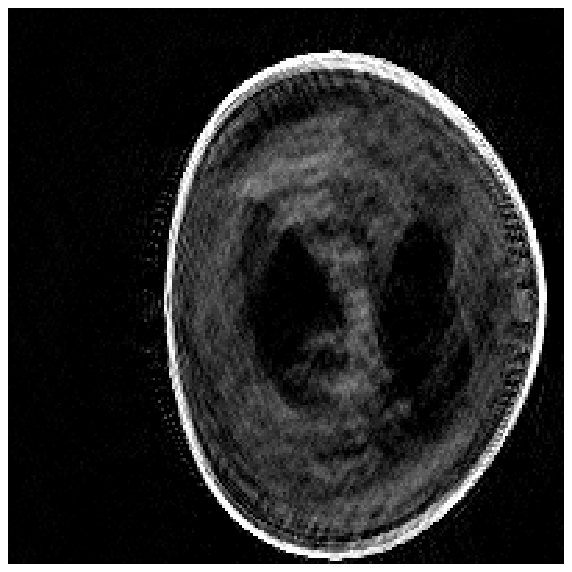}}
          \subfigure[]{
         \includegraphics[width = 4cm, height = 3cm]{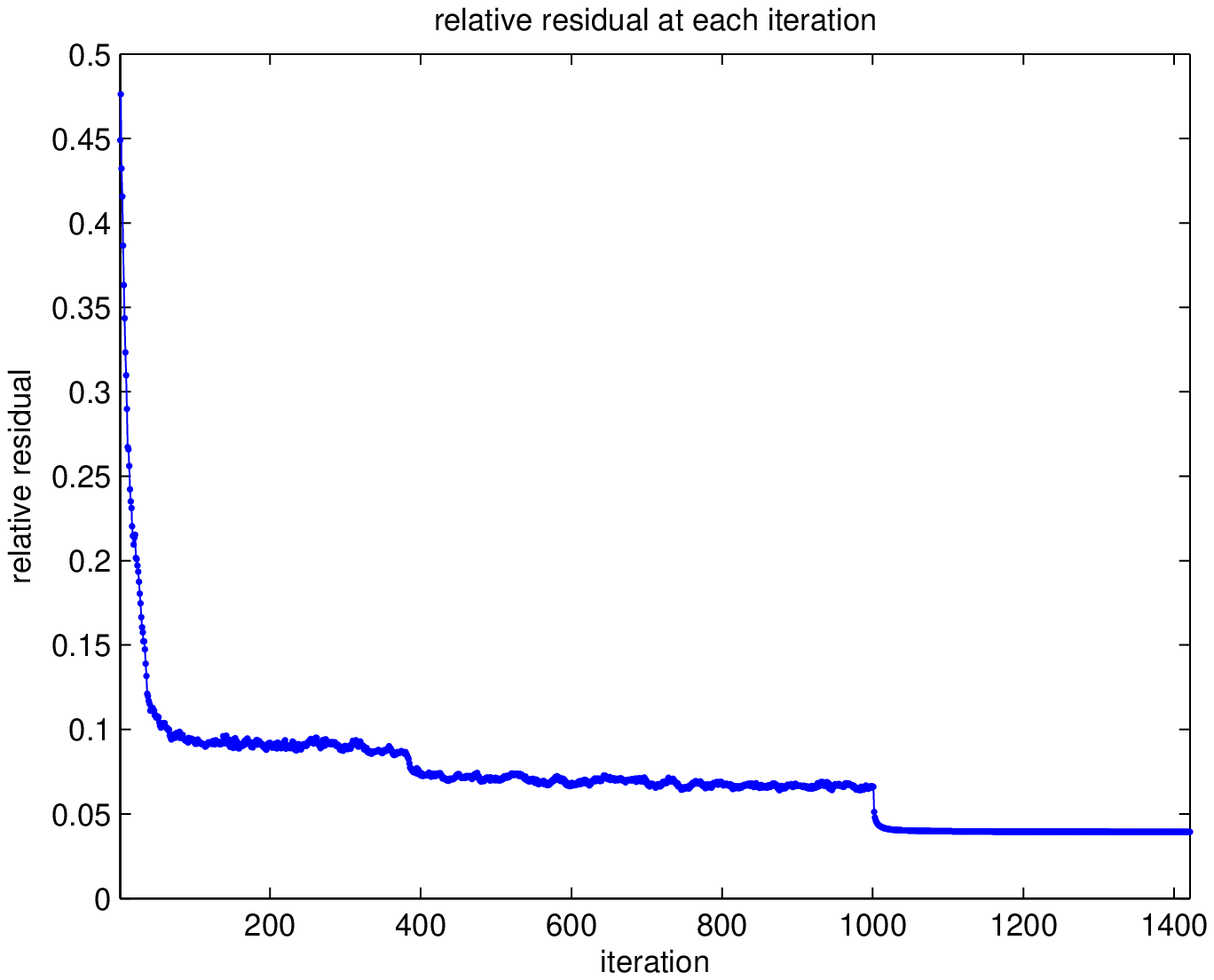}}
                 \subfigure[]{
    \includegraphics[width = 3cm]{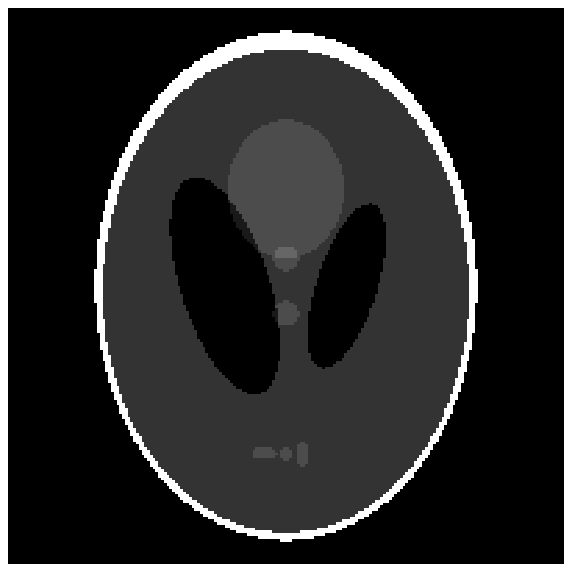}}
          \subfigure[]{
         \includegraphics[width = 4cm, height = 3cm]{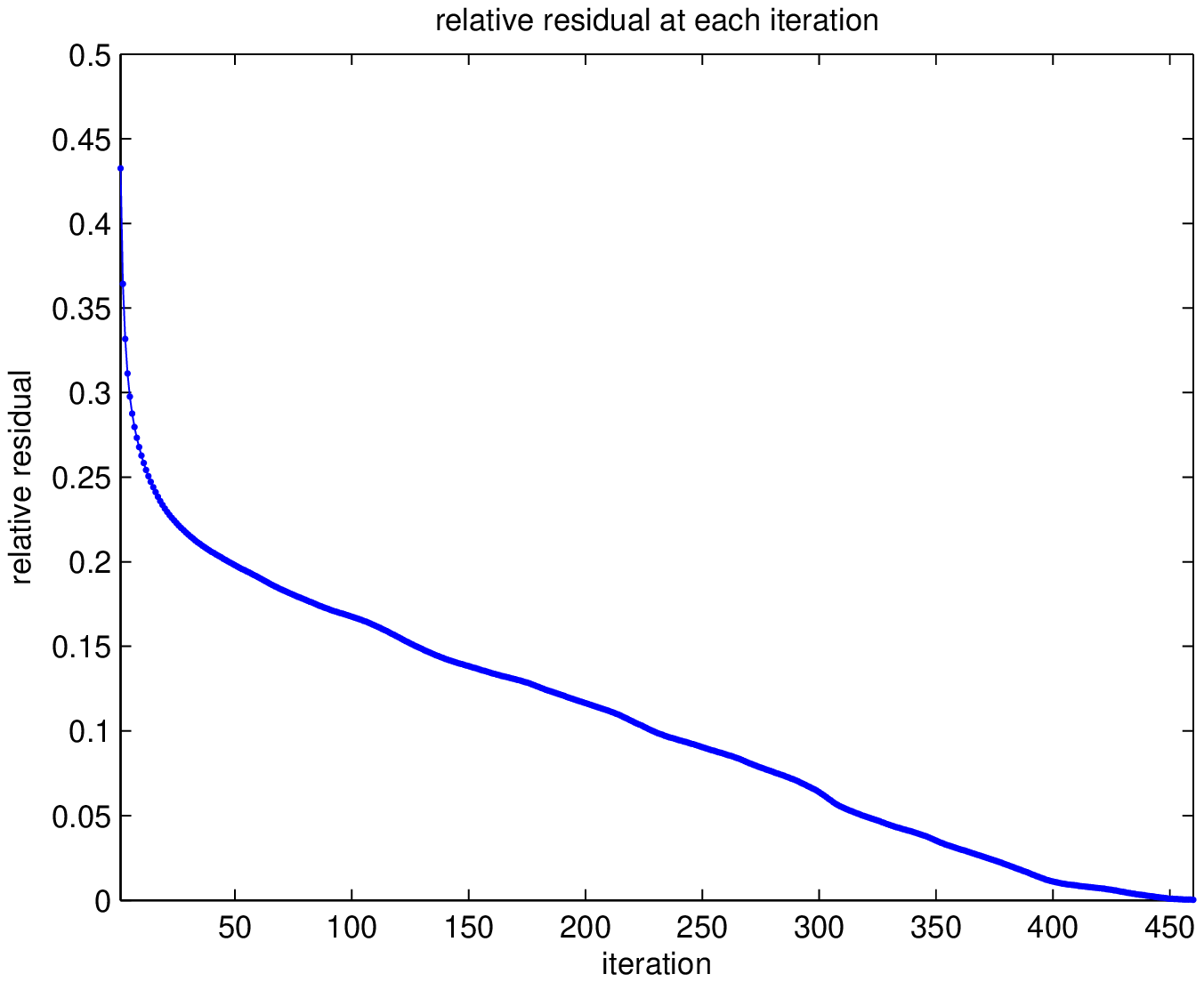}}
          \subfigure[]{
    \includegraphics[width = 3cm]{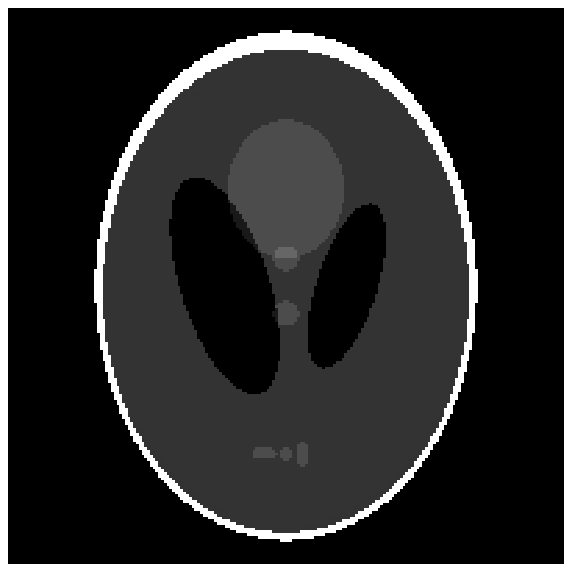}}
          \subfigure[]{
         \includegraphics[width = 4cm, height = 3cm]{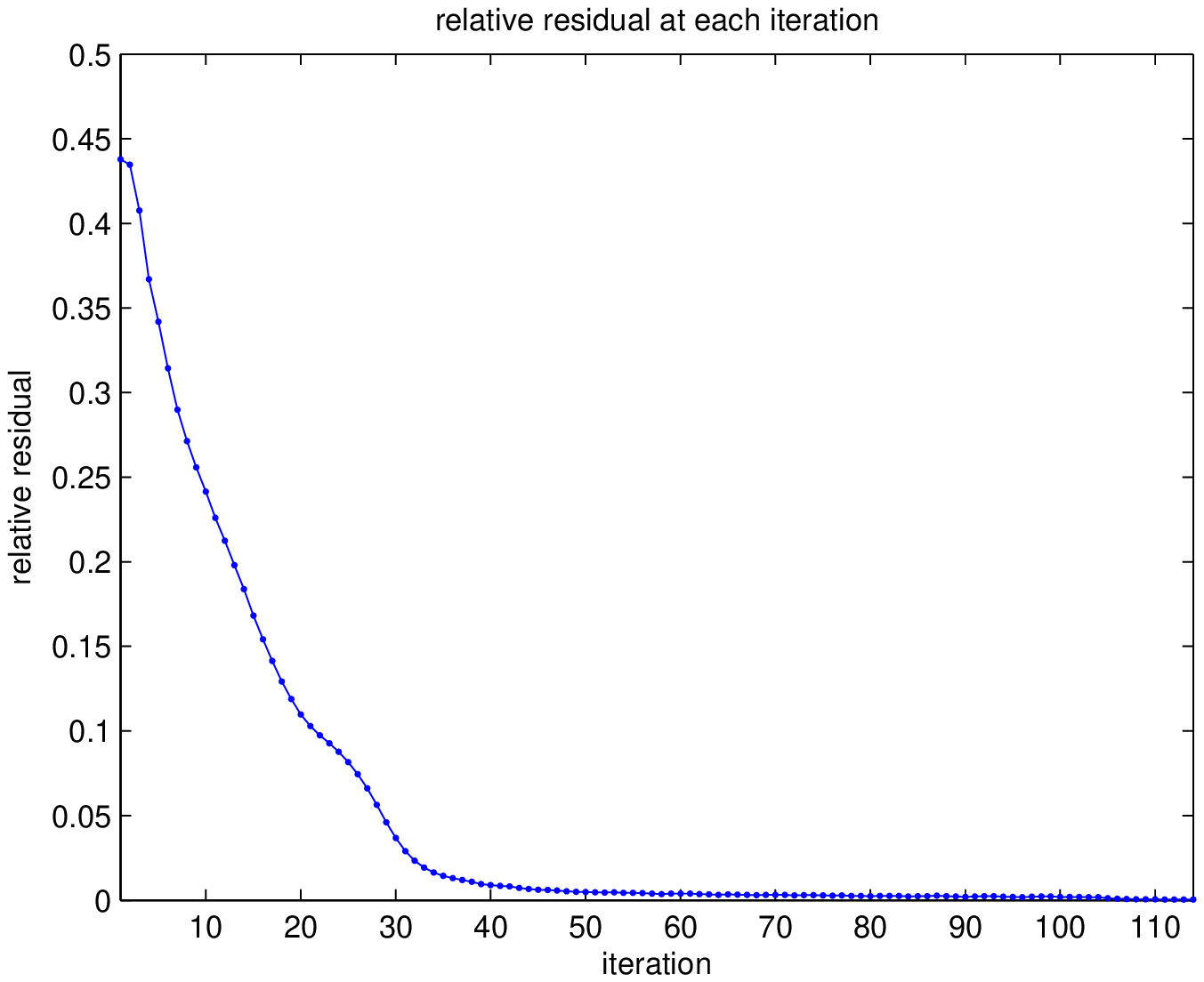}}
     \subfigure[]{
    \includegraphics[width = 3cm]{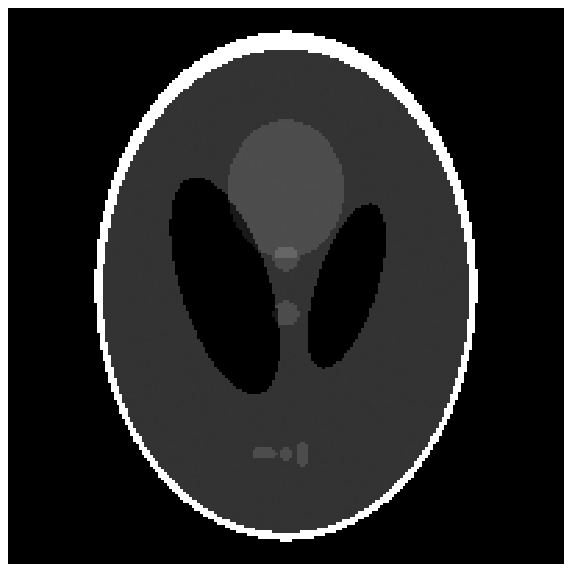}}
          \subfigure[]{
         \includegraphics[width = 4cm, height = 3cm]{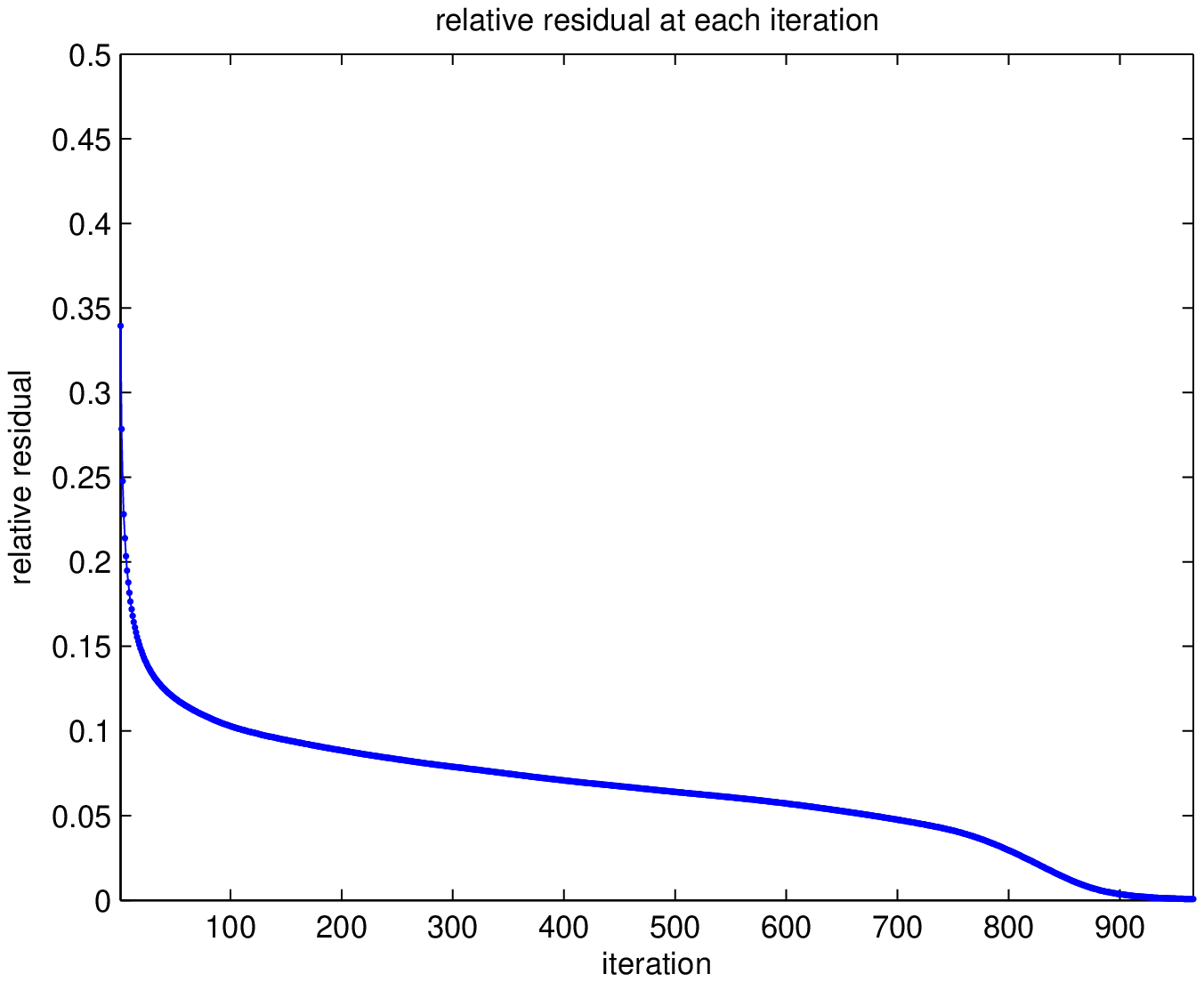}}
          \subfigure[]{
    \includegraphics[width = 3cm]{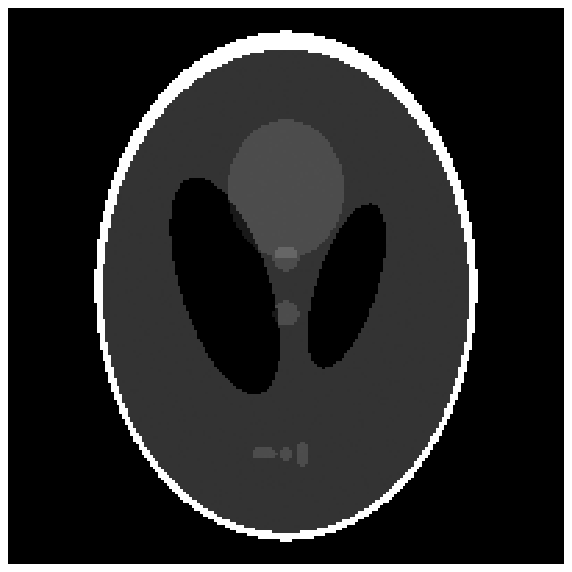}}
          \subfigure[]{
         \includegraphics[width = 4cm, height = 3cm]{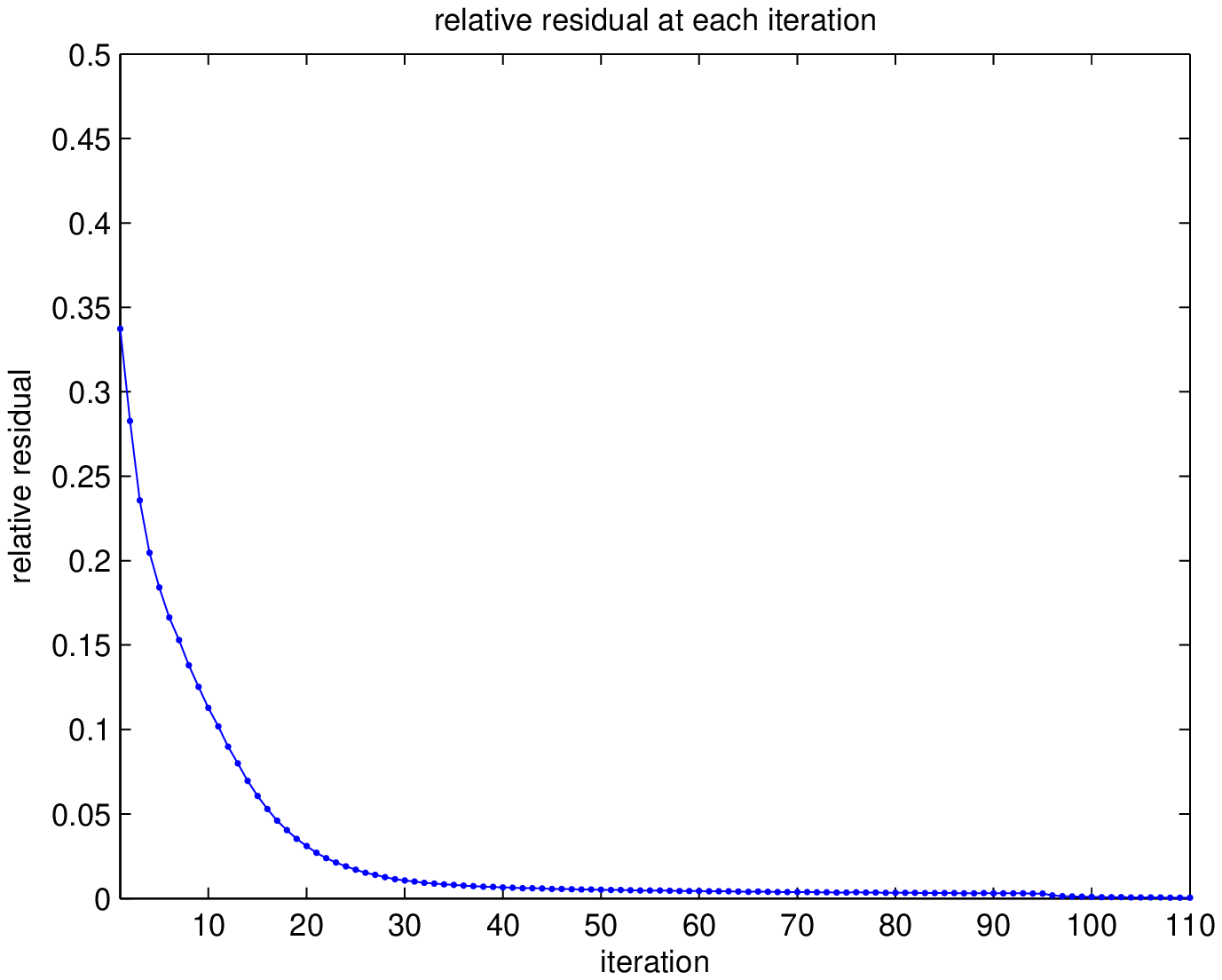}}
          \caption{(a) Recovery by $4140$ ER iterations with UI with $\sigma = 4$.   (b) $r(f_k)$ versus k with
         $r(\hat{f}) \approx 14.71\%$.
          (c) Recovery by $1000$ HIO + $421$ ER steps with one UI with $\sigma = 4$.   (d) $r(f_k)$ versus k
           with $r(\hat{f}) \approx 3.94\%$.
          (e) Recovery by $460$ ER with one low resolution RPI  with $\sigma = 2$.  (f) $r(f_k)$ versus k with $r(\hat{f}) \approx 0.03\%$ and  $e(\hat{f}) \approx 0.09\%$. 
             (g) Recovery by $103$ HIO + $11$ ER steps with one low resolution RPI with $\sigma = 2$.   (h) $r(f_k)$ versus k with $r(\hat{f}) \approx 0.03\%$
             and $e(\hat{f}) \approx 0.12\%$.  
             (i) Recovery  by $966$ ER steps with one high resolution RPI with $\sigma = 1$.  (j) $r(f_k)$ versus k with $r(\hat{f}) \approx 0.06\%$ and $e(\hat{f}) \approx 0.40\%$. 
             (k) Recovery  by $94$ HIO + $16$ ER steps with one high resolution RPI  with $\sigma = 1$.  (l) $r(f_k)$ versus k with $r(\hat{f}) \approx 0.04\%$
             and $e(\hat{f}) \approx 0.19\%$. 
          }
 \label{Phantom} 
\end{figure}

The reconstruction of the  real-valued nonnegative images Cameraman and Phantom with one UI or RPI  is shown in Figures \ref{Cameraman} and \ref{Phantom} respectively.
For Figures \ref{Cameraman} and \ref{Phantom}, we terminate the pure ER when $\|f_{k+1}-f_k\|/\|f_k\|<0.01\%$.  For HIO+ER, HIO is stopped when $\|f_{k+1}-f_k\|/\|f_k\|<1\%$ with a maximal $1000$ iterations and ER is terminated when $\|f_{k+1}-f_k\|/\|f_k\|<0.01\%$. In Figure \ref{Cameraman} and \ref{Phantom}, $\|\hat{f}-\PO\PF\hat{f}\|/\|\hat{f}\| $ is as small as $0.01\%$, implying  that $\hat{f}$ is near a fixed point of $\PO\PF$.

As commented before the pure ER iteration always converges to a fixed point of $\PO\PF$. But 
with one uniform illumination and $\sigma=4$, the fixed point of $\PO\PF$ is not a phasing solution  as the relative residual stagnates at $5.49\%$ in Figure \ref{Cameraman}(b) and at $14.71\%$ in Figure \ref{Phantom}(b). HIO followed by ER improves the recovery over pure ER but the recovered Cameraman in Figure \ref{Cameraman}(c) displays 
the well known  artifact of stripe pattern and the recovered Phantom in Figure \ref{Phantom}(c) is severely blurred and distored. 

With one low resolution RPI (block size: $40 \times 40$) and  $\sigma = 2$, the recovered images in Figure \ref{Cameraman}(e), \ref{Cameraman}(g), \ref{Phantom}(e) and \ref{Phantom}(g) are excellent approximation to the true images, even though absolute uniqueness is not guaranteed for low resolution RPI.  HIO+ER 
is superior to pure ER  in significant speed-up in
convergence (Figure \ref{Cameraman}(f) versus \ref{Cameraman}(h), Figure \ref{Phantom}(f) versus \ref{Phantom}(h)).

With one high resolution RPI, high quality reconstruction 
can still be achieved with the oversampling ratio equal to $1$, cf. Figures \ref{Cameraman}(i), \ref{Cameraman}(k), \ref{Phantom}(i) and \ref{Phantom}(k). Notice the rapid convergence of HIO+ER in Figures \ref{Cameraman}(l)
and \ref{Phantom}(l).

\subsection{Oversampling Ratio Test}
To systematically test the oversampling ratio 
required for phasing with RPI, we
introduce 5\% different types of noise (Gaussian, Poisson, Illumination), use low (block size: 40 $\times $ 40)  as well as high resolution RPI and let $\sigma$ vary.
We use an adaptive version of HIO+ER:  HIO and ER are terminated if the residual increases in  $5$ consecutive  iterations.
The  relative error of reconstruction for 
the nonnegative image Phantom
is averaged over 5 trials and shown in Figure \ref{SRtest}(a). Clearly the relative error steadily decreases as the oversampling ratio increases. Without noise, low resolution RPI can achieve near zero error with $\sigma = 1.1$. With 5\% noise, the relative error stabilizes after  $\sigma=2$ to
a level comparable to the noise. 

Next we consider the complex-valued Phantom with phases randomly distributed in the sector  $[0, \pi/2]$. Figure \ref{SRtest}(b) shows the average relative  error $e(\hat{f})$  with one high resolution or low resolution (block size: $4 \times 4$) RPI and three kinds of noise. Again the relative error stabilizes
after $\sigma=2$ to a level comparable to the noise. Note
that for $1.8<\sigma<2$, there are more free variables in the complex-valued image than in the Fourier intensity data
and yet the reconstructions are still of good quality. 

Finally, we consider the complex-valued Phantom with phases random distributed in $[0, 2\pi]$. Figure \ref{SRtest}(c) shows the average relative error $e(\hat{f})$  with one high resolution or low resolution (block size: $4 \times 4$) RPI  {\em plus one UI}. Excellent recovery is achieved 
for  $\sigma \ge 1.8$. 

\begin{figure}[ht]
  \centering
  \subfigure[]{
    \includegraphics[width = 8cm]{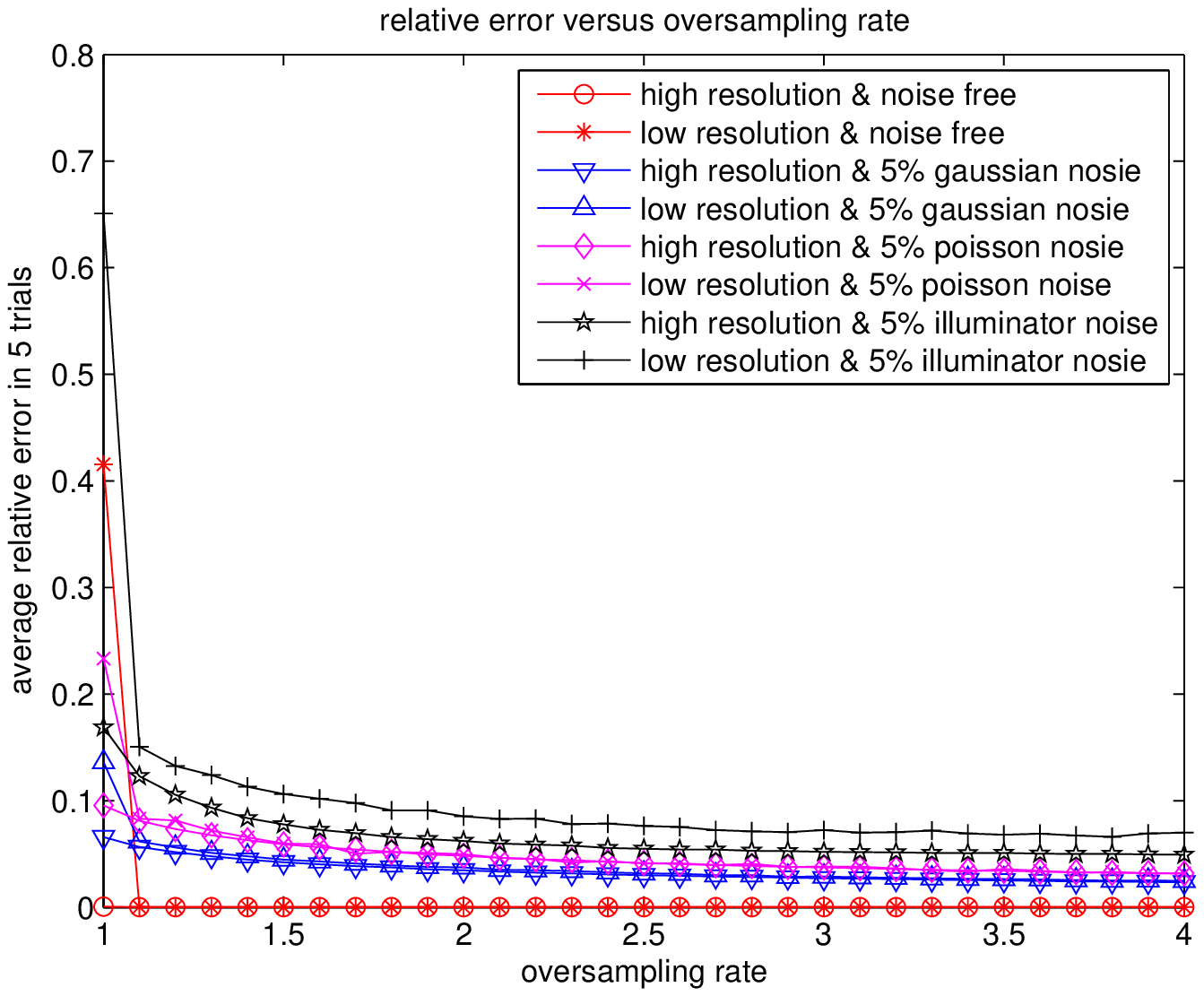}}
      \subfigure[]{
    \includegraphics[width = 8cm]{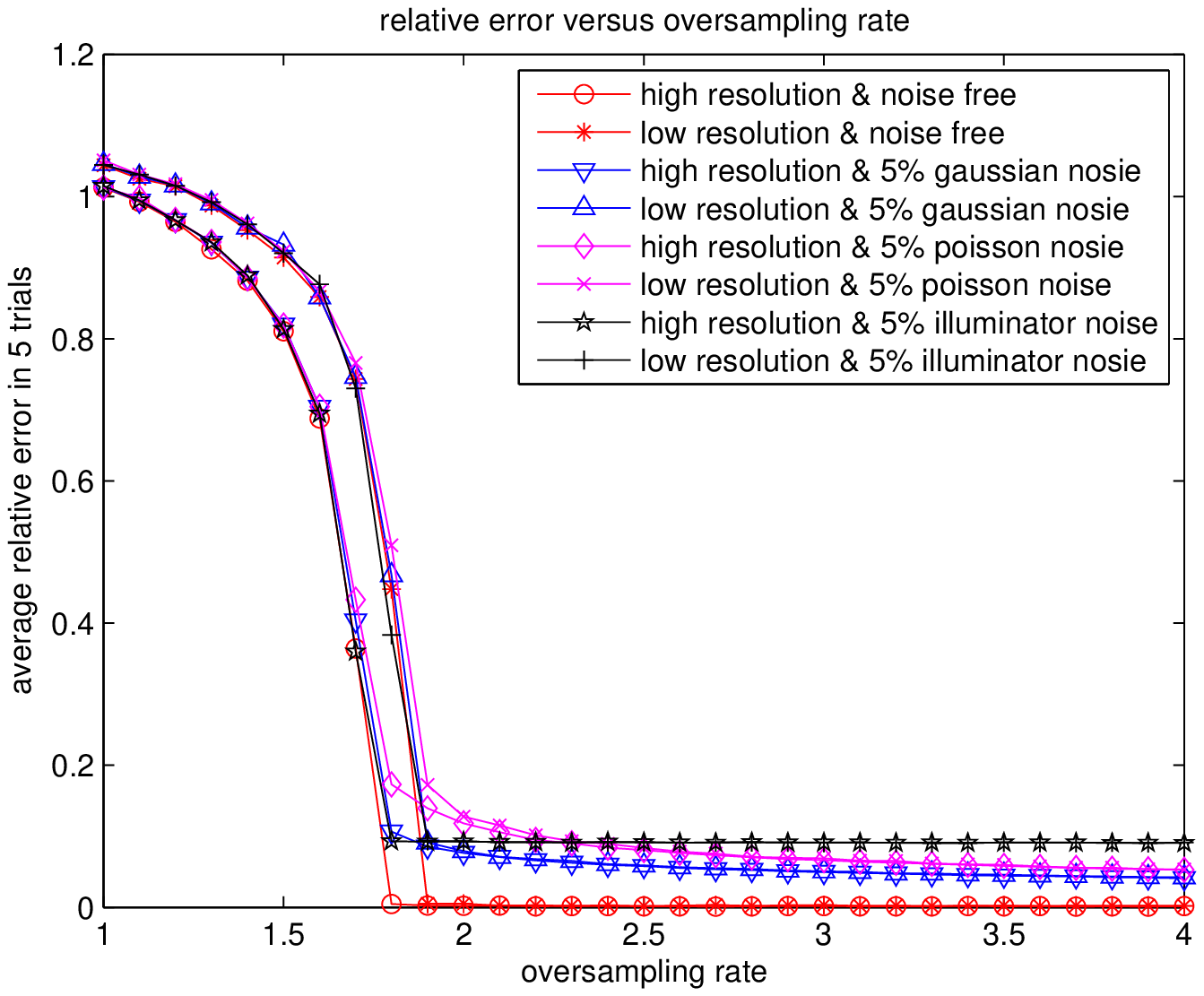}}
\subfigure[]{
    \includegraphics[width = 8cm]{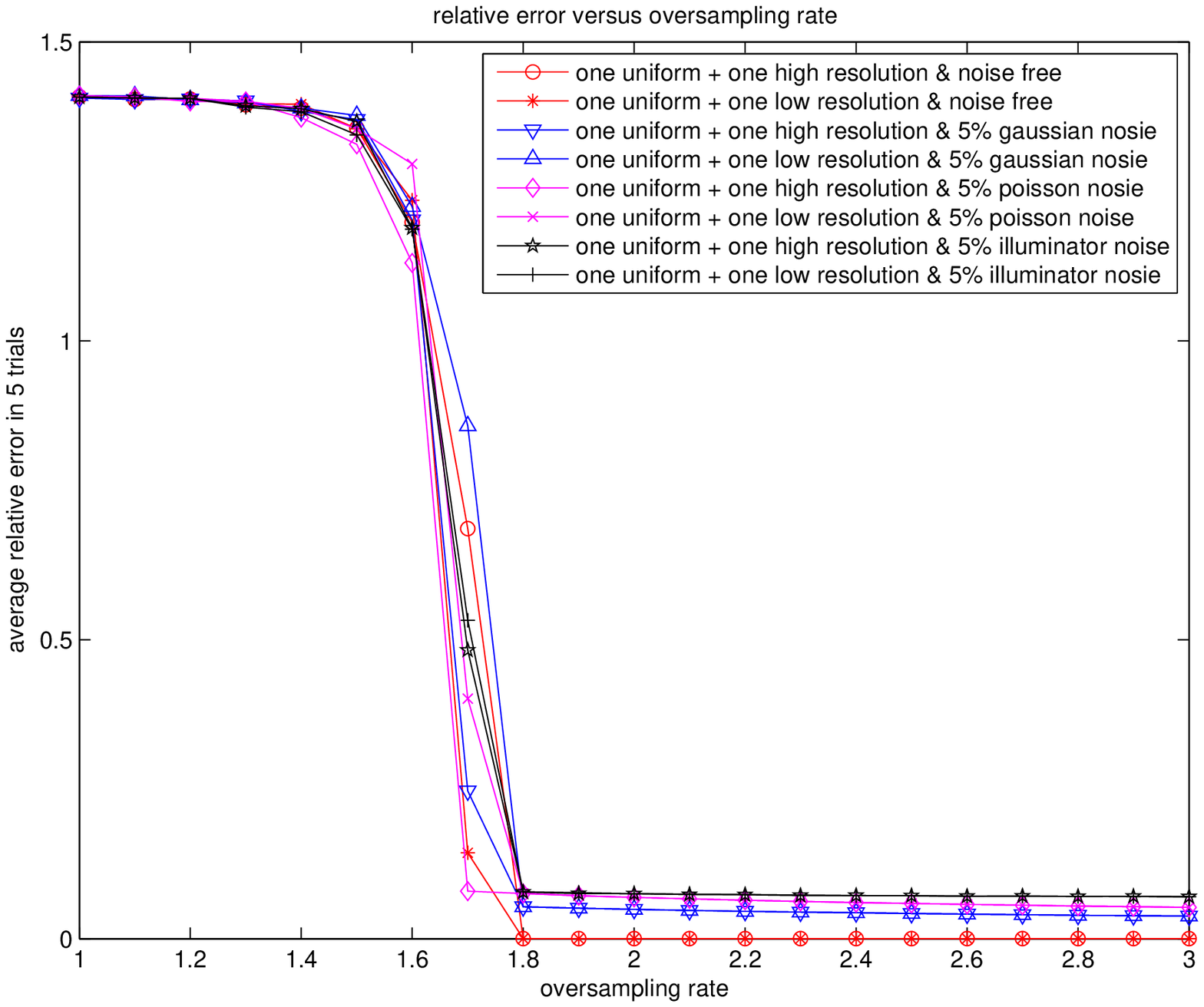}}
  \caption{(a) Relative error with one RPI for nonnegative-valued  Phantom;  (b) Relative error with one RPI for complex-valued  Phantom 
with  phases randomly distributed in $[0, \pi/2]$;  (c) Relative error  by $200$ HIO $+ 300$ ER
with one RPI and UI
for complex-valued  Phantom with  phases 
randomly distributed in $[0,2\pi]$ with one UI and one
RPI of high resolution (block size: $1 \times 1$) or low resolution (block size: $4 \times 4$). }
    \label{SRtest} 
\end{figure}

\begin{figure}[hthp]
  \centering
  \subfigure[]{
    \includegraphics[width = 3cm]{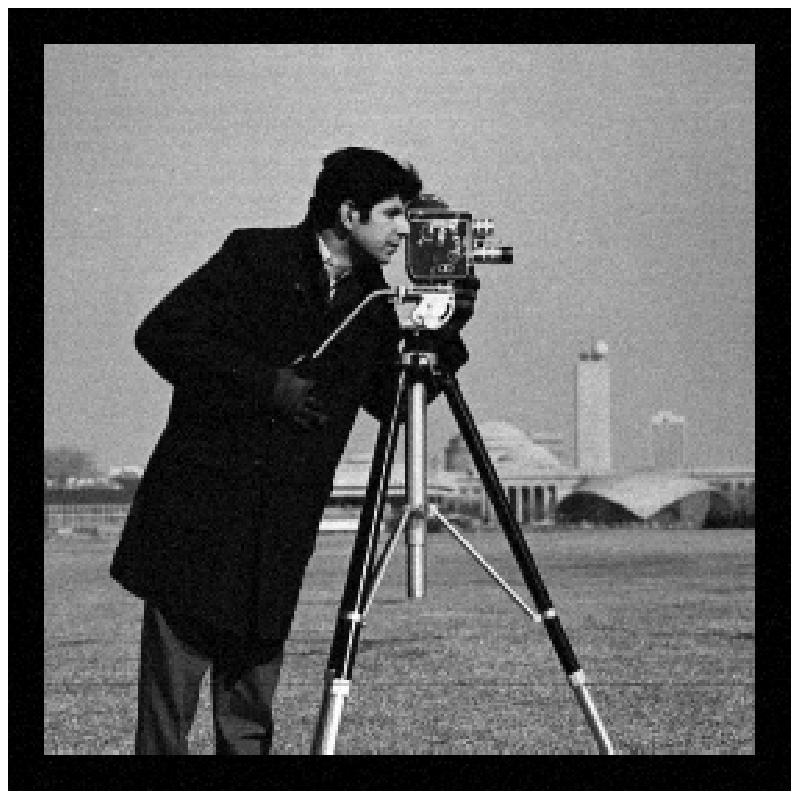}}
          \subfigure[]{
         \includegraphics[width = 4cm, height = 3cm]{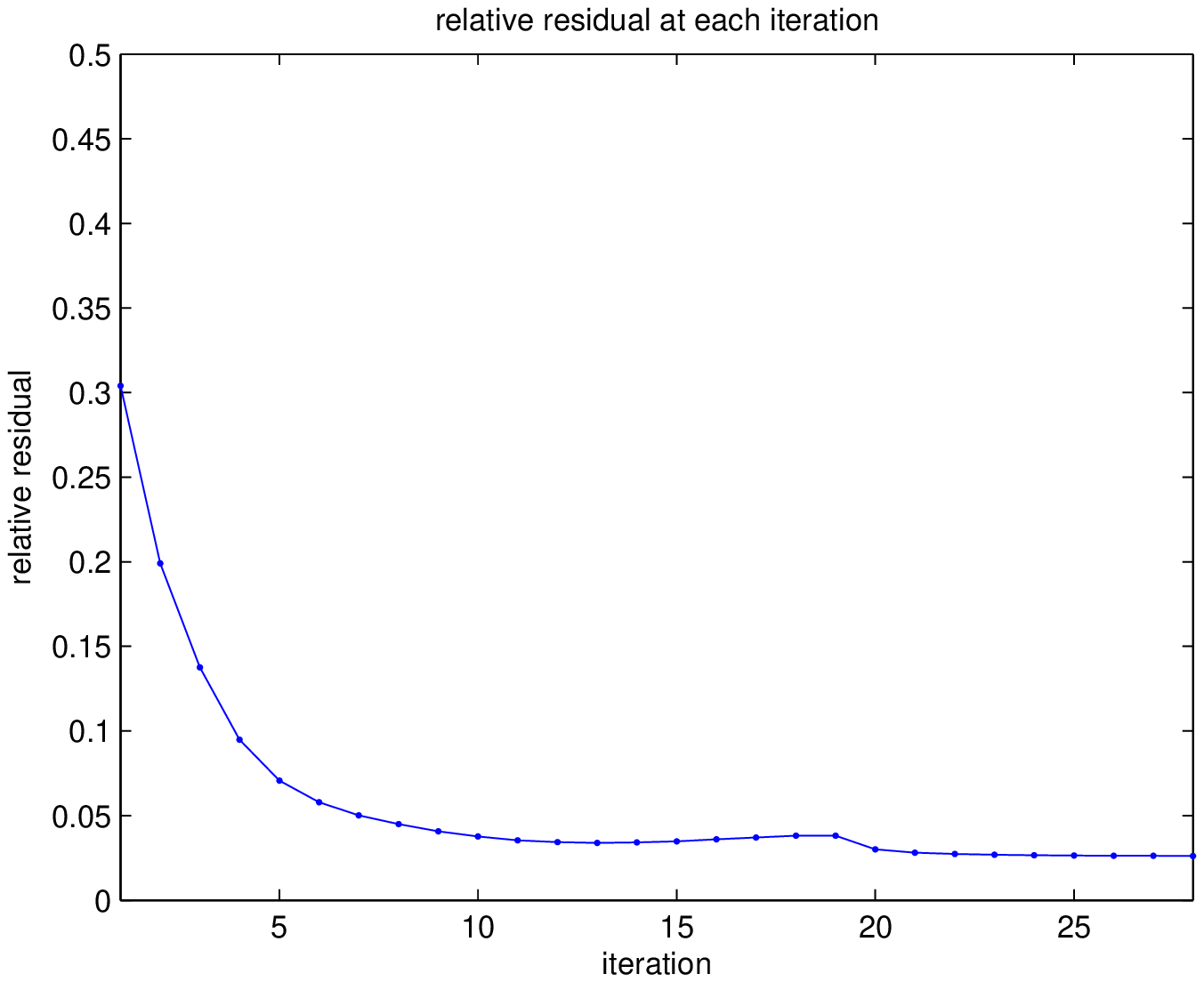}}
        \subfigure[]{
    \includegraphics[width = 3cm]{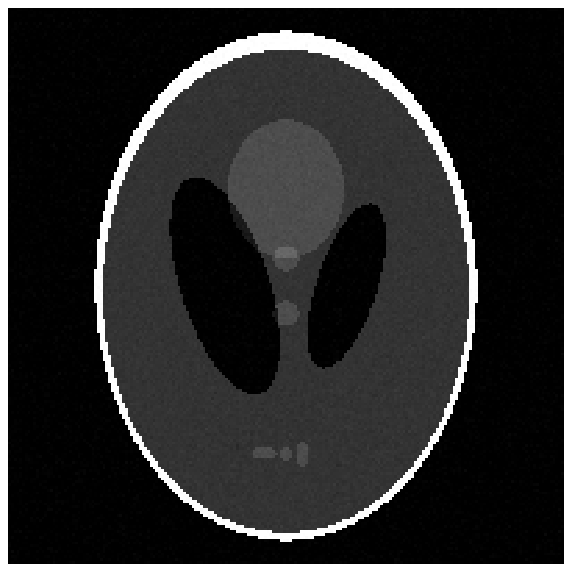}}
          \subfigure[]{
         \includegraphics[width = 4cm, height = 3cm]{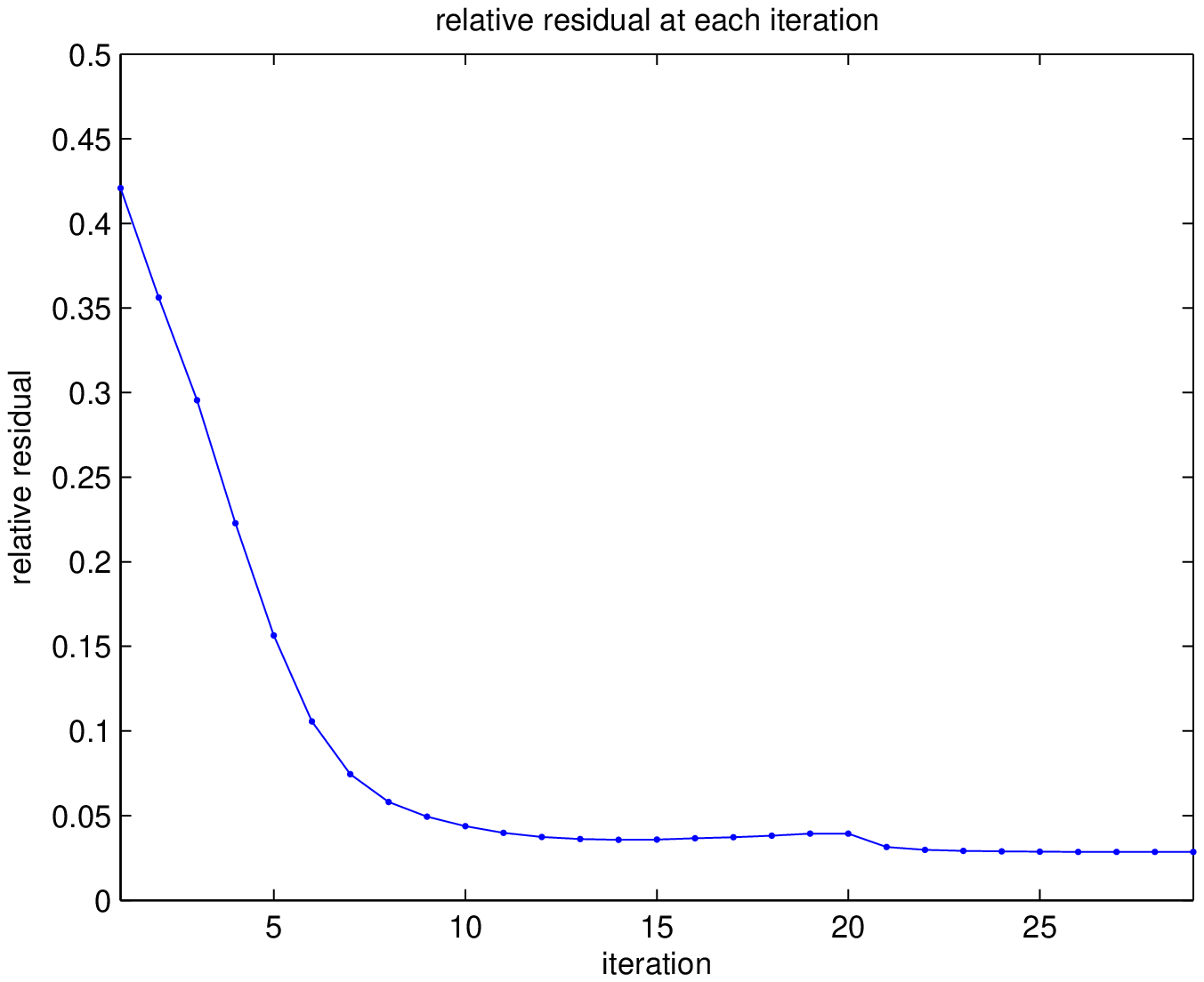}}
         \subfigure[]{
    \includegraphics[width = 3cm]{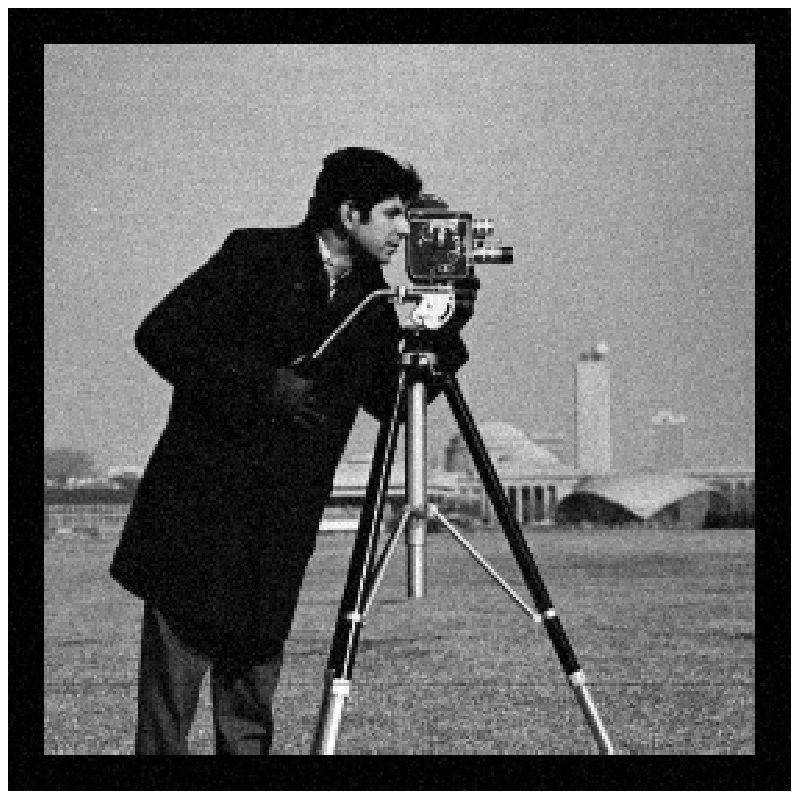}}
          \subfigure[]{
         \includegraphics[width = 4cm, height = 3cm]{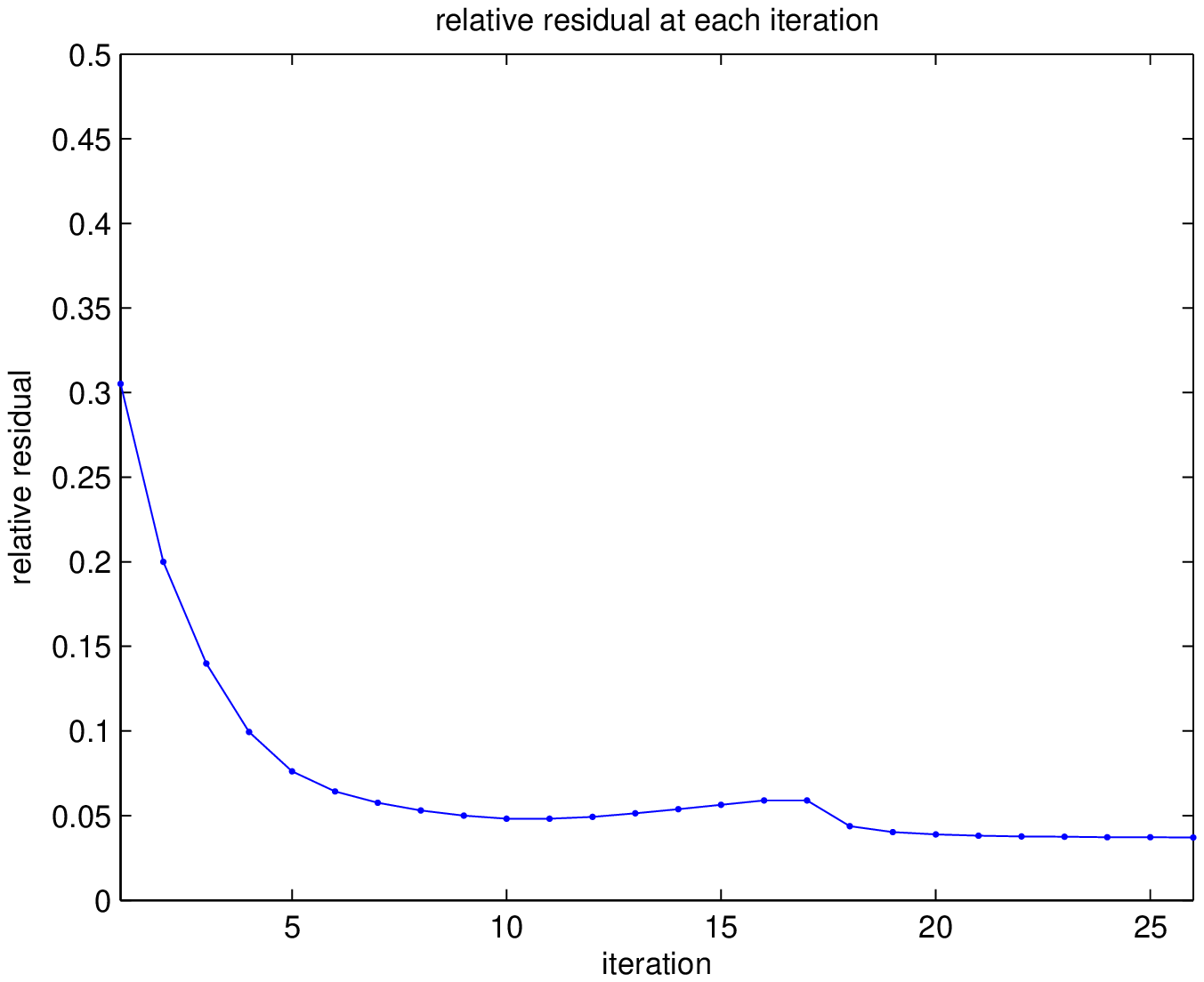}}
        \subfigure[]{
    \includegraphics[width = 3cm]{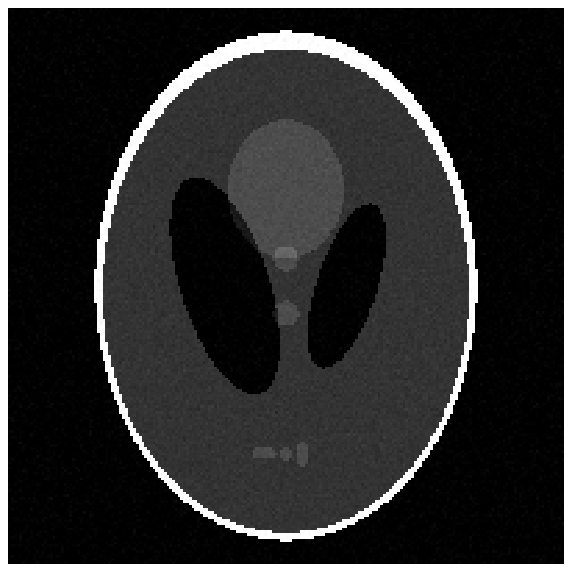}}
          \subfigure[]{
         \includegraphics[width = 4cm, height = 3cm]{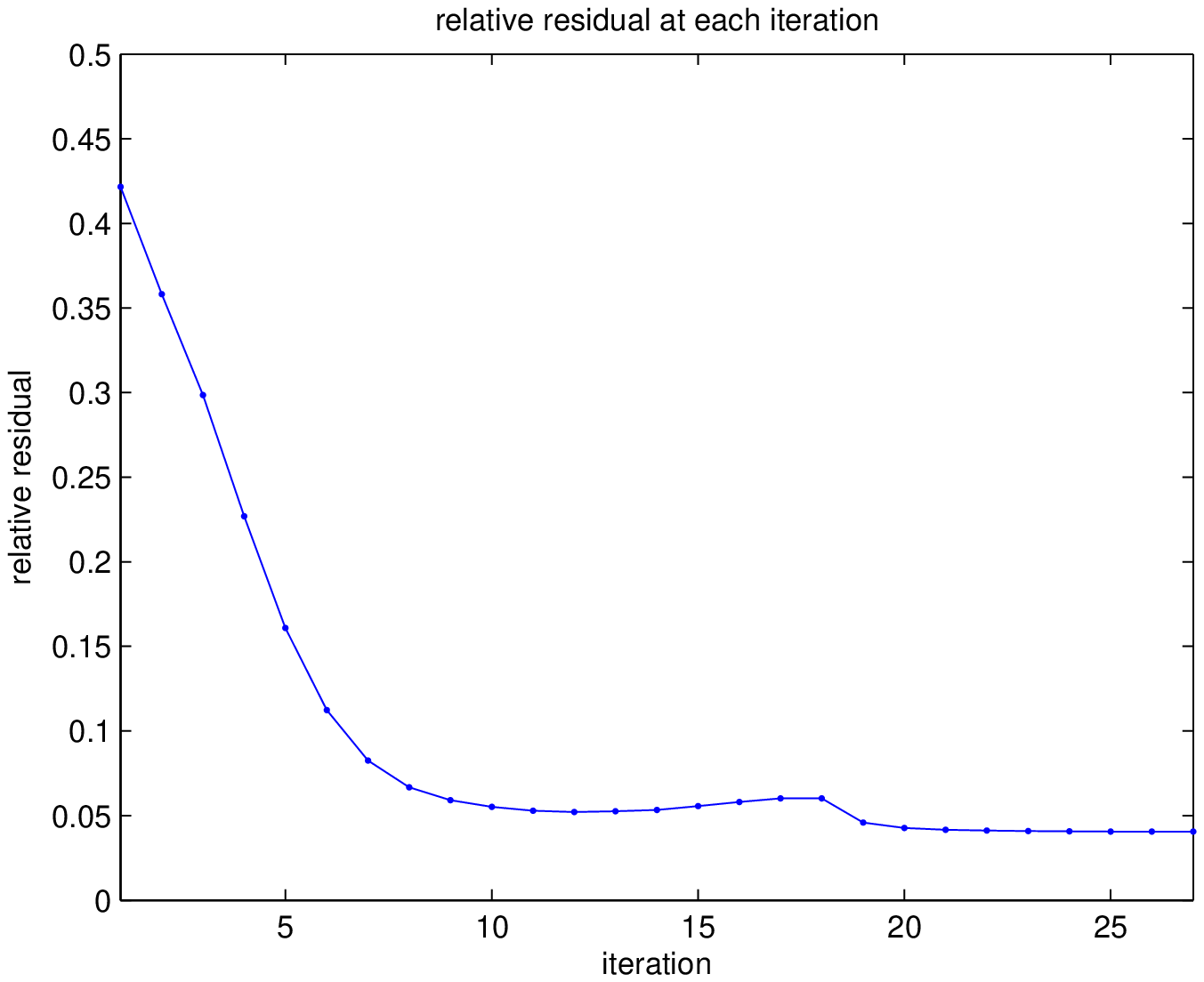}}
          \subfigure[]{
    \includegraphics[width = 3cm]{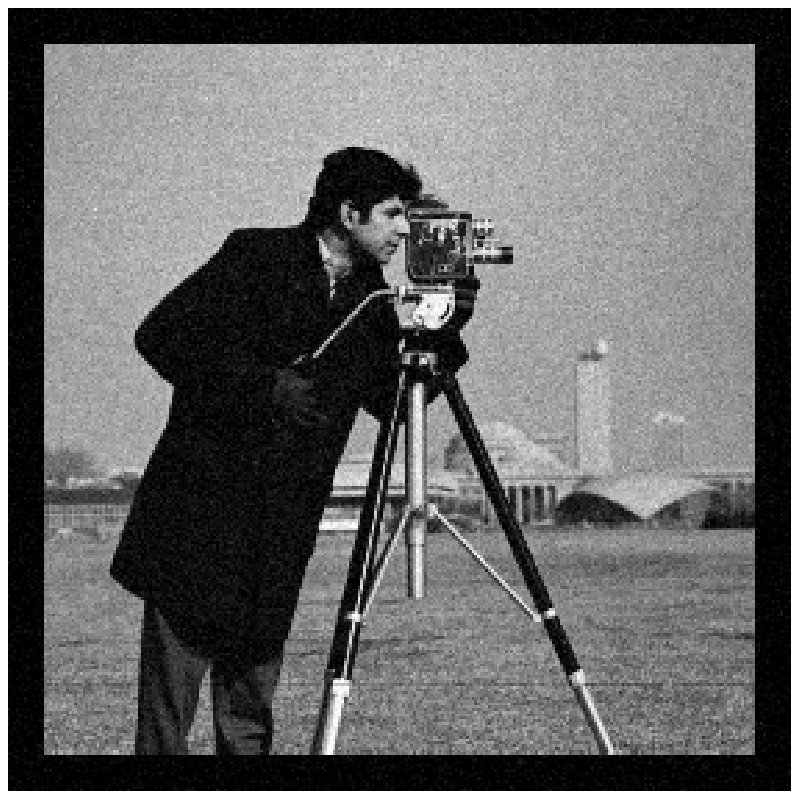}}
          \subfigure[]{
         \includegraphics[width = 4cm, height = 3cm]{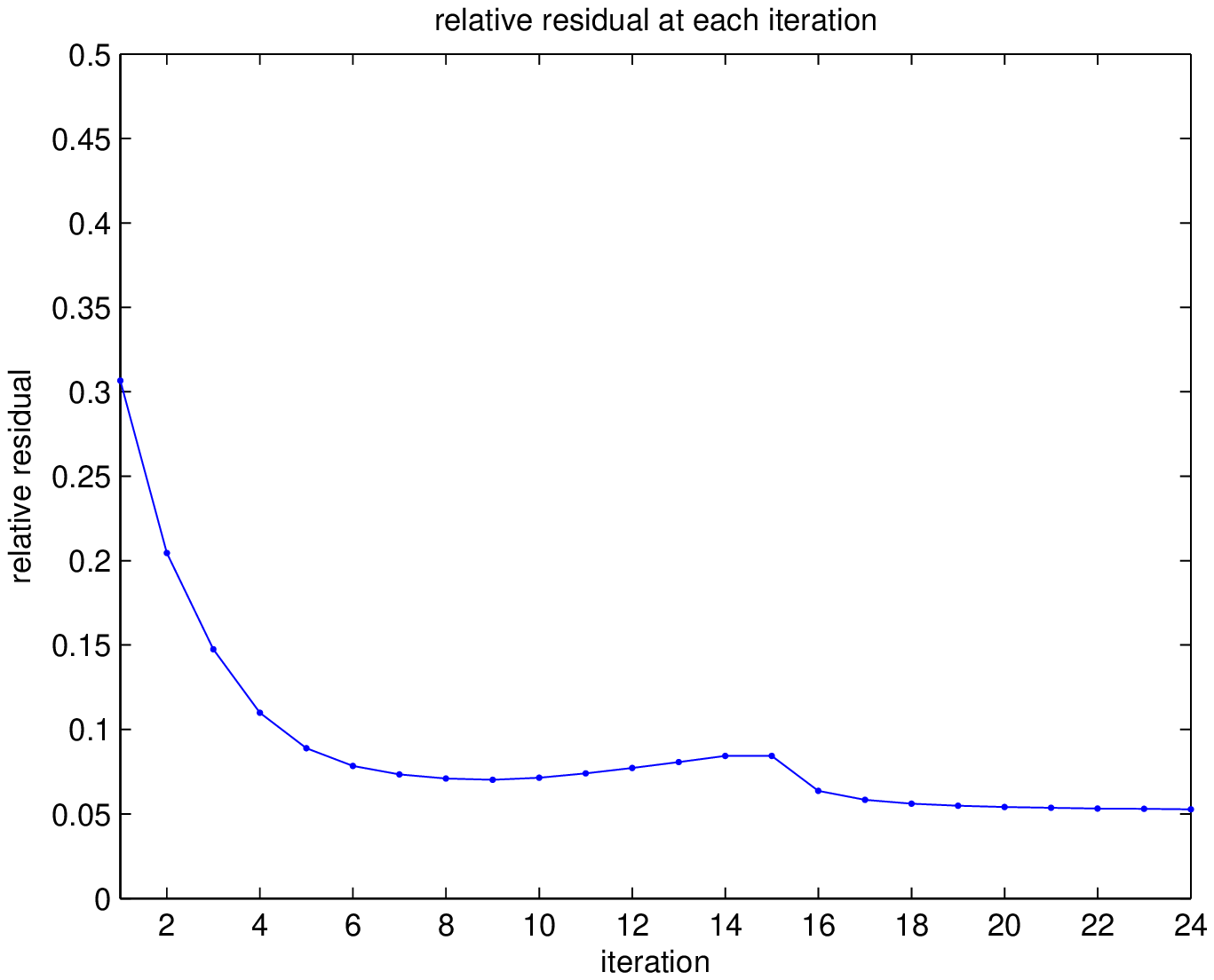}}
        \subfigure[]{
    \includegraphics[width = 3cm]{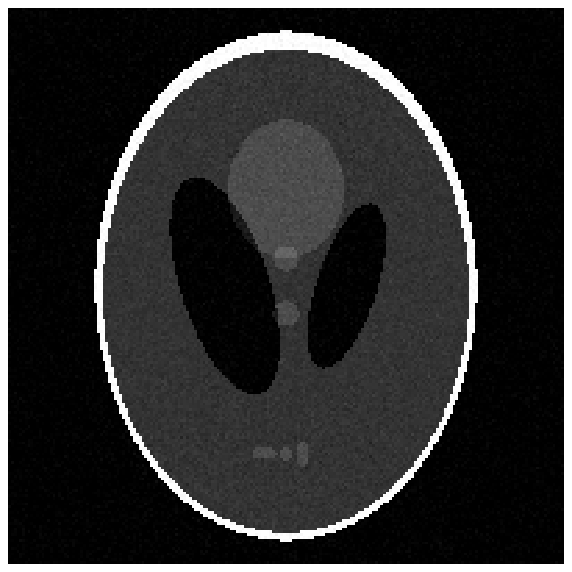}}
          \subfigure[]{
         \includegraphics[width = 4cm, height = 3cm]{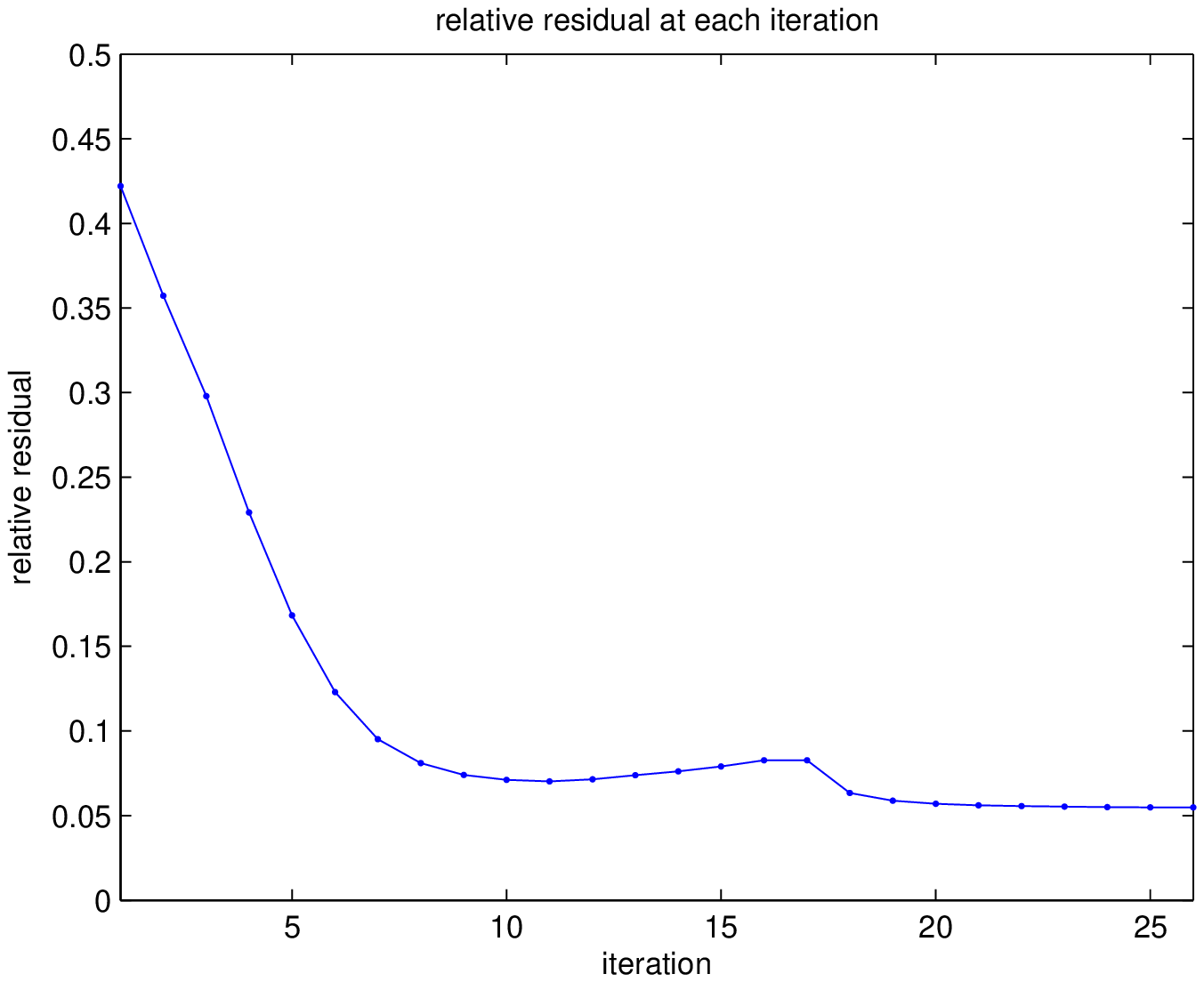}}
           \caption{ Phasing  with $\sigma=2$ and one high resolution RPI: (a) Recovery   by $18$ HIO $+ 10$ ER with  $5\%$ Gaussian noise;  (b) $r(f_k)$ versus k with $r(\hat{f}) \approx 2.62\%$
           and $e(\hat{f}) \approx 4.20\%$;  
           (c) Recovery by $19$ HIO $+ 10$ ER with  $5\%$ Gaussian noise. 
            (d) $r(f_k)$ versus k with $r(\hat{f}) \approx 2.85\%$
            and  $e(\hat{f}) \approx 3.51\%$; 
           (e) Recovery by $16$ HIO $+ 10$ ER with  $5\%$ Poisson noise;
         (f) $r(f_k)$ versus k with $r(\hat{f}) \approx 3.71\%$
         and  $e(\hat{f}) \approx 5.89\%$;            (g) Recovery  by $17$ HIO $+ 10$ ER with $5\%$ Poisson noise;
    (h) $r(f_k)$ versus k with   $r(\hat{f}) \approx 4.05\%$
    and $e(\hat{f}) \approx 4.84\%$; 
           (i) Recovery by $14$ HIO $+ 10$ ER with $5\%$ illuminator noise; (j) $r(f_k)$ versus k with
           $r(\hat{f}) \approx 5.28\%$ and  $e(\hat{f}) \approx 7.75\%$;     (k) Recovery by $16$ HIO $+ 10$ ER  with  $5\%$ illuminator noise;  (l) $r(f_k)$ versus k with
    $r(\hat{f}) \approx 5.48\%$ and  $e(\hat{f}) \approx 6.35\%$.   }
    
     \label{HighNoisy} 
\end{figure}

\subsection{Stability Test}

For images with positivity constraint and with one RPI, we terminate HIO when the relative residual increases  for $5$ consecutive steps and apply  $10$ steps of ER afterward. The maximal HIO iteration is set to be $100$. For
 complex-valued  images with two illuminations, we 
 apply 200 steps of HIO and 300 steps of  ER. 

Figure \ref{HighNoisy} shows the recovery for the nonnegative-valued images with one high resolution RPI 
and  $5\%$ Gaussian ((a)-(d)), Poisson ((e)-(h)) and illuminator noise ((i)-(l)). Multiplicative noise such
as Poisson and illumination noises are generally
more debilitating than the additive Gaussian noise.

With a low resolution RPI (block size: $40\times 40$), 
the quality of reconstruction suffers slightly as shown in
 Figure \ref{LowNoisy} for nonnegative-valued  images. 
 The deterioration is most visible in the case of
 Poisson noise with the blocky pattern  in Figure \ref{LowNoisy}(e) and \ref{LowNoisy}(g).

\begin{figure}[ht]
  \centering
  \subfigure[]{
    \includegraphics[width = 3cm]{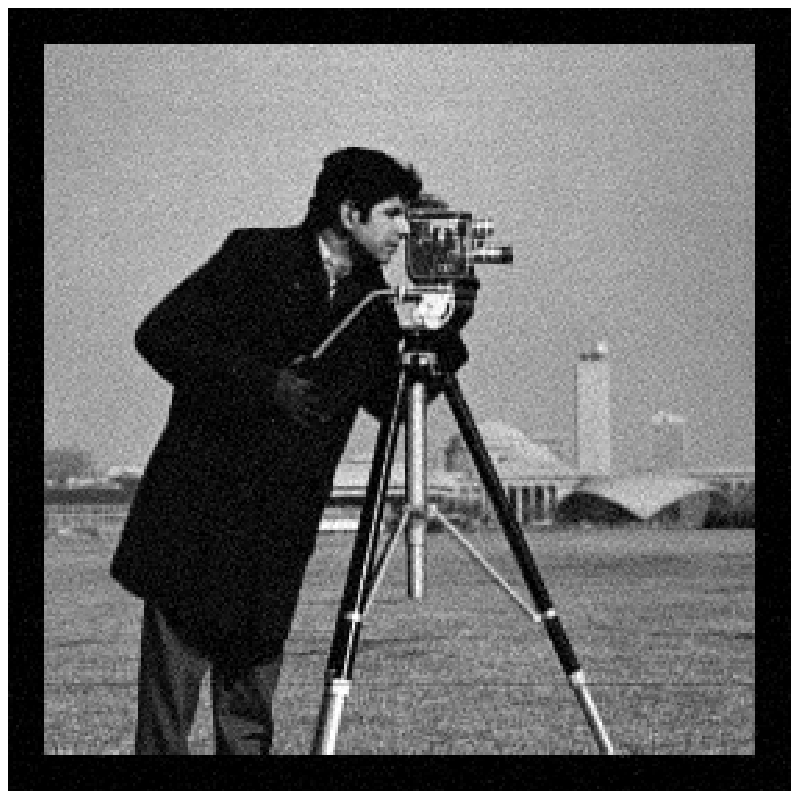}}
          \subfigure[]{
         \includegraphics[width = 4cm, height = 3cm]{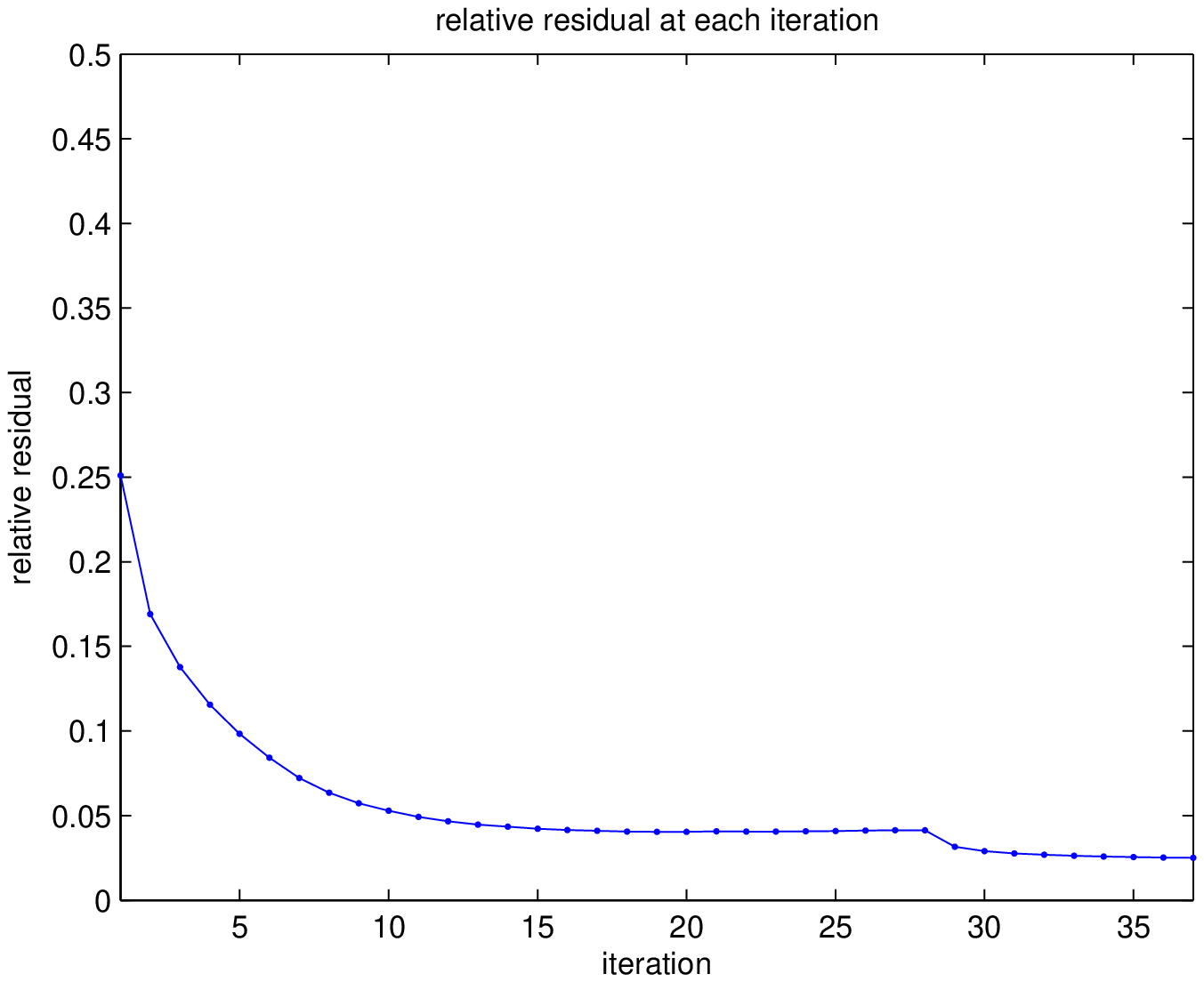}}
        \subfigure[]{
    \includegraphics[width = 3cm]{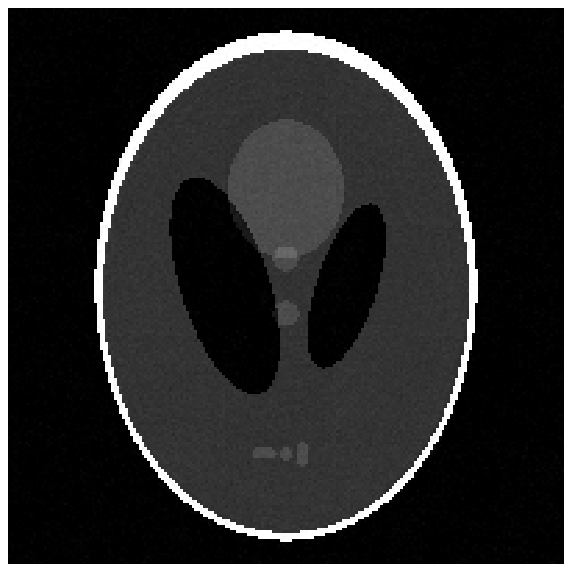}}
          \subfigure[]{
         \includegraphics[width = 4cm, height = 3cm]{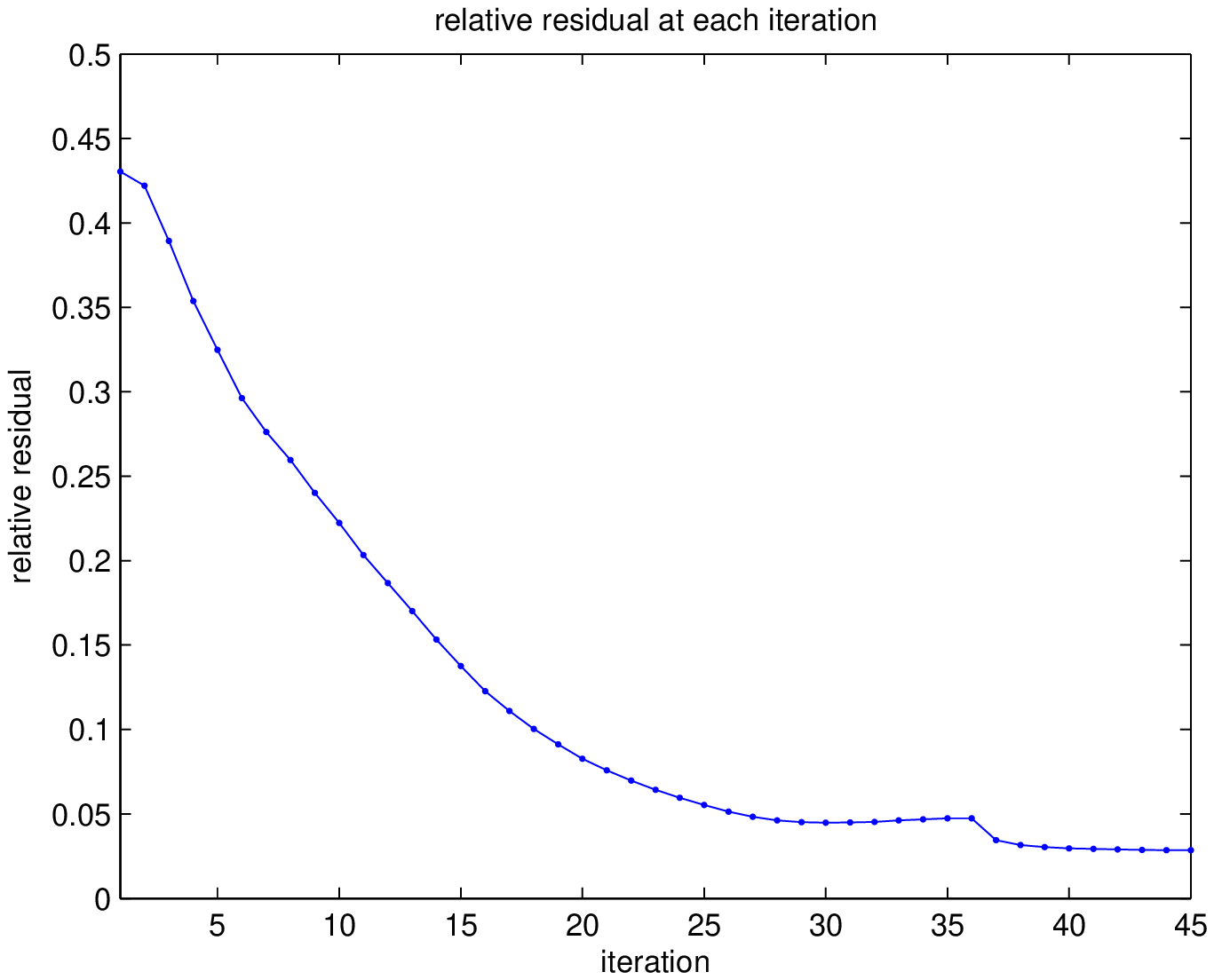}}
         \subfigure[]{
    \includegraphics[width = 3cm]{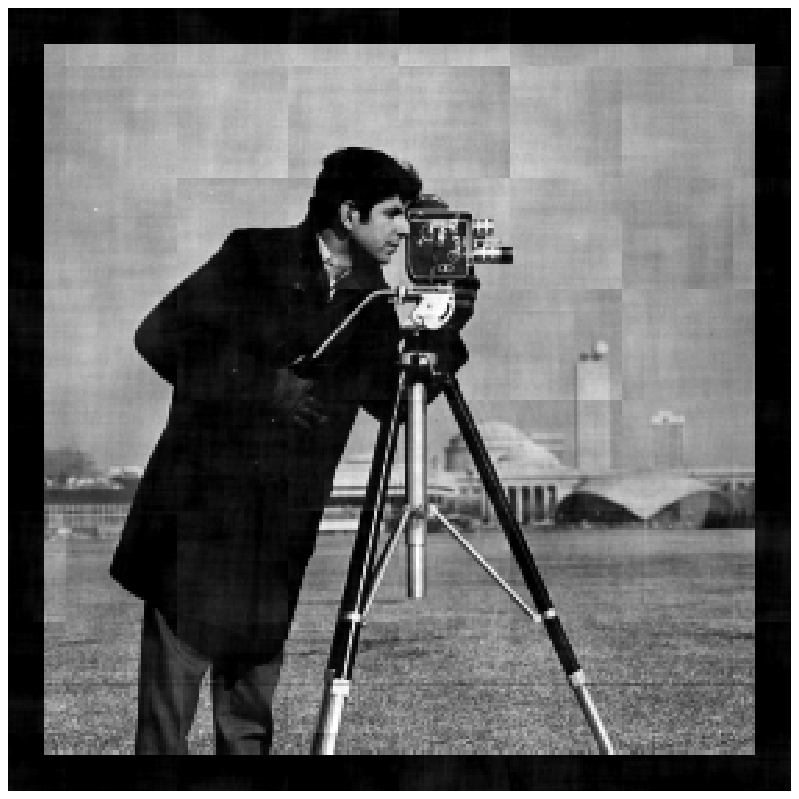}}
          \subfigure[]{
         \includegraphics[width = 4cm, height = 3cm]{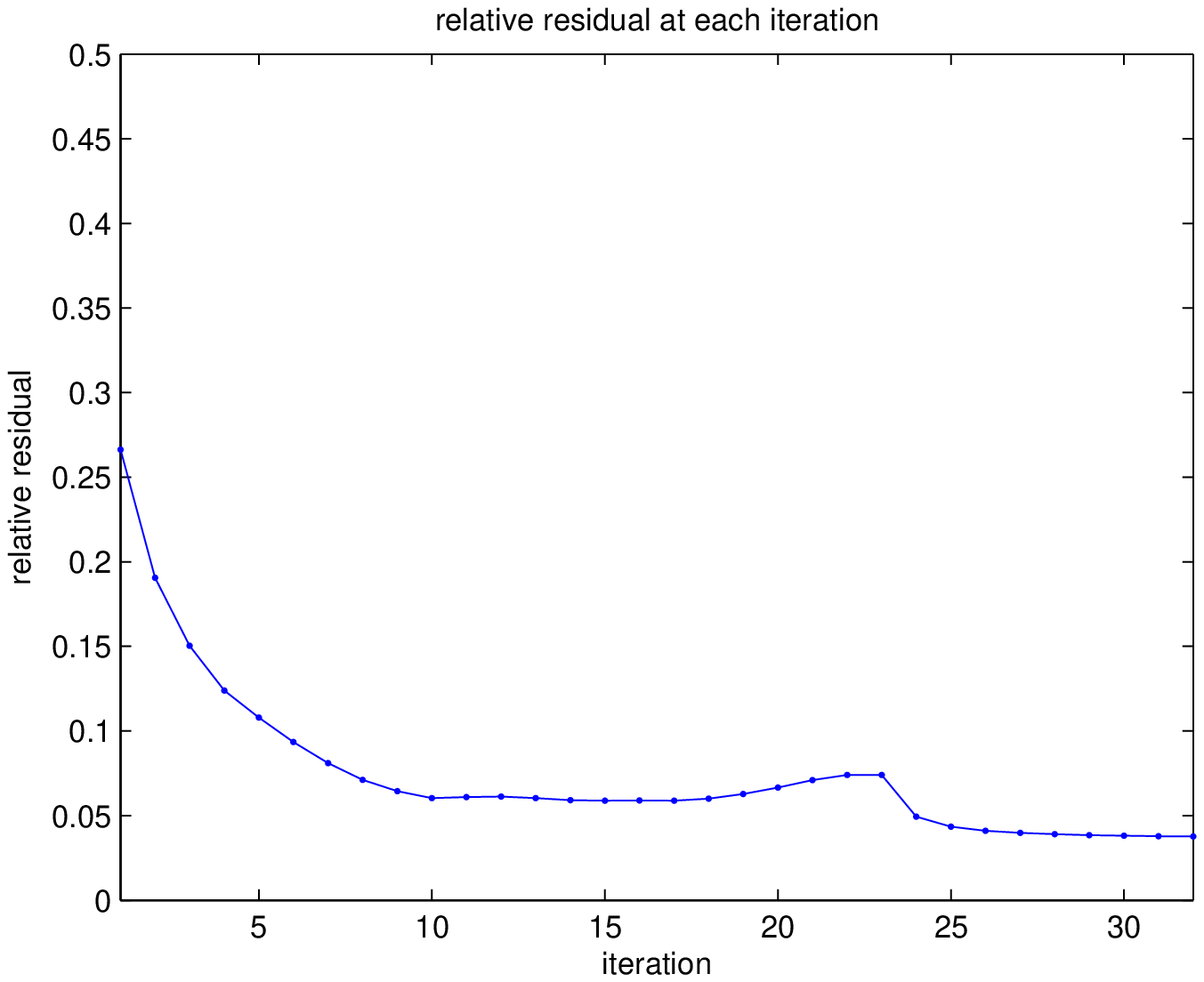}}
        \subfigure[]{
    \includegraphics[width = 3cm]{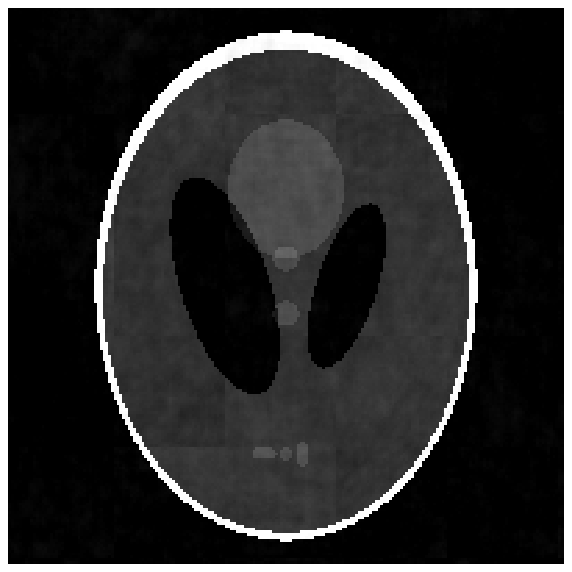}}
          \subfigure[]{
         \includegraphics[width = 4cm, height = 3cm]{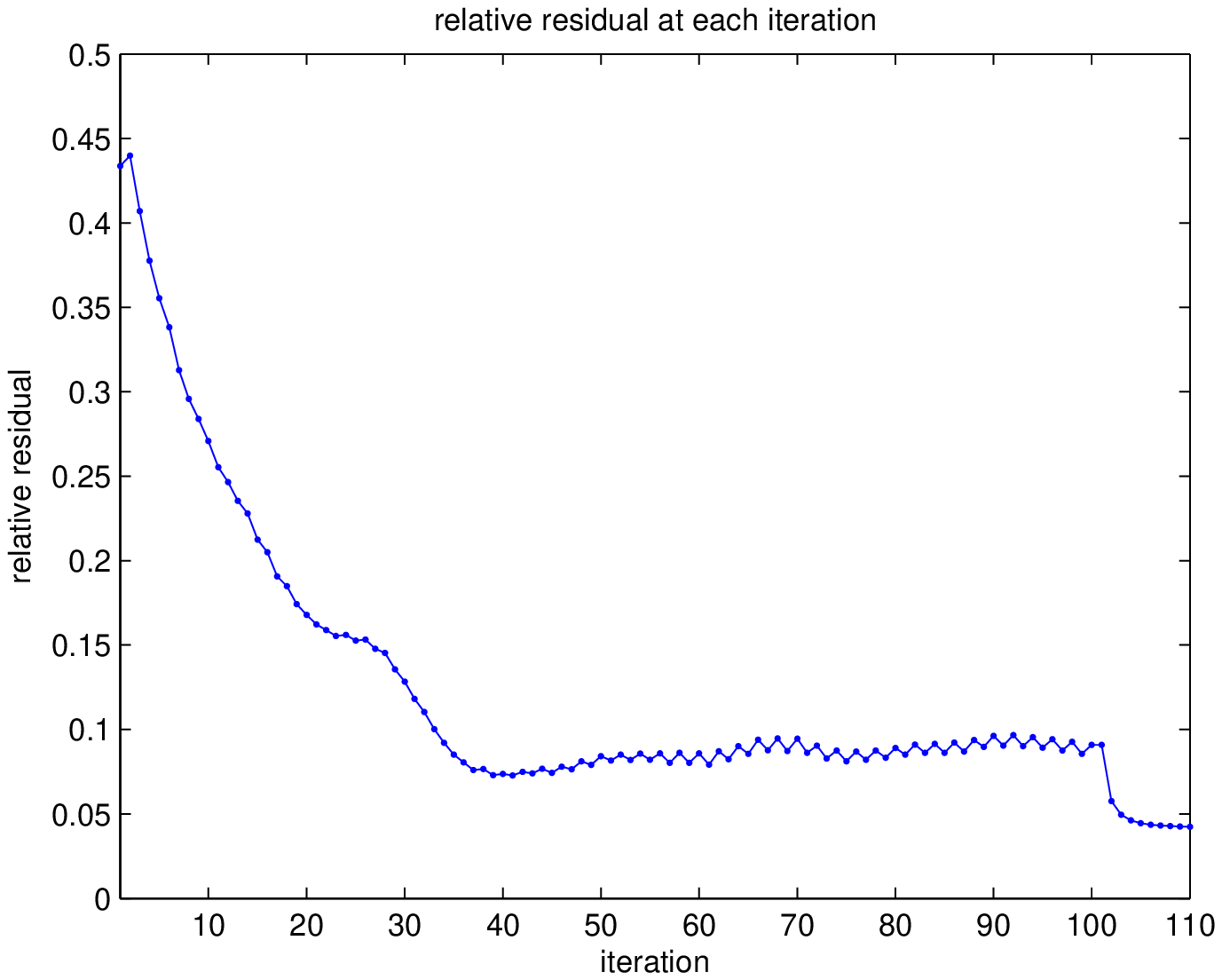}}
          \subfigure[]{
    \includegraphics[width = 3cm]{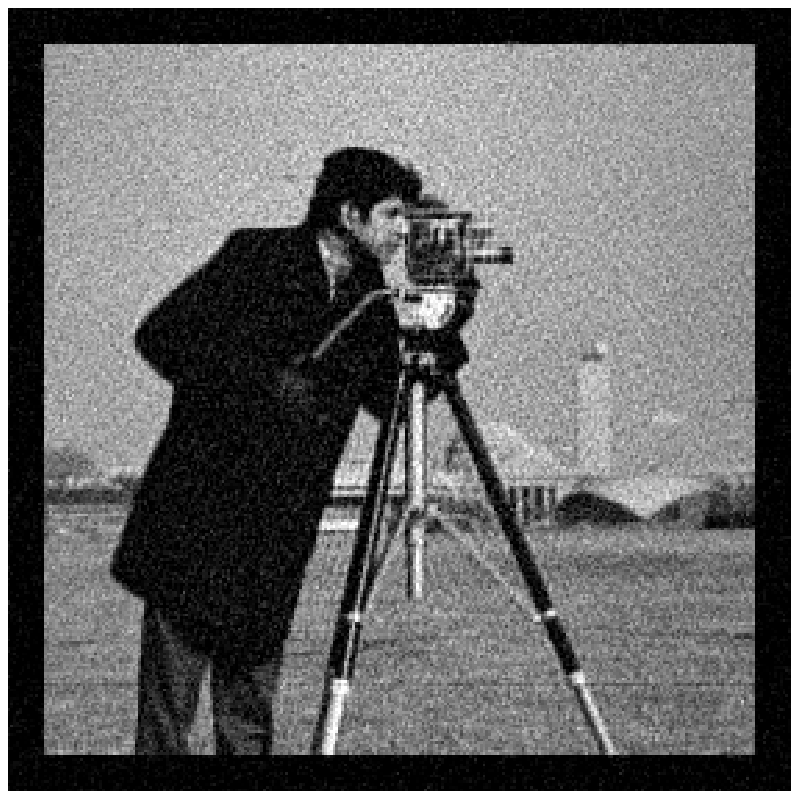}}
          \subfigure[]{
         \includegraphics[width = 4cm, height = 3cm]{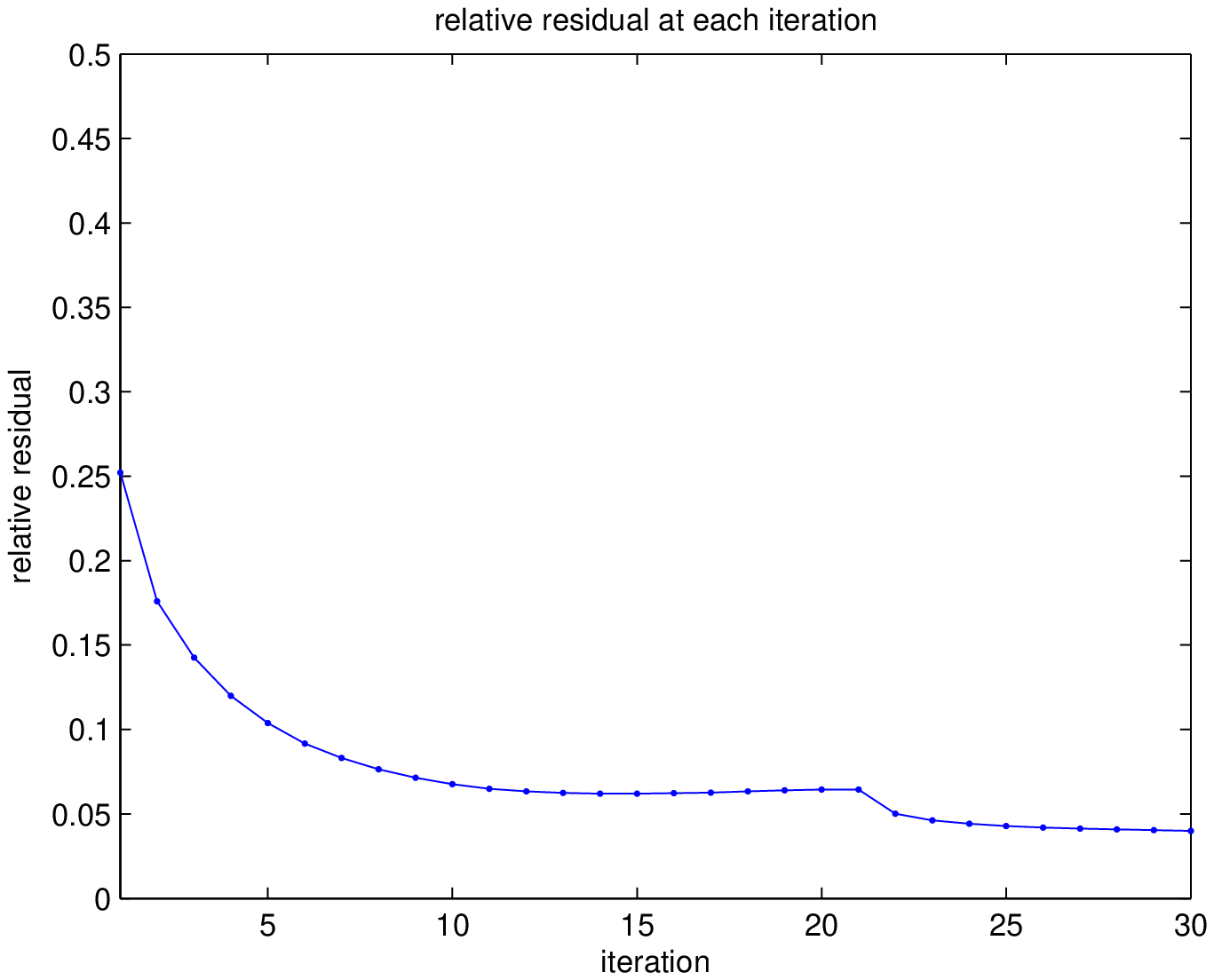}}
        \subfigure[]{
    \includegraphics[width = 3cm]{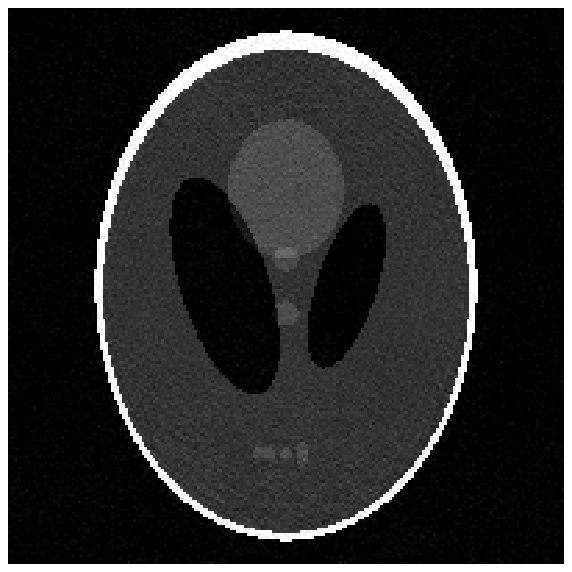}}
          \subfigure[]{
         \includegraphics[width = 4cm, height = 3cm]{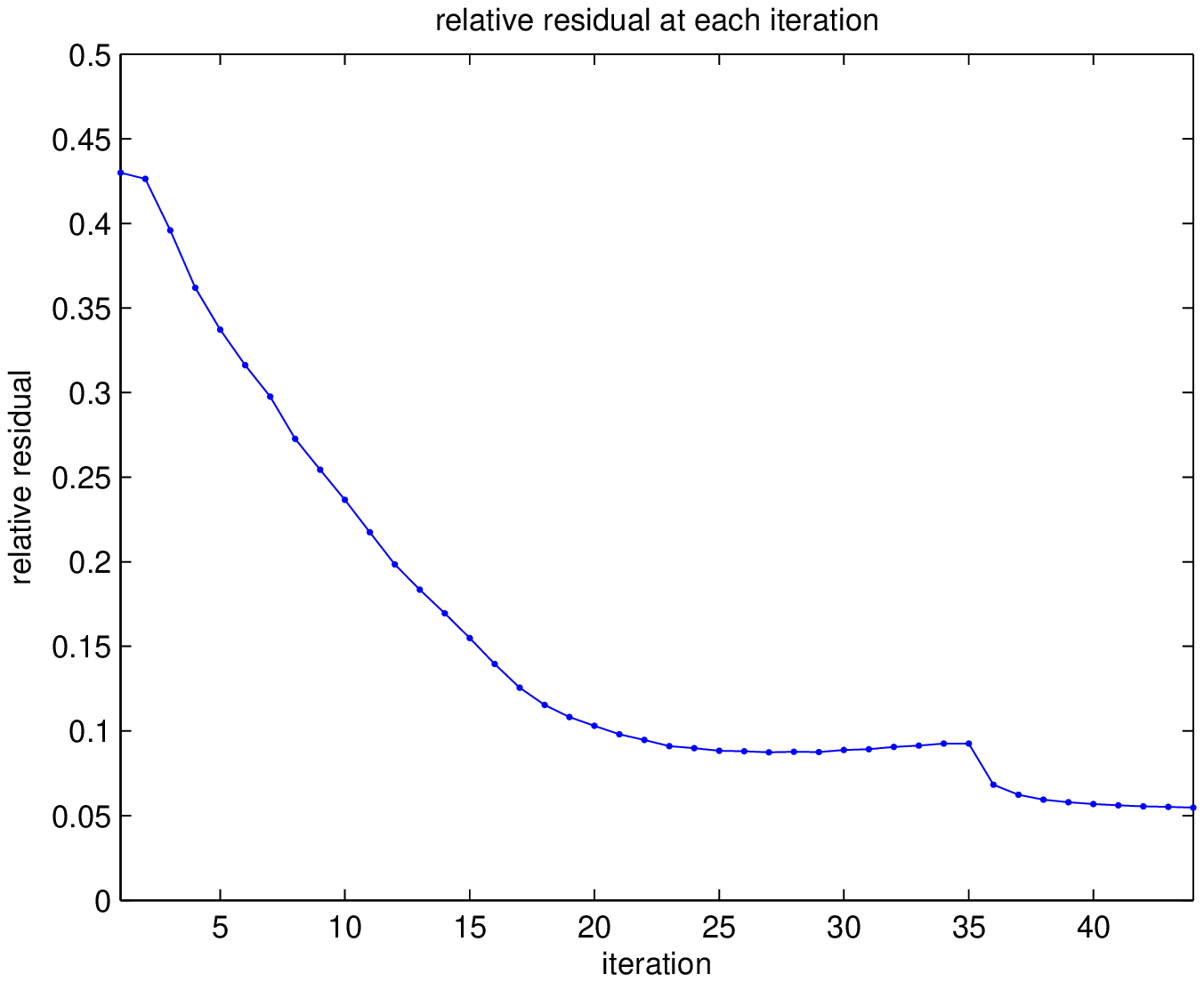}}
           \caption{ Phasing  with $\sigma=2$ and one low ($40\times 40$)  resolution RPI: (a) Recovery   by $27$ HIO $+ 10$ ER with  $5\%$ Gaussian noise;  (b) $r(f_k)$ versus k with $r(\hat{f}) \approx 2.50\%$
           and $e(\hat{f}) \approx 7.37\%$;  
           (c) Recovery by $35$ HIO $+ 10$ ER with  $5\%$ Gaussian noise. 
            (d) $r(f_k)$ versus k with $r(\hat{f}) \approx 2.85\%$
            and  $e(\hat{f}) \approx 4.18\%$; 
           (e) Recovery by $22$ HIO $+ 10$ ER with  $5\%$ Poisson noise;
         (f) $r(f_k)$ versus k with $r(\hat{f}) \approx 3.77\%$
         and  $e(\hat{f}) \approx 6.27\%$;            (g) Recovery  by $100$ HIO $+ 10$ ER with $5\%$ Poisson noise;
    (h) $r(f_k)$ versus k with   $r(\hat{f}) \approx 4.24\%$
    and $e(\hat{f}) \approx 5.09\%$; 
           (i) Recovery by $20$ HIO $+ 10$ ER with $5\%$ illuminator noise; (j) $r(f_k)$ versus k with
           $r(\hat{f}) \approx 4.00\%$ and  $e(\hat{f}) \approx 13.14\%$;     (k) Recovery by $34$ HIO $+ 10$ ER  with  $5\%$ illuminator noise;  (l) $r(f_k)$ versus k with
    $r(\hat{f}) \approx 5.48\%$ and  $e(\hat{f}) \approx 9.46\%$. 
   }
     \label{LowNoisy} 
\end{figure}

  Figure \ref{StabilityTest} shows the average relative error
$e(\hat f)$   versus noise 
  for (a) nonnegative-valued Phantom and $\sigma=2$, (b) Phantom with phases randomly distributed in $[0,\pi/2]$ and $\sigma=4$ and (c) Phantom with  phases randomly distributed in $[0,2\pi]$ and $\sigma=3$. One high or low ($40\times 40$) resolution RPI is used in (a) 
  while  one high or low ($4\times 4$) RPI and one UI  are used in (b) and (c).  The 
  adaptive HIO $+ 50$ ER is used for (a) and (b)
  while $200$ HIO + $300$ ER
  is used for  (c). 
  
Relative  error increases almost linearly with respect to the relative noise level 
 with the noise amplification constant at worst 2.  Clearly the illumination  noise is most debilitating, followed by the Poisson noise. Nevertheless,
the noise stability is achieved with even the low resolution 
RPI for all three types of noise.

\begin{figure}[ht]
  \centering
  \subfigure[]{
    \includegraphics[width = 8cm]{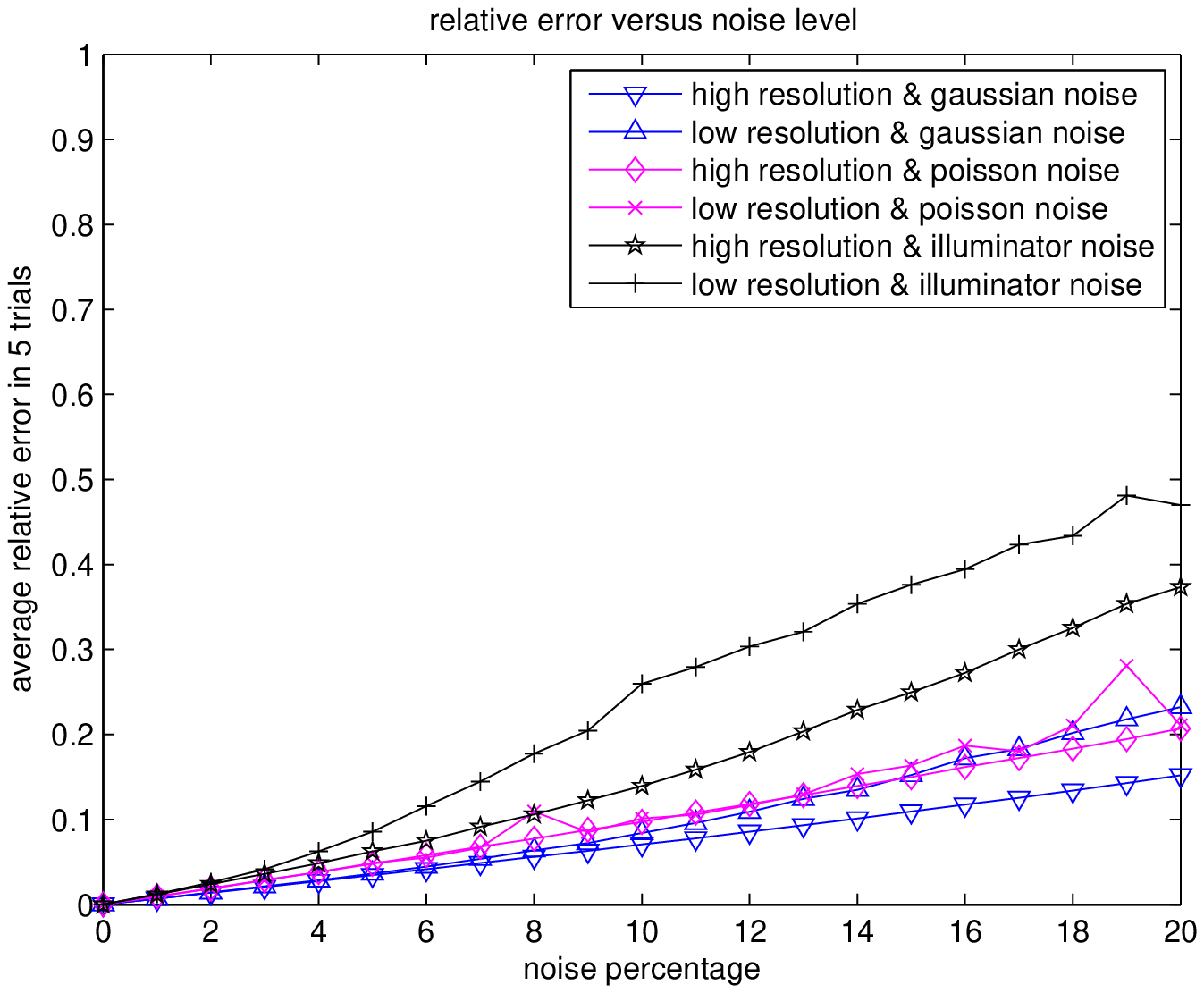}}
      \subfigure[]{
    \includegraphics[width = 8cm]{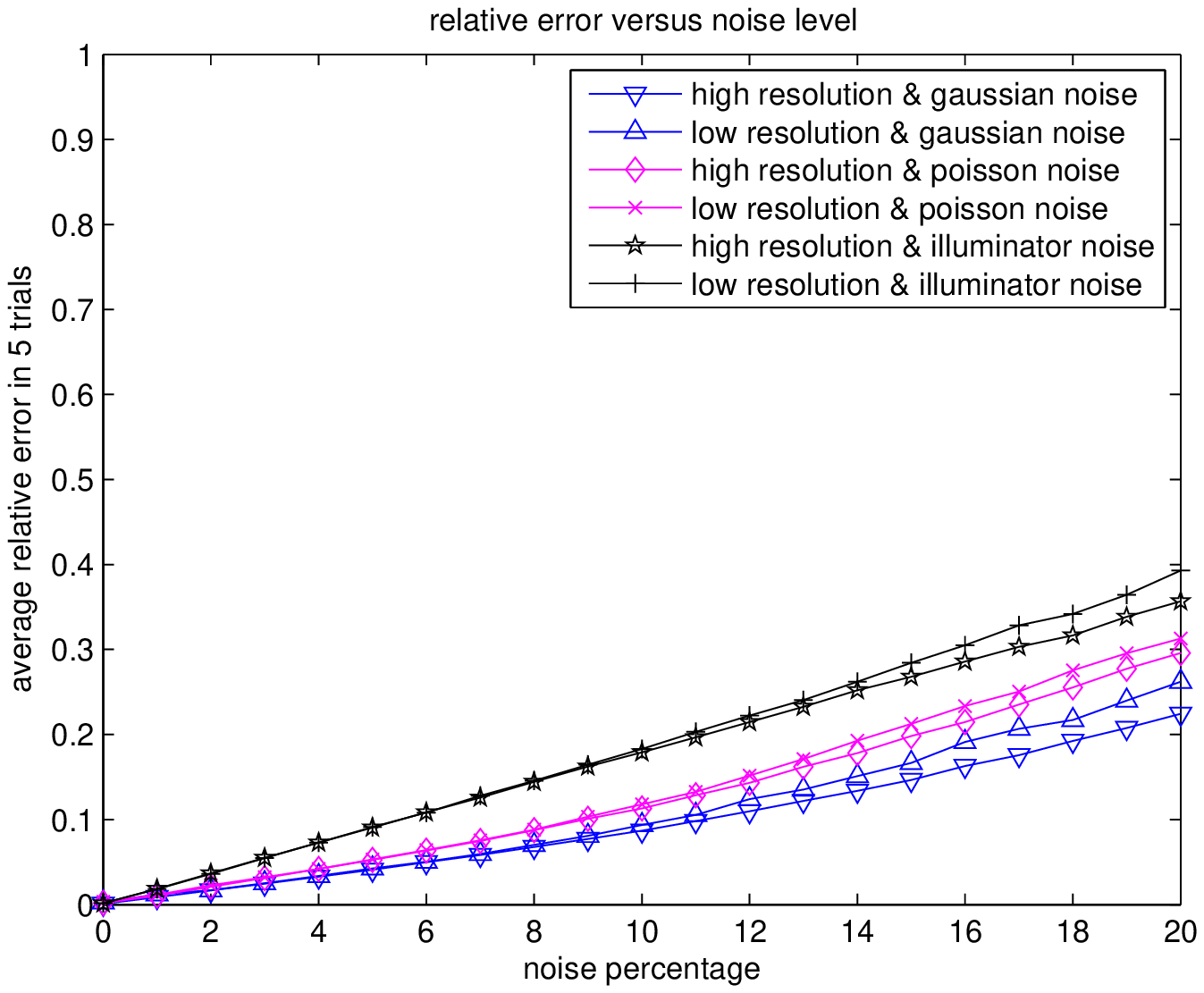}}
 \subfigure[]{
    \includegraphics[width = 8cm]{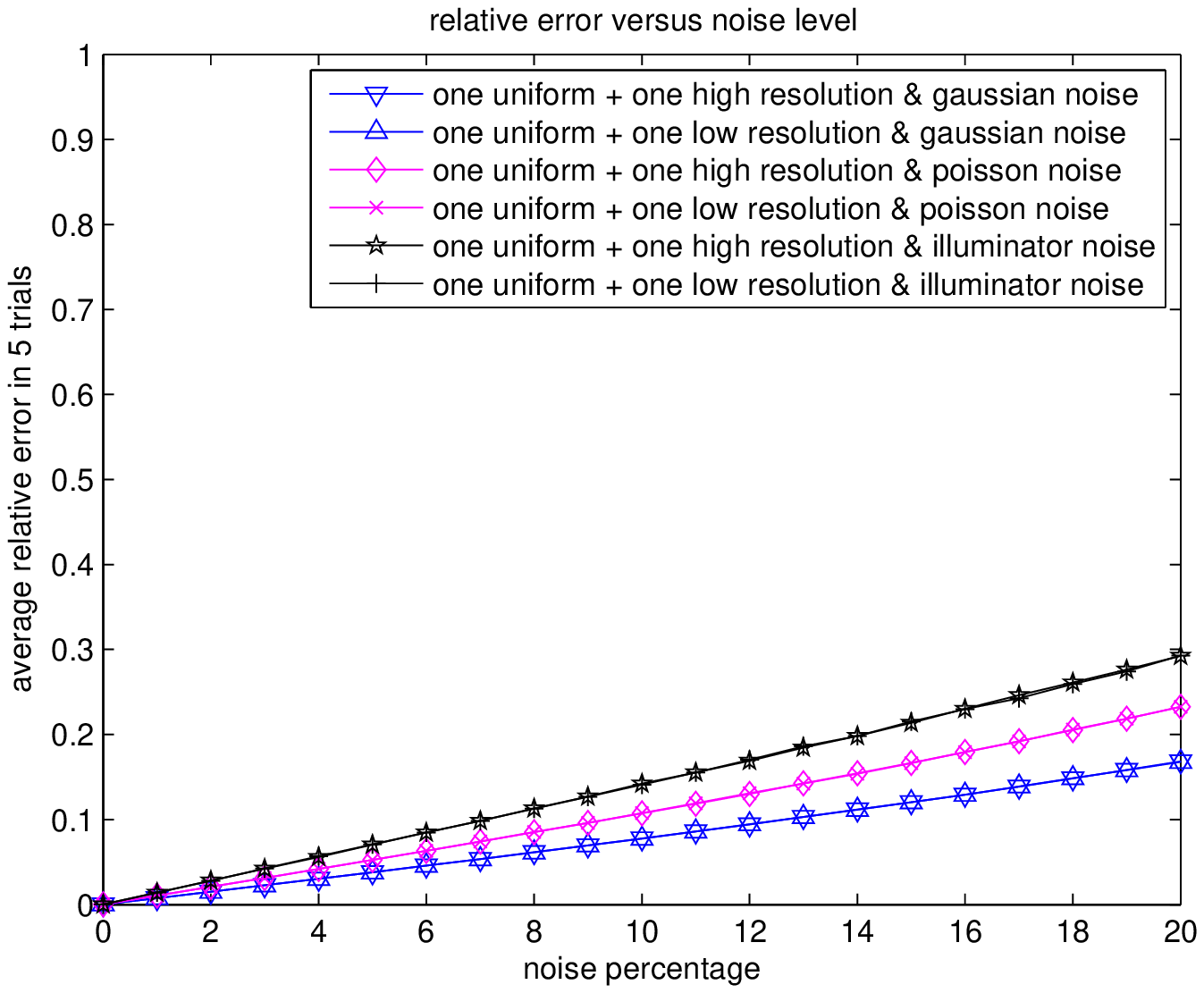}}
  \caption{(a) Relative error for nonnegative-valued Phantom and $\sigma=2$  (b) Relative error for complex-valued Phantom with phases randomly distributed in  $[0, \pi/2]$ and
  $\sigma = 4$;  (c) Relative error  for complex-valued Phantom with phases 
  randomly distributed in $[0, 2\pi]$ and $\sigma = 3$. }
    \label{StabilityTest} 
\end{figure}

\section{Conclusion}
\label{sec5}
We have given a proof of convergence of ER (Theorem \ref{FixedPoint1})  and demonstrated that  the stagnation problem of
standard phasing algorithms such as ER and HIO can
be alleviated if the ambiguities associated with spatial translation and conjugate inversion are removed by
 RPI. In addition, phasing with RPI has the following
 advantages: (i) It is stable
 with respect to additive as well as  multiplicative noises
 with a moderate  noise amplification constant;
 (ii) It reduces the oversampling ratio by more than
 a factor of 2; (iii) It reduces the number of iterations
 by more than an order of magnitude. 
 We have also shown that phasing with RPI 
 performs well with low resolution illumination and 
 can tolerate a high level of illumination error,  adding 
 assurance that the random illumination 
 needs not be calibrated exactly.  
  
The lower bound $\sigma \geq 2$ for phasing of \cite{MSC} was never actually achieved 
but we have  achieved the lower limit in phasing with RPI for complex-valued images under a sector constraint. For 
nonnegative-valued images, phasing with 
one high resolution RPI reduces the oversampling ratio  to unity, the minimum level  by the dimensional count. 

\commentout{
On the issue of convexity, there have been recent attempts to formulate phase retrieval as the convex, trace-norm  minimization by significantly  lifting the dimension of the problem \cite{Candes1} \cite{Candes2}. The notion of absolute uniqueness is essential  in the trace-norm minimization approach which hinges on the equivalence between the  rank and trace norm minimization. 
Otherwise,  if both $f$ and its conjugate inversion or spatial shift $\tilde f$ are solutions of phase retrieval, then all convex combinations, $\lambda f f^{\rm H} + (1-\lambda) \tilde f\tilde f^{\rm H} $, $\lambda \in (0,1)$ 
have  the same trace and yet, typically rank two.
These rank two minimizers can not be the tensor
product of the true image. To remove these ambiguities it
is necessary to  ensure uniqueness up to a global phase.
}

\vspace{1cm}

\centerline{\Large\bf Appendices}


\appendix

\commentout{
\begin{proposition}
Let D be a closed convex subset of $\cC(\cN)$ and $\mathcal{P}_D$ be the orthogonal projection onto D. Then for every pair of elements $x$ and $y$ in $\cC(\cN)$,
$$\|\mathcal{P}_D\{x\} - \mathcal{P}_D\{y\} \| \le \|x-y\|.$$
\end{proposition}
Projection operators onto closed convex sets are nonexpansive and therefore uniformly continuous.\\
}

 \section{Proof of Proposition \ref{accumulation}} 
\begin{proof}
By Proposition \ref{lemma1}, $\lim_{k\to \infty}\vep_f(f_k) =\eta$, for some $\eta\geq 0$. Since 
 $f_{k+1} = \PO \{f'_k\}$,  we have $\|f_{k+1}\|\leq \|f'_k\| = \|G'_k\| = \|Y\|$ and  that $\{f_k\}$ is a bounded sequence. Every bounded sequence in $\cC(\cN)$ has a convergent subsequence, so $\{f_k\}$ has at least one convergent subsequence. Without loss of generality, we assume $\lim_{k\to\infty} f_k =f^\star$. Next, we prove that $f^{\star}$ must be a fixed point of $\PO\PF$ or $\PO\PF^{\theta}$ for some $\theta$.

Since $\Phi$ and $\Lambda$ are unitary matrices, $\lim_{k\to\infty} \Phi \Lambda f_k = \Phi \Lambda f^\star$ in $\|\cdot\|$,
and thus $$\lim_{k\to\infty} \Phi \Lambda f_k(\bom) = \Phi \Lambda f^\star(\bom), \ \forall \bom.$$

\begin{itemize}

\item If $\Phi\Lambda f^\star(\bom)$ vanishes nowhere in $\cL$, then 
\beqn
\lim_{k\to \infty} \measuredangle{\Phi\Lambda f_k} =\measuredangle{\Phi\Lambda f^\star} 
\eeqn
implying 
\beqn
\lim_{k\to \infty}G_k'=\lim_{k\to \infty} \PT \Phi \Lambda f_k= \PT \Phi \Lambda f^\star.
\eeqn
Therefore, 
\beqn
\lim_{k\to\infty} f_{k+1} = \PO\PF f^\star
\eeqn
which along with the convergence of $\vep_f(f_k)$ and $f_k$
implies that  
$\vep_f(f^\star) = \vep_f(\PO\PF f^\star)$. By Proposition \ref{lemma1}, we have
$$f^\star = \PO\PF f^\star.$$

\item If $\Phi \Lambda f^\star (\bom) = 0$ at some $\bom \in \cL$, $\PT \Phi \Lambda f_k(\bom)$ may not converge. However, since $\PT \Phi \Lambda f_k$ is bounded in view of $\|\PT \Phi \Lambda f_k\| = \|Y\|$,  there exists a subsequence $\{ f_{k_j}\}$ and some $\theta(f^\star)$ such that $\lim_{k\to\infty} \PT\Phi\Lambda f_{k_j}(\bom) =\PT^\theta \Phi \Lambda f^\star(\bom)$ where $\PT^\theta$ is defined in (\ref{6}). Therefore
\beq
\lim_{j\to\infty}\PO \Lambda^{-1} \Phi^{-1 } \PT \Phi \Lambda f_{k_j} =\PO \Lambda^{-1} \Phi^{-1}\PT^\theta \Phi \Lambda f^\star,
\eeq
namely
$$\lim_{j\to \infty} f_{k_j+1} = \lim_{j\to\infty}
\PO\PF f_{k_j} = \PO\PF^{\theta} f^\star$$
which along with the convergence of $\ep_f(f_k)$ and $f_k$ implies that 
$\vep_f(f^\star) = \vep_f( \PO\mathcal{P}_f^\theta f^\star)$. By Proposition \ref{lemma1}, it follows that
$$f^\star = \PO\mathcal{P}_f^\theta f^\star.$$
\end{itemize}
\end{proof}

\section{Proof of Theorem \ref{FixedPoint1}}
 Define
 \[
 f_{\mathbf{m}+}(\cdot) = f(\mathbf{m}  + \cdot), \quad  f_{\mathbf{m}-}(\cdot) = f(\mathbf{m} - \cdot). 
 \]

Let 
$$F(\bz) = \sum_{\bn} f(\bn) \bz^{-\bn}$$
be the $z$-transform of $f$.
According to the fundamental theorem of algebra, $F(\bz)$ can   be written uniquely as
$$F(\bz) = \alpha \bz^{-\bn_0} \prod_{k=1}^p F_k(\bz),$$
where $\bn_0$ is a vector of nonnegative integers, $\alpha$ is a complex coefficient, and $F_k(\bz)$ are nontrivial irreducible monic polynomials in $\bz^{-1}$.

\begin{definition}[Conjugate Symmetry]
A polynomial $X(\bz)$ in $\bz^{-1}$ is said to be conjugate symmetric if, for some vector $\bk$ of positive integers and some $\theta \in [0,2\pi)$, 
\begin{equation*}
X(\bz) = e^{i \theta} \bz^{-\bk} \overline{X(\bar{\bz}^{-1})}. 
\end{equation*}
\end{definition}
A conjugate symmetric polynomial may be reducible, irreducible, trivial, or nontrivial. If $A(\bz)$ is an arbitrary polynomial in $\bz^{-1}$, then 
$$X(\bz) = A(\bz) \cdot \bz^{-\bN} \overline{A(\bar{\bz}^{-1})}$$
is conjugate symmetric. Any  monomial $a \bz^{\bk}$ is conjugate symmetric.

The uniqueness of recovering  a real-valued object from its Fourier magnitude or phase only data is discussed in \cite{Hayes}
and can be easily generalized  to the case of complex-valued objects.

\begin{proposition}
Let $f(\bn) \in \mathcal{C}(\cN )$ be a finite array whose $z$-transform is irreducible up to a power of $\bz^{-1}$. If the Fourier transform  $G$ of $g(\bn) \in \mathcal{C}(\cN )$ satisfies $|G(e^{i2pi\bom})|=|F(e^{i2\pi\bom})|, \forall \bom\in \mathcal{L}$, then $\exists\  \theta \in [0,2\pi)$ and $\mathbf{m}$ such that either $g = e^{i\theta} f_{\mathbf{m}+}$ or $g = e^{i\theta}\overline{f_{\mathbf{m}-}}$.
\label{HayesMagnitude}
\end{proposition}

\begin{proposition}
Let $f\in \mathcal{C}(\cN )$ be a finite array whose $z$-transform has no nontrivial  conjugate symmetric factors. If $g \in \mathcal{C}(\cN )$ satisfies $\measuredangle{F(e^{2\pi i \bom})} = \measuredangle{G(e^{2\pi i \bom})}, \forall\bom\in \mathcal{L}$, then $g = \beta f$ for some real positive number $\beta$.
\label{HayesPhase}
\end{proposition}
\begin{proof}
Consider the array $h$ defined by 
$$h(\bn) = f(\bn) \ \star \  \overline{g(-\bn)} $$
whose $z$-transform is 
$$H(\bz) = F(\bz)\overline{G(\bar{\bz}^{-1})}.$$
Since the phase of the Fourier transform of $h(\bn)$ is equal to 
$$\measuredangle{H(e^{2\pi i\bom})} = \measuredangle{F(e^{2\pi i\bom})} - \measuredangle{G(e^{2\pi i\bom})},$$
it follows that if $\measuredangle{F(e^{2\pi i\bom})} = \measuredangle{G(e^{2\pi i\bom})}$, then $\measuredangle{H(e^{2\pi i\bom})}=0$. Thus the Fourier transform of $h$ is real-valued, implying  that
$$H(\bz) = \overline{H(\bar{\bz}^{-1})}.$$
Therefore,
\begin{equation}
F(\bz)\overline{G(\bar{\bz}^{-1})} = \overline{F(\bar{\bz}^{-1})}G(\bz).
\label{temeq1}
\end{equation}
Multiplying both sides of \eqref{temeq1} by $\bz^{-\bN}$ results in the following polynomial equation in $\bz^{-1}$:
\begin{equation}
F(\bz)\overline{G(\bar{\bz}^{-1})}\bz^{-\bN} = \overline{F(\bar{\bz}^{-1})}G(\bz)\bz^{-\bN}.
\label{temeq0}
\end{equation}

Since $F(\bz)$ does not  have trivial factors or nontrivial conjugate symmetric factors, we have 
\begin{equation}
F(\bz) = a \prod_k F_k(\bz),
\label{temeq2}
\end{equation}
where $F_k(\bz)$ are  nontrivial irreducible non-conjugate symmetric monic polynomials in $\bz^{-1}$.  Thus 
\begin{equation}
\bz^{-\bN} \overline{F(\bar{\bz}^{-1})} = a' \bz^{-m'}\prod_k \tilde{F}_k(\bz),
\label{temeq3}
\end{equation}
where $\tilde{F}_k(\bz)$ are the nontrivial irreducible non-conjugate symmetric monic polynomials in $\bz^{-1}$ of the form $\tilde{F}_k(\bz) = \bz^{-\bN+\mathbf{p}_k}\overline{F_k(\bar{\bz}^{-1})}$ for some vector $\mathbf{p}_k$ of positive integers. 

Writing 
\begin{equation}
G(\bz) = b \bz^{-n} \prod_{\ell} G_{\ell}(\bz),
\label{temeq4}
\end{equation}
where $G_\ell(\bz)$ are nontrivial irreducible monic polynomials in $\bz^{-1}$, we have
\begin{equation}
\bz^{-\bN} \overline{G(\bar{\bz}^{-1})} = b' \bz^{-n'} \prod_\ell \tilde{G}_\ell(\bz),
\label{temeq5}
\end{equation}
where $\tilde{G}_\ell(\bz)$ are the nontrivial irreducible monic polynomials in $\bz^{-1}$ of the form  $\tilde{G}_\ell(\bz) = \bz^{-\bN + \mathbf{q}_\ell}\overline{G_\ell(\bar{\bz}^{-1})}$ for some vector $\mathbf{q}_\ell$ of positive integers.

Plugging \eqref{temeq2},\eqref{temeq3}, \eqref{temeq4} and \eqref{temeq5} in \eqref{temeq0} yields
\begin{equation}
a b' \bz^{-n'} \prod_k F_k(\bz) \prod_\ell \tilde{G}_\ell (\bz) = a' b \bz^{-m'-n} \prod_k \tilde{F}_k(\bz) \prod_\ell G_\ell(\bz).
\end{equation}
Each nontrivial irreducible factor $F_k(\bz)$ must be equal to some $\tilde{F}_{k'}(\bz)$ or some $G_{\ell '}(\bz)$. However, if $F_k(\bz) = \tilde{F}_{k}(\bz)$, then $F_k(\bz)$ itself is conjugate symmetric. If, on the other hand, $F_k(\bz) = \tilde{F}_{k'}(\bz)$ for some $k' \neq k$, $F_k(\bz)F_{k'}(\bz) = \tilde{F}_{k'}(\bz)F_{k'}(\bz)$ becomes a conjugate symmetric factor. Both cases, however, are excluded by the assumption  that the $z$-transform of $f$ does not  have conjugate symmetric factors. Thus  each $F_k(\bz)$ must be equal to $G_{\ell '}(\bz)$ for some $\ell '$ and $F(\bz)$ so that $G(\bz)$ must be related by 
\begin{equation}
G(\bz)  = Q(\bz) \prod_k F_k(\bz) = \frac{1}{a}Q(\bz) F(\bz).
\label{temeq6}
\end{equation}
However, $G(\bz)$ and $F(\bz)$ are both polynomials in $\bz^{-1}$, and since $F(\bz)$ contains no trivial factors, so $Q(\bz)$ must be a polynomial of $\bz^{-1}$. Furthermore, plugging \eqref{temeq6} in\eqref{temeq1} yields
$$Q(\bz) = \overline{Q(\bar{\bz}^{-1})}.$$
Therefore, $Q(\bz) = \beta$ and the theorem follows by noting that $\beta$ must be positive if $\measuredangle{F(\bom)} = \measuredangle{G(\bom)}$.

\end{proof}

We next show  that  the $z$-transform of $\{\lambda(\bn)f(\bn)\}$ is almost surely irreducible up to a power $\bz^{-1}$ and not conjugate symmetric.  

\begin{lemma}
Let $f\in \mathcal{C}(\cN )$ be a complex-valued array. Let $\{\lambda(\bn)\}$ be independent and continuous random variables on $\mathbb{S}^1$. Then, $\forall \ \mathbf{t} \neq \mathbf{0}$, the $z$-transform of $(\lambda f)_{\mathbf{t}+} $ and $\overline{(\lambda f)_{\mathbf{t}-}}$ is 
almost surely not conjugate symmetric.
\label{NonCS}
\end{lemma}

\begin{proof}
Let 
$$\tilde{f}_{\bt+} (\bn) = \lambda(\bt+\bn)f(\bt+\bn)$$
whose $z$-transform is
\begin{equation}
\tilde{F}_{\bt+}(\bz) = \sum_{\bn} \lambda(\bt+\bn)f(\bt+\bn)\bz^{-\bn}.
\label{temeq21}
\end{equation}
$\tilde{F}_{\bt+} (\bz)$ is conjugate symmetric if 
\begin{equation}
\tilde{F}_{\bt+}(\bz)  = e^{i \theta} \bz^{-\bk} \overline{\tilde{F}_{\bt+}(\bar{\bz}^{-1})}
\label{temeq22}
\end{equation}
for some vector $\bk$ of positive integers and some $\theta \in [0,2\pi)$. Plugging \eqref{temeq21} in \eqref{temeq22} yields
$$\sum_{\bn} \lambda(\bt+\bn)f(\bt+\bn)\bz^{-\bn} = e^{i\theta}\bz^{-\bk} \sum_{\bn'}\overline{\lambda(\bt+\bn')f(\bt+\bn')}\bz^{\bn'},$$
which implies
\begin{equation}
\lambda(\bt+\bn)f(\bt+\bn) = e^{i\theta} \overline{\lambda(\bt+\bk-\bn)f(\bt+\bk-\bn)}, \ \forall \bn.
\label{temeq23}
\end{equation}
However, $f$ is deterministic, and $\lambda(\bn)$ are independent and continuous random variables, so \eqref{temeq23} fails with probability one for any $\bk$. There are finitely many choices of $\bk$, so the $z$-transform of $(\lambda f)_{\mathbf{t}+}$ is almost surely not conjugate symmetric.

Similarly, the $z$-transform of $\overline{(\lambda f)_{\mathbf{t}-}}$ is also almost surely {\em not} conjugate symmetric.
\end{proof}

\begin{lemma}
Let $f\in \mathcal{C}(\cN )$ be a complex-valued array of rank $\ge 2$. Let $\{\lambda(\bn) \}$ be independent and continuous random variables on $\mathbb{S}^1$. Then, the $z$-transform of $\{\lambda(\bn)f(\bn)\}$ is irreducible up to a power of $\bz^{-1}$ with probability one.
\label{NonIR}
\end{lemma}
For the proof of Lemma \ref{NonIR} see Theorem 2 of \cite{UniqueRI}. 

\begin{lemma}
Let $f$ and $ h$ be two complex-valued arrays. Let $\Phi$ be the discrete Fourier  operator such that $\Phi f(\bom) = \sum_{\mathbf{k}} e^{-2\pi i \bom \cdot \mathbf{k}} f(\mathbf{k})$. Then $\measuredangle{\Phi f_{\mathbf{t}+}} = \measuredangle{\Phi h}$ implies that $\measuredangle{\Phi f} = \measuredangle{\Phi h_{(-\mathbf{t})+}}$.
\label{angle}
\end{lemma}

\begin{proof}
Note that \[
 \Phi f_{\mathbf{t}+}  (\bom) = e^{2\pi i \mathbf{t} \cdot \bom} \Phi f(\bom)
 \]
 which implies
 \[
 2\pi \mathbf{t} \cdot \bom + \measuredangle{\Phi f (\bom)}  \ (\text{mod } 2\pi) = \measuredangle{\Phi h(\bom)}
 \]
 by the assumption $\measuredangle{\Phi f_{\mathbf{t}+}} = \measuredangle{\Phi h}$. Thus
\[ \measuredangle{\Phi f(\bom)} = \measuredangle{\Phi h(\bom)} - 2\pi \mathbf{t}\cdot \bom \ (\text{mod } 2\pi)
\]
which is equivalent to
\[
\measuredangle{\Phi f} = \measuredangle{\Phi h_{(-\mathbf{t})+}}.
\]
\end{proof}

Let us now turn to the proof of Theorem \ref{FixedPoint1}.

 \begin{proof}
 
  Let $f$ be the true image and $h$ a fixed point of the ER iteration. Suppose that $h' =  \PF^{\theta} h$ satisfies the zero-padding condition. Then the following three equations hold:
\beq
 \PO h ' &= &h \\
 |\Phi \Lambda h'| &=& |\Phi \Lambda f| \\
 \measuredangle{\Phi \Lambda h'} &= &\measuredangle{\Phi \Lambda h}
 \label{temeq11}
 \eeq
According to Lemma \ref{NonIR}, the $z$-transform of $\Lambda f$ is irreducible up to a power of $\bz^{-1}$ with probability one, so there exists some integer-valued vector $\mathbf{m}$ with $-\bN \le \mathbf{m} \le \mathbf{0}$ and some $\nu \in [0,2\pi)$ such that
$$h' = e^{i \nu} \Lambda^{-1} \Lambda_{\mathbf{m}+} f_{\mathbf{m}+}$$
or 
$$h' = e^{i \nu} \Lambda^{-1} \overline{\Lambda_{\mathbf{m}-} f_{\mathbf{m}-}}.$$

In the case of $h' = e^{i \nu} \Lambda^{-1} \Lambda_{\mathbf{m}+} f_{\mathbf{m}+}$,
the third equation in \eqref{temeq11} becomes
$$\measuredangle{e^{i\nu}\Phi \Lambda_{\mathbf{m}+} f_{\mathbf{m}+}} = \measuredangle{\Phi \Lambda h}.$$

By Lemma \ref{angle},
\begin{equation}
\measuredangle{e^{i\nu} \Phi \Lambda f} = \measuredangle{\Phi \Lambda_{(-\mathbf{m})+}h_{(-\mathbf{m})+}}.
\label{temeq111}
\end{equation}
Lemma \ref{NonCS} and \ref{NonIR}, together with the assumption that $f(\mathbf{0}) \neq 0$ imply that the $\bz$-transform of $\Lambda f$ is an irreducible,  nontrivial and non-conjugate symmetric polynomial of $\bz^{-1}$ with probability one. 

Next, we apply Proposition \ref{HayesPhase} to \eqref{temeq111}. Both $\Lambda f$ and $\Lambda_{(-\mathbf{m})+}h_{(-\mathbf{m})+}$ are supported on a subset of $\{\bn \ | \ -\bN \le \bn \le \bN \}$. By Proposition \ref{HayesPhase}, we obtain
$$\gamma e^{i\nu} \Lambda f = \Lambda_{(-\mathbf{m})+}h_{(-\mathbf{m}-)+}$$
or equivalently 
$$h(\bn) = \gamma e^{i\nu} \frac{\lambda(\bn+\mathbf{m})}{\lambda(\bn)} f(\bn+\mathbf{m})$$
for some positive number $\beta$.
\commentout{
 When applying Proposition   \ref{HayesPhase} to \eqref{temeq111}, we need oversampling rate $\sigma \ge 4^d$. However, thanks to the Shannon Sampling Theorem, Fourier magnitude measurement with oversampling rate $\sigma \ge 2^d$ renders the correlation function \eqref{autof} and furthermore all the Fourier magnitude measurement. 
 }

\begin{description}
\item[\rm (a)] If the true image $f(\cdot)$ is real-valued, then  $h = \PO h'$  is real-valued, which by the proof of Theorem \ref{UniqueReal} (see Corrolary 1 of \cite{UniqueRI}) implies that 
 $\nu = 0 , \pi$ and $\mathbf{m} = \mathbf{0}$ 
 or equivalently
\beq
h = \pm \gamma f
\label{temeq112}
\eeq
with probability one. 
Plugging \eqref{temeq112} in $\PO \PF h = h$ yields
$\gamma =  1$
and thus  $h = \pm f$ with probability one. 

\item[\rm (b)] If $f$ satisfies the sector condition of Theorem \ref{UniqueComplexPositive}, then   $h = \PO h'$ satisfies the same sector condition which by the proof of Theorem \ref{UniqueComplexPositive} (see Theorem 4 (i) of \cite{UniqueRI})
implies that $\mathbf{m} = \mathbf{0}$  and 
\begin{equation}
h = \gamma e^{i \nu}f
\label{temeq113}
\end{equation}
with probability at least $1- |\mathcal{N}| (\beta-\alpha)^{\llfloor S/2 \rrfloor} (2\pi)^{-\llfloor S/2\rrfloor}$. 
Plugging \eqref{temeq113} in $\PO \PF h = h$ yields $\gamma = 1$
and thus $h = e^{i \nu} f$.

\end{description}

By the similar argument one reaches the same conclusion in the case of $h'=e^{i \nu}\Lambda^{-1}\Lambda_{\mathbf{m}-}f_{\mathbf{m}-}$.

 \end{proof}

\bibliographystyle{unsrt}

\end{document}